\documentclass[journal]{IEEEtran}

\usepackage[utf8]{inputenc}

\usepackage{amsmath, amssymb}
\usepackage{amsfonts}

\usepackage{makecell}
\usepackage{multirow}
\usepackage{threeparttable}
\usepackage{balance}
\usepackage{flushend}
\usepackage{changepage}
\usepackage{tikz}
\usepackage{diagbox}
\usepackage{array}
\usepackage{mathrsfs}
\usepackage{subfigure}
\usepackage{algorithm}  
\usepackage{algorithmicx}  
\usepackage{algpseudocode}  
\usepackage{multirow} 
\usepackage{xcolor}

\usepackage{graphicx}
\usepackage{color}
\usepackage{url}
\usepackage{bm}
\usepackage[all]{xy}
\usepackage{arydshln}
\usepackage{tikz}
\usepackage{float}

\usepackage{enumerate}

\newtheorem{theorem}{Theorem}
\newtheorem{lemma}{Lemma}
\newtheorem{assumption}{Assumption}
\newtheorem{definition}{Definition}
\floatname{algorithm}{Algorithm}

\newenvironment{proof}{{\it Proof}\quad}{\hfill $\square$\par}

\usepackage{hyperref}

\usepackage{pifont}

\ifCLASSINFOpdf
\else
\fi
\begin{document}

\title{Dap-FL: Federated Learning flourishes by adaptive tuning and secure aggregation}

\author{Qian~Chen\textsuperscript{*},
	Zilong~Wang,~\IEEEmembership{Member,~IEEE},
	Jiawei~Chen\textsuperscript{*},
	Haonan~yan,
	and~Xiaodong~Lin,~\IEEEmembership{Fellow,~IEEE}
	
	\thanks{Qian~Chen, Zilong~Wang, Jiawei~Chen, and haonan~Yan are with the State Key Laboratory of Integrated Service Networks, School of Cyber Engineering, Xidian University, Xi'an, China (e-mail: qchen\_4@stu.xidian.edu.cn; zlwang@xidian.edu.cn; xidianjiaweichen@gmail.com; yanhaonan.sec@gmail.com).}
	\thanks{Xiaodong~Lin is with the School of Computer Science, University of Guelph, Guelph, Canada (e-mail: xlin08@uoguelph.ca).}
	\thanks{\textsuperscript{*}These authors contributed equally to this work.}
}

\maketitle
	
\begin{abstract}

Federated learning (FL), an attractive and promising distributed machine learning paradigm, has sparked extensive interest in exploiting tremendous data stored on ubiquitous mobile devices. However, conventional FL suffers severely from resource heterogeneity, as clients with weak computational and communication capability may be unable to complete local training using the same local training hyper-parameters. In this paper, we propose Dap-FL, a deep deterministic policy gradient (DDPG)-assisted adaptive FL system, in which local learning rates and local training epochs are adaptively adjusted by all resource-heterogeneous clients through locally deployed DDPG-assisted adaptive hyper-parameter selection schemes. Particularly, the rationality of the proposed hyper-parameter selection scheme is confirmed through rigorous mathematical proof. Besides, due to the thoughtlessness of security consideration of adaptive FL systems in previous studies, we introduce the Paillier cryptosystem to aggregate local models in a secure and privacy-preserving manner. Rigorous analyses show that the proposed Dap-FL system could guarantee the security of clients' private local models against chosen-plaintext attacks and chosen-message attacks in a widely used honest-but-curious participants and active adversaries security model. In addition, through ingenious and extensive experiments, the proposed Dap-FL achieves higher global model prediction accuracy and faster convergence rates than conventional FL, and the comprehensiveness of the adjusted local training hyper-parameters is validated. More importantly, experimental results also show that the proposed Dap-FL achieves higher model prediction accuracy than two state-of-the-art RL-assisted FL methods, i.e., 6.03\% higher than DDPG-based FL and 7.85\% higher than DQN-based FL.

\end{abstract}
	

%
\IEEEpeerreviewmaketitle

\section{Introduction}
\IEEEPARstart{A}{s} estimated by Cisco, nearly 850 zettabytes of data burst out from all people, machines, and things by 2021, up from 220 zettabytes generated in 2016 \cite{Cisco2021}. Coupled with the rise of Machine Learning (ML) \cite{Jordan2015} and Deep learning (DL) \cite{LeCun2015}, these valuable data could unfold countless opportunities for modern society, e.g., Internet of Things (IoT) \cite{Xu2014}, Internet of Vehicle (IoV) \cite{Kaiwartya2016}, and Healthcare \cite{Beam2017}.

However, practical scenarios are somewhat disappointing: utilizing such tremendous data for ML is more difficult than we thought. On the one hand, collecting data for ML model training would eventually saturate WAN bandwidths, while the WAN bandwidth is a scarce resource, whose growth has been decelerating for many years \cite{Vulimiri2015}. On the other hand, data security and privacy are strengthened by states across the world, e.g., General Data Protection Regulation (GDPR) \cite{GDPR2017}, which significantly increases the difficulty of data collection. 

Facing the above two enormous difficulties, Federated Learning (FL) \cite{Konecny2016,McMahan2017,Bonawitz2017} has emerged as an attractive and promising paradigm, which is in stark contrast to traditional ML with a data center. In the typical FL, clients collaboratively train a global ML model under the orchestration of a central server without collecting local data in a data center. Specifically, clients locally train individual models using their data in parallel, and the central server aggregates local contributions to update the global model subsequently. Such a process is executed periodically until the global model converges. Therefore, FL is a direct application of the {\it data minimization} principle \cite{White2013}, as it decouples model training from the need for direct access to the raw data, which alleviates the difficulties of communication overhead and privacy. With the primary advantages of communication efficiency and privacy preservation, FL has been deployed in a wide range of applications, including the keyboard application for smartphones \cite{Bonawitz2019}, Healthcare \cite{Brisimi2018}, and Industrial Internet of Things (IIoT) \cite{Sun2021,Zhang2021}.

Despite the promising benefits, FL is still caught in a dilemma between model prediction performance and system efficiency in practice. Such a dilemma is mainly caused by {\it resource heterogeneity} \cite{Chai2020}, where clients have diverse local resources for computation and communication. When enforcing these resource-heterogeneous clients to train local models using the same hyper-parameters, some clients cannot complete the local training due to insufficient local resources, thereby leading to a serious {\it straggler problem} \cite{Dean2012}. Dropping straggler clients not only reduces the global model convergence rate but also reduces the global model prediction performance, although the system efficiency is guaranteed in some ways \cite{Imteaj2022}. On the contrary, waiting for straggler clients blindly significantly wastes lots of time, although the global model prediction accuracy may increase due to the extra local contributions. To meet the requirements of the model prediction performance and the system efficiency simultaneously, the FL system should carefully manipulate the local training, which is referred to as {\it adaptive or self-tuning algorithms} \cite{Andrew2021,Bonawitz2019b}.

The adaptive hyper-parameter adjustment has a long history \cite{Kohavi1995}, but it mainly focuses on model prediction accuracy \cite{Bergstra2011,Snoek2015} rather than communication and computational efficacy for heterogeneous clients or the system. In FL, more potential hyper-parameters are adaptively adjusted, e.g., participating clients or the amount of participating clients \cite{Luo2021,Shi2020,Wang2020,Zhang2021,Lu2020}, local training iterations \cite{Luo2021}, frequency of aggregation \cite{Wang2019,Sun2021}, and allocated resources \cite{Shi2020,Tran2019,Nguyen2020}. In a nutshell, current adaptive FL methods usually take one of two routes (the details are shown in Section VII): (1) theoretic methods \cite{Luo2021}-\cite{Tran2019}, and (2) Reinforcement Leaning (RL)-based methods \cite{Wang2020}-\cite{Lu2020}. Both the two routes formulate the hyper-parameter adjusting process in FL as a constrained and non-convex optimization problem, which is hard to solve directly. The former route converts the prime optimization problem into a convex optimization problem, which can be solved easier, by making some prior hypotheses to practical environment constraints. However, such transformations with prior hypotheses may be less reasonable, as real-world environments, e.g., wireless communication ranges and battery power, are time-varying and have no statistical regularity to follow. Recently, the RL-based method is proposed as a slightly more sophisticated but more effective and practical method to achieve adaptive FL, which formulates the hyper-parameter adjusting process as a Markov decision process (MDP) and searches for appropriate hyper-parameters by introducing the Deep-Q Network (DQN) \cite{Wang2020}-\cite{Nguyen2020} or the deep deterministic policy gradient (DDPG) \cite{Zhang2021,Lu2020} algorithm on the server side. Among them, the aggregation frequency adjustment in Sun {\em et al.}'s DQN-based method \cite{Sun2021} and the client selection in Zhang {et al.}'s DDPG-based method \cite{Zhang2021} are the state-of-the-art RL-based adaptive FL methods. 

However, even for \cite{Sun2021,Zhang2021}, existing RL-based adaptive FL methods have five main shortages: (1) DQN is less effective than DDPG when the state and action spaces are continuous and high-dimensional. (2) Selecting clients is equivalent to abandoning straggler clients, which may result in low global model prediction accuracy. (3) Selecting hyper-parameters for clients on the server side is less effective since the server cannot capture clients' fast-changing training conditions. (4) Fixing the time-varying constraints as a constant is unreasonable. (5) Ignoring security requirements is not conducive to the development of adaptive FL. Therefore, an intuitive question is: {\em How to achieve adaptive FL better with the satisfaction of global model performance, efficiency, and security requirements?} 

In response, we propose a \underline{D}DPG-assisted \underline{a}daptive and \underline{p}rivacy-preserving federated learning (Dap-FL) system. Specifically, we introduce the DDPG algorithm to help clients adaptively adjust local learning rates and local training epochs for the purpose that all participating clients with heterogeneous resources could collaboratively train a global model efficiently, rather than abandoning straggler clients. In addition, Dap-FL adopts the Paillier cryptosystem for privacy-preserving and secure local model aggregation. 

We stress the superiority of the proposed Dap-FL compared to the RL-based FL designs, especially the state-of-the-art \cite{Sun2021,Zhang2021}, in the following five aspects: (1) Different from selecting clients' local training hyper-parameters by a server, Dap-FL adaptively adjusts clients' local training hyper-parameters through DDPG models maintained by the clients themselves, which suits clients' fast-changing training conditions more. (2) Rather than abandoning straggler clients with low local resources, Dap-FL aggregates all local contributions, which helps FL out of the model prediction accuracy and system efficiency dilemma. (3) Instead of leveraging the global loss function, Dap-FL formulate clients' local contributions as the reward function, which is validated more effective in our experiments. (4) Dap-FL considers time-varying constraint conditions with respect to clients' local resources, which is more reasonable and practical than modeling the constraints as a constant. (5) Compared to all RL-based adaptive FL designs without any security concerns, Dap-FL aggregates local models in a secure manner, which makes the proposed Dap-FL competitive. To the best of our knowledge, Dap-FL is the first adaptive FL system by locally adjusting hyper-parameters through the DDPG algorithm, and is the first secure and privacy-preserving adaptive FL system. 

The main contributions of this paper are as follows:

\textbf{DDPG-assisted adaptive federated learning system.} We propose a DDPG-assisted adaptive hyper-parameter selection scheme, which formulates every client's local training process as a constrained MDP, and hence solves its Lagrangian dual problem by the DDPG algorithm, where the rationality of the Lagrangian-based problem transformation is confirmed through rigorous mathematical proof. With the deployment of the proposed scheme, we exploit an adaptive FL system, Dap-FL, where the global model prediction accuracy and the global model convergence rate are enhanced.

\textbf{Secure and privacy-preserving federated learning system.} To satisfy general secure requirements in FL, Dap-FL introduces the Paillier cryptosystem to aggregate local models in a secure and private manner. Through rigorous analyses, the proposed Dap-FL system could guarantee the privacy, source authentication, and data integrity of local models in a widely used honest-but-curious participants and active adversaries security model \cite{Bonawitz2017,Bonawitz2019,Hao2019}. In particular, clients' private local models are semantic secure against chosen-plaintext attacks \cite{Goldreich2004}, and are secure against chosen-message attacks \cite{Goldreich2004} in the random oracle model \cite{Bellare1993}.

\textbf{Extensive experimental evaluation.} We evaluate the performance of the proposed Dap-FL by deploying three different ML models on two datasets. Experimental results show that the proposed Dap-FL achieves higher model prediction accuracy and faster convergence rate, which could be regarded as achieving higher communication efficiency, compared to conventional FL. More importantly, through the comparison with two state-of-the-art RL-assisted FL methods, the prediction accuracy of the global model trained by Dap-FL is $6.03\%$ higher than the DDPG method \cite{Zhang2021} and $7.85\%$ ($2.65\%$ with a small learning rate) higher than the DQN-based method \cite{Sun2021} in our experimental settings.

The rest of the paper is organized as follows. Section II formalizes the system model, security requirements, and design goal. Section III introduces some primitive concepts. In Section IV, we formulate the local training process as a constrained MDP and solve its Lagrangian dual problem by the DDPG algorithm. In Section V, we propose the Dap-FL system, followed by the security analysis and experimental evaluation in Section VI and Section VII. Section VIII summarizes related works on adaptive FL. Finally, we draw our conclusions in Section IX.

\section{System Model, Security Requirement and Design Goal}
In this section, we first formalize the FL system under our consideration. Then, we analyze security requirements and identify our design goal.

\begin{figure*}[htbp]
	\centering
	\includegraphics[width=18cm,height=9cm,trim=40 220 50 130,clip]{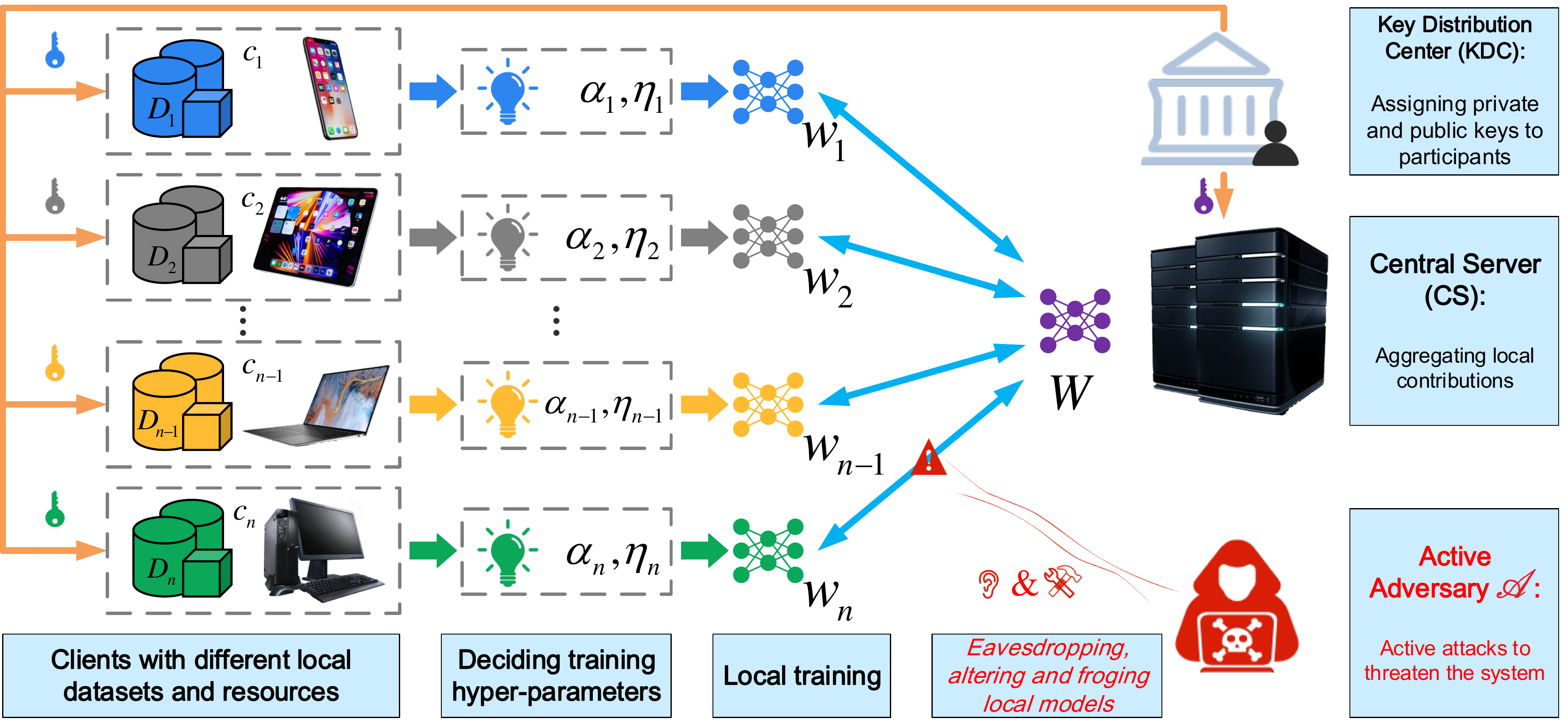}
	\caption{Centralized Federated Learning System.}
	\label{fig-1}
\end{figure*}
	
\subsection{Federated Learning System}
As shown in Fig. \ref{fig-1}., the FL system under our consideration consists of three participants: a key distribution center ($\mathit{KDC}$), a central server ($\mathit{CS}$), and a group of $ N $ clients $ \left\{c_{i},i=1,\cdots,N \right\} $. The main job of the $\mathit{KDC}$ is generating and assigning public and private keys to each participant. Each client $ c_{i} $ has a local dataset $ D_{i} $ with the size of $ \vert D_{i}\vert $ containing sensitive information. The goal of all clients is to collaboratively train a global model under the orchestration of the $\mathit{CS}$ in a privacy-preserving manner. Formally, the training process can be seen as solving an optimization problem mathematically, defined as:
\begin{equation}\label{eq-1}
	W^{*}=\mathop{\arg\min}_{W} F\left(W; D_{1},\cdots,D_{N}\right),
\end{equation}
where $ F(\cdot) $ is the global loss function, $ W $ represents the global model, and $ W^{*} $ is the final converged global model. Note that although the global model (as well as the local ML model below) is parameterized by a vector or a matrix in practice, we regard it as a scalar in this paper for convenience.
	
In general, such an optimization problem could be solved by periodically performing the distributed gradient descent algorithm. For clarity, we take the $ t $-th period (also named as training round in this paper) as an example to describe the problem-solving process, where $ t=0,1,\cdots,T $. $ \mathit{CS} $ first distributes the global model $ W(t) $ to all clients. After receiving the global model, each $ c_{i} $ initializes its local model $ w_{i}(t) $ as $ W(t) $ and its loss function $ f_{i}(\cdot) $ as $ F(\cdot) $. Then, $ c_{i} $ trains $ w_{i}(t) $ on $ D_{i} $ locally for $ \alpha_{i}(t) $ epochs to update its local model, which can be simply expressed as:
\begin{equation}\label{eq-2}	
	\hat{w}_{i}(t)=w_{i}(t)-{{\eta}_{i}(t)}\cdot \nabla f_{i}\left({w_{i}(t)};D_{i};\alpha_{i}(t)\right),
\end{equation}
where $ \hat{w}_{i}(t) $ is the updated local model, $ \nabla f_{i}\left(\cdot\right) $ is the gradient, and $ \eta_{i}(t) $ is the learning rate. Subsequently, each client $ c_{i} $ uploads $ \hat{w}_{i}(t) $ and $ \vert D_{i}\vert $ to the $ \mathit{CS} $ in parallel, and the $ \mathit{CS} $ computes the weighted mean of $ \hat{w}_{i}(t) $s as the updated global model, denoted by $ W(t+1) $, and defined as:
\begin{equation}\label{eq-3}
	W(t+1)=\frac{1}{|D|} \sum\limits_{i=1}^{N}{{|{D}_{i}|}{\hat{w}_{i}(t)}},
\end{equation}
where $|D|=\sum_{i=1}^{N}|{D}_{i}|$.

\textit{\textbf{Time-varying resource consumption model}.} Different from the conventional FL without considering clients' local resources, we take the communication and computational resources of each client into account, which is time-varying and has no statistical regularity to follow. Thus, the local resource of $ c_{i} $ in the $ t $-th training round is donated by $ E_{i}(t) $. Particularly, in the $ t $-th training round, the resources that $ c_{i} $ consumes for one epoch of local training and single time of communication with the $ \mathit{CS} $ are denoted by $ E_{i}^{cpt}(t) $ and $ E_{i}^{cmu}(t) $, which limit the selection of local training hyper-parameters.   

\textit{\textbf{Communication model}.} In the FL system under our consideration, the communication between a client $ c_{i} $ and the $ \mathit{CS} $ is through the relatively inexpensive WiFi technology. In other words, each client can directly communicate with the $ \mathit{CS} $ within the WiFi coverage range. However, the communication within the WiFi coverage range is not private, which means all participants could be eavesdropped by others. Besides, clients might also communicate with others directly through WiFi technology. On the other hand, the $ \mathit{KDC} $ assigns keys for secure communication through expensive but secure communication channels. That is, the keys assigned by the $ \mathit{KDC} $ cannot be eavesdropped, but such secure communication channels can only be used for key assignment instead of the local model aggregation due to the scarce bandwidth \cite{Hao2019}.

\subsection{Security Requirements}
Security and privacy draw wide attention since FL is proposed. In particular, clients' private local model is vulnerable to model inversion attacks \cite{Fredrikson2015} by any participants with access to a local model, which leads to severe privacy leakage. Additionally, active adversaries could execute active attacks, e.g., masquerade attacks \cite{Schonlau2001} and replay attacks \cite{Goldreich2004}, to threaten the source authentication and data integrity of the FL system.  

Therefore, in this paper, we consider an honest-but-curious participants and active external adversary threat model, which is widespread in FL \cite{Bonawitz2017,Bonawitz2019,Hao2019}. Specifically, the $ \mathit{KDC} $ is fully trustworthy and hence would never collude with others, as the role of a $ \mathit{KDC} $ is always played by an official institution with credibility in practice. In addition, the $ \mathit{CS} $ and the clients are considered honest-but-curious, which means all participants (except the $ \mathit{KDC} $) would honestly train the global model but may infer sensitive information contained in clients' local training data by executing model inversion attacks. Besides, we further consider an external adversary $ \mathscr{A} $ who can eavesdrop the individual local models. More seriously, $ \mathscr{A} $ could execute active attacks to break down the FL system, such as masquerading attacks, where $ \mathscr{A} $ masquerades as a legitimate client and uploads an incorrectly formatted local model. In order to prevent $ \mathscr{A} $ from breaking down the FL system and guarantee the privacy of an individual local model, the following security requirements should be achieved in our FL system.

\begin{itemize}
	\item{{\it Local model confidentiality.} The content of an individual local model should not be obtained by anyone but the corresponding client itself. Confidentiality also requires the FL system could resist eavesdropping attacks by the adversary $ \mathscr{A} $.}

	\item{{\it Source Authentication and Data Integrity.} To prevent the existing active attacks from $ \mathscr{A} $, each participant is required to confirm the received contents that are sent by a legitimate participant and have not been altered and/or forged during the transmission.}
	
\end{itemize}

\subsection{Design Goal}
Under the aforementioned system model and security requirements, our design goal mainly focuses on two aspects: local training adaptiveness and secure aggregation. Specifically, the following two objectives should be achieved.
	
\textit{The local training adaptiveness should be guaranteed in our FL system.} As mentioned in the system model, clients expect to adaptively select the local learning rates and local training epochs within the constraints of local resources for higher global model prediction accuracy and communication efficiency. Therefore, our FL system should achieve local training hyper-parameter adaptive adjustment.
	
\textit{The security requirements should be guaranteed in our FL system.} Privacy preservation and security are of crucial importance for the flourishing of FL. Thus, our FL system should satisfy the privacy of individual local models and the security of the whole system.

\section{Preliminaries}	
In this section, we provide a brief overview of the reinforcement learning (RL) \cite{Kaelbling1996} and the deep deterministic policy gradient (DDPG) algorithm \cite{Lillicrap2015} which serves as the basis of the proposed adaptive local training hyper-parameter selection scheme. Besides, we introduce the Paillier cryptosystem for the design of the privacy-preserving FL.

\subsection{Deep Reinforcement Learning}	
		
RL is proposed to help a client make decisions sequentially through trial-and-error interactions with environment. Formally, such a sequential decision-making process is modeled as a Markov decision process (MDP) \cite{Bellman1957}, denoted by a $ 3 $-tuple $ \left\{\mathcal{S}, \mathcal{A}, R \right\} $, where $ \mathcal{S} $ is the state space, $ \mathcal{A} $ is the action space, and $ R $ is the reward function. In detail, the client selects an action $ a(t) \in \mathcal{A} $ according to its state $ s(t) \in \mathcal{S} $ at time step $ t $, which can be expressed as:
\begin{equation}\label{eq-4}
	{\mu}:s(t)\rightarrow a(t),
\end{equation}
where the action-producing function $ \mu $ represents a deterministic policy in practice. After performing the selected action $ a(t) $, the client transits to the next state $ s(t+1) $ and receives an immediate reward $ r(t)=R\left(s(t), a(t), s(t+1)\right)$ from the environment. Thus, the state-action value function at time step $ t $, depicting the expected long-term reward from time $ t $, can be naturally denoted by $ Q^{\mu} $ and defined as:
\begin{equation} \label{eq-5}
	Q^{\mu}\left(s(t),a(t)\right)=\mathbb{E}\left[{\sum\limits_{\tau=t}^{T}{\gamma}^{\tau-1}\cdot r(\tau)}\right],
\end{equation}
where $ \tau $ is the index of the time step, and $ \gamma\in(0,1) $ is the discount factor. Therefore, the purpose of the client can be seen as maximizing the state-action value function with respect to the policy $ \mu $, which is simply expressed as:
\begin{equation}\label{eq-6}
{\mu}^{*}=\underset{\mu}{\mathop{\arg\max}}\,{Q}^{\mu}\left(\mathcal{S},\mathcal{A}\right),
\end{equation}
where $ \mu^{*} $ is the best policy. 
	
\noindent \textit{Deep Deterministic Policy Gradient Algorithm.}

To search for the best policy of the above optimization problem, deep reinforcement learning (DRL) \cite{Arulkumaran2017} introduces the deep neural network to approximate the state-action value function. In particular, the DDPG algorithm serves as an efficient DRL algorithm facing a high-dimensional and continuous state and/or action space. Specifically, the DDPG algorithm sets up two {\it Actor-Critic} networks, namely MainNet and TargetNet. In the MainNet, the deterministic policy and the state-action value function are parameterized by $ \theta^{\mu}(t) $ and $ \theta^{Q}(t) $. And the TargetNet is a copy of the MainNet with parameters ${\theta}^{\mu'}(t)$ and ${\theta}^{Q'}(t)$. Through iteratively training the two nets, the best policy $ \mu^{*} $ could be obtained.

\subsection{Paillier Cryptosystem}
The Paillier cryptosystem \cite{Paillier1999} consists of the encryption scheme and the digital signature scheme. We first recall the definition of $ n $-th residues modulo $ n^{2} $ and the Decisional Composite Residuosity (DCR) assumption as follows.
	
\begin{definition}[$ n $-th residues modulo $ n^{2} $ \cite{Paillier1999}]\label{def-1}
A number $ x $ is said to be a $n$-th residue modulo ${n}^{2}$, if there exists a number $z\in{\mathbb{Z}}_{{n}^{2}}^{*}$ such that
\begin{equation*}
	x = {z}^{n} \bmod {n}^{2},
\end{equation*}
where $n=p \cdot q$ is the product of two large primes.
\end{definition}

\begin{assumption}[Decisional Composite Residuosity (DCR) Assumption \cite{Paillier1999}]\label{asp-1}
There exists no polynomial time distinguisher for $n$-th residues modulo ${n}^{2}$.
\end{assumption}

Thus, a public-key encryption scheme and a digital signature scheme work as follows.

\noindent \textbf{Paillier Encryption scheme} 

\textit{Key Generation Algorithm.}
The key generation algorithm first picks two large primes $ p^{\mathit{Enc}} $ and $ q^{\mathit{Enc}} $ randomly such that $gcd\left(p^{\mathit{Enc}}\cdot q^{\mathit{Enc}}, \left({p}^{\mathit{Enc}}-1 \right) \cdot \left({q}^{\mathit{Enc}}-1 \right)  \right) =1 $. Next, it computes $n={p}^{\mathit{Enc}}\cdot {q}^{\mathit{Enc}}$ and $\varrho = lcm\left(p^{\mathit{Enc}}-1,q^{\mathit{Enc}}-1 \right)$. After picking a random $ g $ from $ \mathbb{Z}_{n^{2}}^{*} $, the algorithm computes $ \delta=\left(L\left(g^{\varrho}\bmod{n^{2}} \right) \right)^{-1}\bmod{n} $, where $ L\left(x \right)=\frac{x-1}{n} $. Thus, the public key is $ \mathit{PK}^{\mathit{Enc}}=\left(n, g \right) $, and the private key is $ \mathit{SK}^{\mathit{Enc}}=\left(\varrho,\delta \right) $. We simply express the key generation algorithm as:
\begin{equation}\label{eq-7}
     \left({\mathit{PK}}^{\mathit{Enc}},{{\mathit{SK}}}^{\mathit{Enc}} \right) \leftarrow \textit{KeyGen}\left(p^{\mathit{Enc}},q^{\mathit{Enc}} \right).
\end{equation}		

\textit{Encryption Algorithm.}
Given a plaintext $ M \in {{\mathbb{Z}}_{n}} $ and the public key $ \mathit{PK}^{\mathit{Enc}}$, the encryption algorithm selects a random $ \zeta\in \mathbb{Z}_{{n}}^{*} $, and further computes the ciphertext $\varepsilon \left(M \right)={{g}^{M}}\cdot {\zeta^{n}}\bmod {{n}^{2}}$.
The encryption algorithm is simply shown as:
\begin{equation}\label{eq-8}
	\varepsilon \left(M \right)	\leftarrow \textit{Enc} \left(M,{\mathit{PK}}^{\mathit{Enc}} \right).
\end{equation}	

\textit{Decryption Algorithm.}
Given a ciphertext $\varepsilon \left(M \right) \in \mathbb{Z}_{{{n}^{2}}}^{*}$ and the private key ${{\mathit{SK}}}^{\mathit{Enc}}$, the decryption algorithm computes the plaintext by $M=L\left({{\varepsilon(M)}^{\varrho}}\bmod {{n}^{2}}\right)\cdot \delta \bmod n$, which is expressed as:
\begin{equation}\label{eq-9}
    M \leftarrow \textit{Dec} \left(\varepsilon \left(M \right),{\mathit{SK}}^{\mathit{Enc}} \right).
\end{equation}
	
\textit{Evaluation Algorithm.}
The evaluation algorithm is used to verify the addition homomorphism of the Paillier encryption scheme, i.e.,
\begin{equation}\label{eq-10}
	\varepsilon\!\left(\!{M}_{1}\!\right)\cdot\varepsilon\!\left(\!{M}_{2}\!\right)\!\bmod\!{{n}^{2}}\!\!=\!\textit{Enc}\left((M_{1}\!\!+\!\! M_{2}\!\bmod\! n),\!{\mathit{PK}}^{\mathit{Enc}} \right)\!,
\end{equation}
where ${M}_{1}$ and ${M}_{2}$ are two plaintexts, and $ \varepsilon\left({M}_{1}\right)$ and $\varepsilon\left({M}_{2}\right) $ are corresponding ciphertexts.

\noindent \textbf{Paillier Digital Signature Scheme}

\textit{Key Generation Algorithm.}
The key generation algorithm has the same steps as that in the Paillier encryption scheme, shown as:
\begin{equation}\label{eq-11}
	\left({\mathit{PK}}^{\mathit{Sign}},{{\mathit{SK}}}^{\mathit{Sign}}\right)\leftarrow\textit{KeyGen}\left({p}^{\mathit{Sign}},{q}^{\mathit{Sign}} \right),
\end{equation}
where ${p}^{\mathit{Sign}}$ and ${q}^{\mathit{Sign}}$ are two large primes, and ${\mathit{PK}}^{\mathit{Sign}}$ and ${\mathit{SK}}^{\mathit{Sign}} $ are the public key and private key.
		
\textit{Signing Algorithm.}
Given a hash function $ H:\{0,1\}^{*}\rightarrow{\mathbb{Z}}_{{n}^{2}}^{*} $, a private key $ {\mathit{SK}}^{\mathit{Sign}} $, and a message $ M\in {{\mathbb{Z}}_{n}} $, the signing algorithm first computes $ h\leftarrow H(M) $.
Next, it computes the signature ${\sigma}=\frac{L(h^{\varrho }\bmod {{n}^{2}} )}{L(g^{\varrho }\bmod {{n}^{2}})}\bmod n $ and $ \tilde{\sigma}={{\left(h\cdot {{g}^{-{\sigma}}} \right)}^{\left(\frac{1}{n}\bmod \varrho\right) }}\bmod n $, where $ L\left(x \right)=\frac{x-1}{n} $.
Thus the signing algorithm is simply expressed as:
\begin{equation}\label{eq-12}
	M\vert\vert\left(\sigma,\tilde{\sigma}\right)\leftarrow\textit{Sign}(M, H, {{\mathit{SK}}}^{\mathit{Sign}}),
\end{equation}
where $ \vert\vert $ is the string concatenation operator.

\textit{Verification Algorithm.}
Given a hash function $H$, a public key ${\mathit{PK}}^{\mathit{Sign}}$, and a message $ M $ with a signature $\left(\sigma,\tilde{\sigma}\right)$, the verification algorithm first computes $ h\leftarrow H(M)$.
Next, it checks whether the equation $h = {g}^{\sigma} \cdot {\tilde{\sigma}}^{n}\bmod{n}^{2}$ holds.
Thus, the verification algorithm is expressed as:
\begin{equation}\label{eq-13}
	\{0,1\}\leftarrow\textit{Verify}\left(M\vert\vert \left(\sigma,\tilde{\sigma} \right),H,{\mathit{PK}}^{\mathit{Sign}} \right).
\end{equation}

\section{DDPG-assisted Adaptive Hyper-parameter Selection}
In this section, we analyze hyper-parameters affecting local training, and hence formulate the client's local training process as an optimization problem. Further, we convert the optimization problem as a constrained MDP and obtain the best policy by solving its Lagrangian dual problem with the usage of the DDPG algorithm. The entire problem formalization, transformation, and solving processes constitute the proposed DDPG-assisted adaptive hyper-parameter selection scheme. 
	
\subsection{Problem Formulation}
To achieve the goal of adaptively adjusting local training hyper-parameters during the FL process, we must model it mathematically. In the FL under our consideration, every client has heterogeneous and time-varying resource constraints regarding computation and communication, which means clients may adopt totally different hyper-parameters for local training according to diverse resource constraints, thereby making different contributions to the global model. Therefore, it is essential to model the local training process with resource constraints for each client separately. 

But the first question is how the client's local training hyper-parameters affect the global model convergence. Empirical results show that the convergence of the global model is strongly impacted by local contributions which are dependent on the local training hyper-parameters significantly, i.e., local training epoch \cite{Wang2018} and the local learning rate \cite{Li2020}. In detail, a small local training epoch is more helpful for the precise approximation to the converged global model, but clients would suffer from more aggregations, which decelerates the global model convergence rate and improves clients' communication consumption. Conversely, a large local training epoch helps the global model converge to an approximate result of the converged global model rapidly, but it makes the final convergence of the global model extremely difficult, since the variance of local models increases sharply. Besides, the local learning rate has a similar impact on the global model to the local training epoch. Although the decaying learning rate \cite{Li2020} is utilized to adaptively tune models in the ML community, it may prevent the global model from convergence, as the learning rate becomes too small after too many training rounds. As a result, we plan to adaptively adjust the local training epoch and the local learning rate in each training round.

However, in the $ t $-th training round, the influence of $ c_{i} $'s local contribution on the global model is not intuitive. Naturally, we consider measuring the influence by defining the following three metrics.

\begin{definition}[Training loss value difference]\label{def-2}
	The difference between the training loss values of $ w_{i}(t) $ and $ w_{i}(t-1) $ on $ D_{i} $ is denoted by $ \psi_{1} $ and defined as:
	\begin{equation}\label{eq-14}
		\psi_{1}={\mathit{loss}}_{i}(t-1)-{\mathit{loss}}_{i}(t),
	\end{equation}
	where $ {\mathit{loss}}_{i}(t-1) $ and ${\mathit{loss}}_{i}(t) $ are the training loss values of the initialized local models at $ (t\!-\!1) $-th and $ t $-th training rounds on $ D_{i} $.
\end{definition}

\begin{definition}[Training accuracy difference]\label{def-3}
	The difference between the training accuracy of $ w_{i}(t) $ and $ w_{i}(t-1) $ on $ D_{i} $ is denoted by $ \psi_{2} $ and defined as:
	\begin{equation}\label{eq-15}
		\psi_{2}={\mathit{acc}}_{i}(t)-{\mathit{acc}}_{i}(t-1),
	\end{equation}
	where $ {\mathit{acc}}_{i}(t) $ and $ {\mathit{acc}}_{i}(t-1) $ are the training accuracy of the initialized local models at $ t $-th and $ (t\!-\!1) $-th training rounds on $ D_{i} $.
\end{definition}

\begin{definition}[Training F-1 score difference]\label{def-4}
	The difference between the training F-1 scores of $ w_{i}(t) $ and $ w_{i}(t-1) $ on $ D_{i} $ is denoted by $ \psi_{3} $ and defined as:
	\begin{equation}\label{eq-16}
		\psi_{3}={\mathit{fs}}_{i}(t)-{\mathit{fs}}_{i}(t-1),
	\end{equation}
	where $ {\mathit{fs}}_{i}(t) $ and $ {\mathit{fs}}_{i}(t-1) $ are the training F-1 scores of the initialized local models at $ t $-th and $ (t\!-\!1) $-th training rounds on $ D_{i} $.
\end{definition}

Intuitively, an effective local contribution can reduce the global model loss value and/or increase the global model accuracy and F-1 score. In other words, a local contribution could be measured by the linear combination of $ \psi_{1} $, $ \psi_{2} $, and $ \psi_{3} $, shown as:
\begin{equation}\label{eq-17}
	{\xi}_{1}\cdot\psi_{1}+{\xi}_{2}\cdot\psi_{2}+{\xi}_{3}\cdot\psi_{3},
\end{equation}  
where $ \xi_{1}, \xi_{2}, \xi_{3}\in(0,+\infty) $. Thus, the purpose of $ c_{i} $'s local training at $ t $-th training round can be simply formulated as maximizing the local contribution by selecting proper $ \alpha_{i}(t) $ and $ \eta_{i}(t) $, defined as:
\begin{equation}\label{eq-18}
    \underset{\alpha_{i}(t),\eta_{i}(t)}{\mathop{\max}}\left[{\xi}_{1}\!\cdot\!\psi_{1}+{\xi}_{2}\!\cdot\!\psi_{2}+{\xi}_{3}\!\cdot\!\psi_{3} \right].
\end{equation}

Nevertheless, the above maximization problem is not comprehensive to formulate the local training process at $ t $-th training round, since the client might tend to select an oversize $ \alpha_{i}(t) $ to maximize the local contribution, while the local resource $ E_{i}(t) $ limits the selection of local training epochs. Therefore, to prevent $ c_{i} $ from selecting local training hyper-parameters exceeding $ E_{i}(t) $, we restrain the size of $ \alpha_{i}(t) $ in the $ t $-th training round by a constraint condition, defined as:
\begin{equation}\label{eq-19}
	{\alpha}_{i}(t)\cdot{{E}_{i}^\mathit{cpt}(t)}+{2}\cdot E_{i}^{\mathit{cmu}}(t)\le {{E}_{i}(t)}.
\end{equation}

As a result, the selection of $ c_{i} $'s local training hyper-parameters during the whole FL process can be formulated as a constrained optimization problem, denoted by $ \textbf{P}_{\bm 0} $ and defined as:
\begin{equation*}
	\begin{matrix}
		\textbf{P}_{\bm 0}:~\underset{\bm{\alpha}_{i},\bm{\eta}_{i}}{\mathop{\max}}\sum\limits_{t=1}^{T}\left[{\xi}_{1}\cdot\psi_{1}+{\xi}_{2}\cdot\psi_{2}+{\xi}_{3}\cdot\psi_{3}\right],   \\ 
		\text{s.t.}~{\alpha}_{i}(t)\cdot{{E}_{i}^\mathit{cpt}(t)}+{2}\cdot{E}_{i}^{\mathit{cmu}}(t)\le {{E}_{i}(t)}, \\ 
	\end{matrix}
\end{equation*}
where $ \bm{\alpha}_{i}=(\alpha_{i}(t),t=1,2,\cdots,T) $ is the sequence of selected $ \alpha_{i}(t) $, and $ \bm{\eta}_{i}=(\eta_{i}(t),t=1,2,\cdots,T) $ is the sequence of selected $ \eta_{i}(t) $. Note that we model the whole FL process instead of a single training round like Equation \ref{eq-18}, since $ c_{i} $'s overall contribution to a converged global model is a continuing pursuit, rather than an occasional rest on some peak that $ c_{i} $ has reached, however large. Besides, the local contributions are accumulated from the $ 1 $-st training round instead of the $ 0 $-th training round, as $ c_{i} $ cannot measure the local contribution without a prior local model in the initial training round.


\subsection{Problem transformation using MDP}

Obviously, $ \textbf{P}_{\bm 0} $ cannot be solved directly, as the selection of $ \alpha_{i}(t) $ and $ \eta_{i}(t) $ impacts the convergence of the global model over time, and the convergence of the global model impacts the selection of $ \alpha_{i}(t+1) $ and $ \eta_{i}(t+1) $ in turn. Thus, there is an intuition that $ \textbf{P}_{\bm 0} $ could be converted to a constrained MDP \cite{Altman1999}. In detail, we first define the MDP $ 3 $-tuple of $ c_{i} $'s hyper-parameter selection process as follows:
	
\noindent \emph{State}

A direct indicator evaluating $ c_{i} $'s local model is its prediction performance over $ D_{i} $. Thus, $ c_{i} $'s state $ s_{i}(t) $ is defined as:
\begin{equation}\label{eq-20}
{s}_{i}(t)=\left[\mathit{loss}_{i}(t),\mathit{acc}_{i}(t),\mathit{fs}_{i}(t)\right].
\end{equation}

\noindent \emph{Action}
	
Once a state is observed, an action is selected and performed by $ c_{i} $. Intuitively, the selected $ \alpha_{i}(t) $ and $ \eta_{i}(t) $ constitute the action $ a_{i}(t) $, defined as:
\begin{equation}\label{eq-21}
{a}_{i}(t)=\left[{\eta}_{i}(t), {\alpha}_{i}(t)\right].
\end{equation}

\noindent \emph{Reward Function}

The reward function in FL under our consideration specifies $ c_{i} $'s contribution of the local model trained by selecting hyper-parameters (action $ a_{i}(t) $) in a state, which could be defined as:
\begin{equation}\label{eq-22}
\begin{aligned}
{r}_{i}(t)
~=~&R\left({s}_{i}(t),{a}_{i}(t),{s}_{i}(t+1)\right)\\		
~=~&{\xi}_{1}\cdot\psi_{1}+{\xi}_{2}\cdot\psi_{2}+{\xi}_{3}\cdot\psi_{3}.
\end{aligned}
\end{equation}

Based on the above 3-tuple, we can convert $ \textbf{P}_{\bm 0} $ to a constrained MDP, denoted by $ \textbf{P}_{\bm 1} $, and defined as:
\begin{equation*}
	\begin{matrix}
		\textbf{P}_{\bm 1}~:~\underset{{{\mu}_{i}}}{\mathop{\max}}~G(\mu_{i}), \\ 
		\text{s.t.}~b_{i}(t)\leq 0,~t=1,2,\cdots,T, \\ 
	\end{matrix}
\end{equation*}
where $ {{\mu}_{i}}:s_{i}(t)\rightarrow a_{i}(t) $ is $ c_{i} $'s action-producing function, $ G(\mu_{i})= {\sum_{t=1}^{T}{\gamma}^{\left(t-1\right)}{r}_{i}(t)}  $, $ \gamma\in(0,1) $ is the discount factor, and $ b_{i}(t)={\alpha}_{i}(t)\cdot{{E}_{i}^\mathit{cpt}(t)}+{2}\cdot{E}_{i}^{\mathit{cmu}}(t) -{{E}_{i}(t)} $. Note that $ \textbf{P}_{\bm 1} $ is approximately equivalent to $ \textbf{P}_{\bm 0} $, since $ \gamma $ is usually set as $ 0.99 $ in practice. 

However, such a constrained MDP is still intractable, since the complexity is exacerbated in domains where the optimization objective $ G(\mu_{i}) $ and the constraints must be controlled jointly \cite{Meuleau1998}. We therefore convert $ \textbf{P}_{\bm 1} $ into a Lagrangian dual problem \cite{Altman1999} by adding the constraint conditions to the objective function as follows:
\begin{equation}\label{eq-23}
	J(\mu_{i})={\sum\limits_{t=1}^{T}{\gamma}^{\left(t-1\right)}{r}_{i}(t)}-{\sum\limits_{t=1}^{T}\lambda_{i}(t)b_{i}(t)},
\end{equation}
where $ \lambda_{i}(t)\geq 0 $ is the Lagrangian multiplier.

Then, the Lagrangian dual problem $ \textbf{P}_{\bm 2} $ is formulated as:
\begin{equation*}
	\textbf{P}_{\bm 2}~:~\underset{{{\lambda}_{i}(t)}}{\mathop{\min}}~\underset{{{\mu}_{i}}}{\mathop{\max}}~J(\mu_{i}),~t=1,2,\cdots,T.
\end{equation*}

The rationality of such a Lagrangian-based transformation from $ \textbf{P}_{\bm 1} $ to $ \textbf{P}_{\bm 2} $ is given in {\it Theorem \ref{thm-1}}.
\begin{theorem}\label{thm-1}
	For the prime constrained optimization problem $ \textbf{P}_{\bm 1} $ and its Lagrangian dual problem $ \textbf{P}_{\bm 2} $, there exists the optimal values $ G^{*} $ and $ J^{*}$, such that $ G^{*} \leq J^{*} $, and the equality sign holds if and only if $ \textbf{P}_{\bm 1} $ is convex.
\end{theorem}

\begin{proof}   
    We first recall the Slater condition in {\it Definition \ref{def-5}}.
    \begin{definition}[Slater condition \cite{Slater2014}]\label{def-5}
    	For the problems $ \textbf{P}_{\bm 1} $ and $ \textbf{P}_{\bm 2} $, there exists a $ \mu_{i} \in \operatorname{relint} \left(\bigcap\limits_{t=1}^{T} dom\left(b_{i}(t)\right)\right)$, such that $ b_{i}(t)\textless 0, t=1,2,\cdots,T $. 
    \end{definition}

    Then, to explore the equivalence of $ \textbf{P}_{\bm 1} $ and $ \textbf{P}_{\bm 2} $, we recall the {\it Strong Duality} in {\it Lemma \ref{lem-1}}.
    \begin{lemma}[Strong Duality \cite{Boyd2004}]\label{lem-1}
    	Suppose that Slater condition holds and $ \textbf{P}_{\bm 1} $ is convex. Then, $ G^{*}=J^{*}$.
    \end{lemma}
    
    Thus, according to {\it lemma \ref{lem-1}}, if $ \textbf{P}_{\bm 1} $ is convex, $ \textbf{P}_{\bm 2} $ is totally equivalent to $ \textbf{P}_{\bm 1} $.
    
    However, since the objective function in $ \textbf{P}_{\bm 1} $ is defined based on the combination of {\it Training loss value difference}, {\it Training accuracy difference}, and {\it Training F-1 score difference}, we cannot figure out the convexity of $ \textbf{P}_{\bm 1} $. In other words, the {\it Strong Duality} of $ \textbf{P}_{\bm 1} $ and $ \textbf{P}_{\bm 2} $ does not always hold.
    
    Therefore, we loose the convexity constraint of $ \textbf{P}_{\bm 1} $, and recall the {\it Weak Duality} in {\it lemma \ref{lem-2}}.
    
    \begin{lemma}[Weak Duality \cite{Boyd2004}]\label{lem-2}
    	 $ G^{*} $ is upper bounded by $ J^{*} $, i.e., $ G^{*}\leq J^{*} $.
    \end{lemma}

    {\it Weak Duality} discloses that the optimal value of $ \textbf{P}_{\bm 2} $ is approximate to the optimal value of $ \textbf{P}_{\bm 1} $, no matter $ \textbf{P}_{\bm 1} $ is convex or not, which completes the proof.

\end{proof}

{\it Theorem \ref{thm-1}} states that $ \textbf{P}_{\bm 2} $ is equivalent or approximately equivalent to $ \textbf{P}_{\bm 1} $, and converting $ \textbf{P}_{\bm 1} $ to $ \textbf{P}_{\bm 2} $ for tractably solving the optimal policy of the prime constrained MDP is rational.

\subsection{Problem solving by DDPG}

\begin{figure}
	\centering
	\includegraphics[scale=0.2,trim=70 675 70 55,clip]{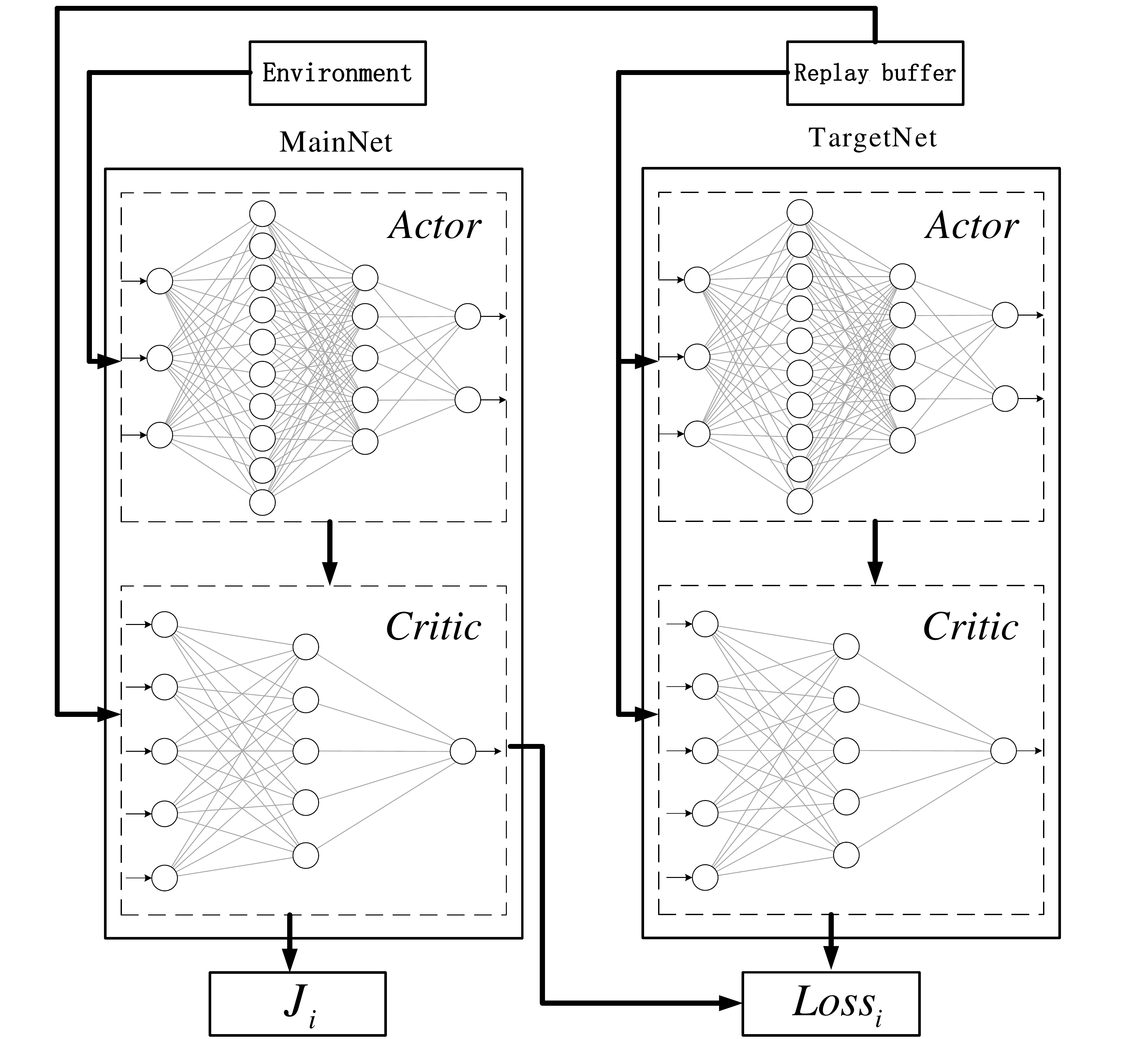}
	\caption{The structure of the DDPG model.}
	\label{fig-2}
\end{figure}

To solve the problem $ \textbf{P}_{\bm 2} $, it is intuitive to consider the DDPG as an effective method, since the state and action space in $ \textbf{P}_{\bm 2} $ are high-dimensional and continuous. Thus, we design a specific {\it Actor-Critic} structure for the DDPG in Fig. \ref{fig-2}. to solve the problem $ \textbf{P}_{\bm 2} $. Specifically, the {\it Actor} in the MainNet, which represents the parameterized deterministic policy $ \mu_{i} \left(\cdot; {\theta}^{\mu}_{i}(t)\right) $, takes the state $ {s}_{i}(t)=\left[\mathit{loss}_{i}(t),\mathit{acc}_{i}(t),\mathit{fs}_{i}(t)\right] $ as the input and outputs the action $ {a}_{i}(t)=\left[{\eta}_{i}(t), {\alpha}_{i}(t)\right] $ through two hidden layers. In addition, the \textit{Critic} in the MainNet, which has an input layer, two hidden layers, and an output layer, inputs $ 3 $ state items and $ 2 $ action items and outputs the value of $ c_{i} $'s parameterized state-action function $ Q^{\mu}_{i}\left(\cdot;\theta^{Q}_{i}(t) \right) $. As mentioned in Section III, since the TargetNet is a copy of the MainNet, the TargetNet can be parameterized by $ \mu'_{i} \left(\cdot; {\theta}^{\mu'}_{i}(t)\right) $ and $Q^{\mu'}_{i}\left(\cdot;{\theta}^{Q'}_{i}(t) \right) $.

As a result, the problem $\textbf{P}_{\bm 2}$ can be solved by iteratively update parameters ${\theta}^{\mu}_{i}(t)$, ${\theta}^{Q}_{i}(t)$, ${\theta}^{\mu'}_{i}(t)$, ${\theta}^{Q'}_{i}(t)$, and the Lagrangian multiplier $ \lambda_{i}(t) $ in the specific {\it Actor-Critic}-structured DDPG over its replay buffer $ \mathcal{B}_{i}=\left\{(s_{i}(t),a_{i}(t),r_{i}(t),s_{i}(t+1)),t=1,\cdots,T-1\right\} $. In detail, at the beginning of the $ t $-th training round, DDPG first randomly samples a batch $ B $ of experience tuples $\left\lbrace \left(s_{i}(k),a_{i}(k),r_{i}(k),s_{i}(k+1) \right),k=1,\cdots,\vert B\vert \right\rbrace $ from $ \mathcal{B}_{i} $, where $ \vert B\vert $ is the batch size. Then, the TargetNet computes the target value $ {y}_{i}(k) $ using ${\theta}_{i}^{\mu'}(t)$ and ${\theta}^{Q'}_{i}(t)$ for each experience sample, shown as:
\begin{align}\label{eq-24}
	{y}_{i}(k)=&\gamma\!\cdot\!{Q}_{i}^{\mu'}\!\left({s}_{i}(k\!+\!1),\mu'_{i}\left(s_{i}(k\!+\!1);{\theta}^{\mu'}_{i}\!(t\!-\!1)\right)\!;\!{{\theta}_{i}^{{Q}^{\prime}}\!(t\!-\!1)}\right)\nonumber\\
	&+{r}_{i}(k).
\end{align}
Next, DDPG updates the \textit{Critic} of the MainNet by minimizing the mean squared error (MSE) loss across all sampled experience tuples:
\begin{align}\label{eq-25}
	&\mathit{Loss}_{i}(t-1)\nonumber\\
	=&\frac{1}{|B|}{\sum\limits_{k=1}^{|B|}{\left({y}_{i}(k)-Q_{i}^{\mu}\left({s}_{i}(k),{a}_{i}(k);{{\theta }_{i}^{Q}(t-1)} \right) \right)}^{2}}.
\end{align}
Then, DDPG updates the {\it Critic} of the MainNet by the gradient descent method, shown as:
\begin{equation}\label{eq-26}
	{\theta}_{i}^{Q}(t)={\theta}_{i}^{Q}(t-1)-{l}_{i}^{C}\cdot {{\nabla}_{{\theta}_{i}^{Q}}} \mathit{Loss}_{i}(t-1),
\end{equation}
where $ l_{i}^{C} $ is the learning rate, and $ {{\nabla}_{{\theta}_{i}^{Q}}} \mathit{Loss}_{i}(t-1) $ is the gradient. Afterwards, DDPG updates the \textit{Actor} of the MainNet by the policy gradient ascent method, which is expressed as:
\begin{equation}\label{eq-27}
	{\theta}_{i}^{\mu}(t)={\theta}_{i}^{\mu}(t-1)+{l}_{i}^{A}\cdot {{\nabla}_{{\theta}_{i}^{\mu}}} J_{i}(t-1),
\end{equation}
where $l_{i}^{A}$ is the learning rate, and ${\nabla}_{{\theta}_{i}^{\mu}}J_{i}(t-1)$ is the policy gradient. Next, DDPG updates the TargetNet with a tiny updating rate $\beta_{i}$, which is shown as:
\begin{equation}\label{eq-28}
	\left\{ \begin{aligned}
	{\theta }_{i}^{\mu'}(t)= &~\beta_{i} \cdot {{\theta }_{i}^{\mu }(t-1)}+\left( 1-\beta_{i} \right) \cdot {\theta }_{i}^{\mu'}(t-1) \\ 
	{\theta }_{i}^{Q'}(t)= &~\beta_{i} \cdot {{\theta }_{i}^{Q}(t-1)}+\left( 1-\beta_{i} \right) \cdot {\theta }_{i}^{Q'}(t-1). \\ 
	\end{aligned} \right.
\end{equation}
Finally, DDPG updates the Lagrangian multiplier $ \lambda_{i}(t) $ by the gradient descent method, shown as:
\begin{equation}\label{eq-29}
	\lambda'_{i}(t)=\lambda_{i}(t)-l_{i}^{L}\cdot{{\nabla}_{{\lambda}_{i}}} J_{i}(t-1),
\end{equation}
where $ l_{i}^{L} $ is the learning rate, $ {{\nabla}_{{\lambda}_{i}}} J_{i}(t-1) $ is the gradient, and $ \lambda'_{i}(t) $ is the updated Lagrangian multiplier.
	
By iteratively performing the above operations, DDPG outputs $c_{i}$'s best policy $\mu_{i}^{*}$, which is the optimal result of $\textbf{P}_{\bm 2}$. The output of $\mu_{i}^{*}$ is the sequences of the selected local training hyper-parameters, i.e., the optimal results $ \bm{\alpha}^{*}_{i} $ and $ \bm{\eta}^{*}_{i} $ of the prime optimization problem $ \textbf{P}_{\bm 0} $.
	
\section{DDPG-assisted FL System}
In this section, we proposed our \underline{D}DPG-assisted \underline{a}daptive and \underline{p}rivacy-preserving federated learning (Dap-FL) system, which consists of the four parts: system initialization, local DDPG model update, local ML model training and uploading, and local ML model aggregation. For clarity, we illustrate a high-level view of the proposed Dap-FL in Fig. \ref{fig-3}.
	
\subsection{System Initialization}
Since there is a large number of interactions during the whole FL process, it is of vital importance to ensure the security of communication initially. Therefore, the $ \mathit{KDC} $ distributes a series of keys for all participants before collaborative training the global model. Specifically, the $ \mathit{KDC} $ first generates a series of Paillier-based private and public keys $ \{({\mathit{SK}}_{\mathit{ cs}}^{\mathit{Enc}},{\mathit{PK}}_{\mathit{cs}}^{\mathit{Enc}}) $ and $\{ ({\mathit{SK}}_{\mathit{ i}}^{\mathit{Enc}},{\mathit{PK}}_{\mathit{i}}^{\mathit{Enc}}),i=1,\cdots,N \}$. Then, the $ \mathit{KDC} $ sends $ \left\{{\mathit{SK}}_{\mathit{i}}^{\mathit{Enc}}, i=1,\cdots,N \right\}$ to the $ \mathit{CS} $, and sends $ {\mathit{PK}}_{\mathit{CS}}^{\mathit{Enc}} $, $ {\mathit{SK}}_{\mathit{CS}}^{\mathit{Enc}} $, and $ {\mathit{PK}}_{\mathit{i}}^{\mathit{Enc}} $ to each client $ c_{i} $. Besides, the $ \mathit{KDC} $ generates Paillier-based $ ({\mathit{SK}}_{\mathit{ cs}}^{\mathit{Sign}},{\mathit{PK}}_{\mathit{cs}}^{\mathit{Sign}}) $ and $\{ ({\mathit{SK}}_{\mathit{ i}}^{\mathit{Sign}},{\mathit{PK}}_{\mathit{i}}^{\mathit{Sign}}),i=1,\cdots,N \}$ for the $ \mathit{CS} $ and all clients to sign and verify messages.

After receiving the above keys, the $ \mathit{CS} $ distributes the original global ML model  $ W(0) $ in a secure manner. Specifically, the $ \mathit{CS} $ first signs on $ W(0) $ using its private key $ {\mathit{SK}}_{\mathit{cs}}^{\mathit{Sign}} $ based on the Paillier digital signature scheme, shown as:
\begin{equation}\label{eq-30}
W(0)\vert\vert\left(\varsigma(0),\tilde{\varsigma}(0)\right)\leftarrow\textit{Sign}(W(0), H, {\mathit{SK}}_{\mathit{cs}}^{\mathit{Sign}}),
\end{equation}
where $ (\varsigma(0),\tilde{\varsigma}(0)) $ is the digital signature. Then, the $ \mathit{CS} $ distributes $ W(0)\vert\vert\left(\varsigma(0),\tilde{\varsigma}(0)\right) $ to all clients. Next, all clients verify the signature, shown as:
\begin{equation}\label{eq-31}
\{0,1\}\leftarrow\textit{Verify}\left(W(0)\vert\vert \left(\varsigma(0),\tilde{\varsigma}(0) \right),H,{\mathit{PK}}_{\mathit{ cs}}^{\mathit{Sign}} \right).
\end{equation}
Thus, every $ c_{i} $ could initialize the local model $ w_{i}(0) $ as $ W(0) $. 

So far, the whole system is ready to collaboratively train the global ML model. Note that we include the distribution of the original global model into the system initialization, because the original global model is distributed in the form of plaintext, which is different from other FL training rounds.

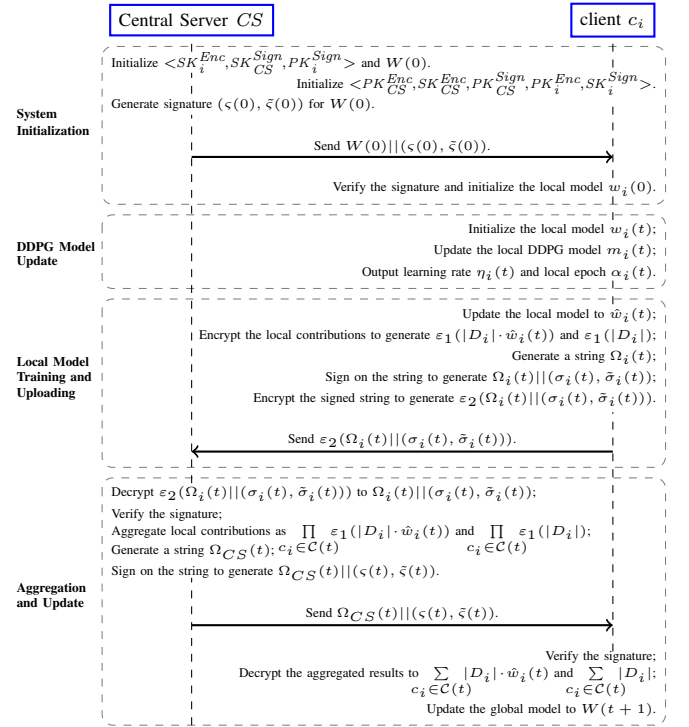
\begin{figure}
	\centering 
	\begin{tikzpicture}[scale=0.7]
	
	\path[fill=yellow!0,rounded corners, draw=black!50, dashed](0.3,6.7) rectangle (11.0,3.7);
	
	\node[right]at (-1.5,5.4) { \tiny{\textbf{System}}};
	
	\node[right]at (-1.5,5.1) { \tiny{\textbf{Initialization}}};
	
	\node[fill=blue!0,draw=blue,thick] at (2,7.2) (CS){\scriptsize{Central Server $ \mathit{CS} $}};
	
	\node[fill=blue!0,draw=blue,thick] at (10,7.2) (User){\scriptsize{client $ c_{i} $}};	
	
	\node[right] at (0.3,6.4){\tiny{Initialize $<\!\! {\mathit{SK}}_{\mathit{i}}^{\mathit{Enc}}\!,\! {\mathit{SK}}_{\mathit{CS}}^{\mathit{Sign}}\!,\! {\mathit{PK}}_{\mathit{i}}^{\mathit{Sign}}\!\!>$ and $W(0)$.}};
	
	\node[left] at (11,6){\tiny{Initialize $<\!\! {\mathit{PK}}_{\mathit{CS}}^{\mathit{Enc}}\!,\! {\mathit{SK}}_{\mathit{CS}}^{\mathit{Enc}}\!,\!{\mathit{PK}}_{\mathit{CS}}^{\mathit{Sign}}\!,\!{\mathit{PK}}_{\mathit{i}}^{\mathit{Enc}}\!,\!{\mathit{SK}}_{\mathit{i}}^{\mathit{Sign}} \!\!>$.}};
	
	\node[right]() at (0.3,5.6) {\tiny{Generate signature $ (\varsigma(0),\tilde{\varsigma}(0)) $ for $ W(0) $.}};	
	
	\draw[thick,->](2,4.6)-- (10,4.6);
	
	\node[above]() at (6,4.5) {\tiny{Send $ W(0)\vert\vert(\varsigma(0),\tilde{\varsigma}(0)) $.}};
	
	\node[left]() at (11,4.0) {\tiny{Verify the signature and initialize the local model $ w_{i}(0) $.}};
	
	
	\path[fill=yellow!0,rounded corners, draw=black!50, dashed](0.3,3.5) rectangle (11.0,2.1);
	
	\node[right]at (-1.5,2.9) { \tiny{\textbf{DDPG Model}}};
	
	\node[right]at (-1.5,2.6) { \tiny{\textbf{Update}}};
	
	\node[left]() at (11,3.2) {\tiny{Initialize the local model $ w_{i}(t) $;}};
	
	\node[left]() at (11,2.8) {\tiny{Update the local DDPG model $ m_{i}(t) $;}};
	
	\node[left]() at (11,2.4) {\tiny{Output learning rate $ {\eta}_{i}(t) $ and local epoch $ {\alpha}_{i}(t) $.}};
	
	
	\path[fill=yellow!0,rounded corners, draw=black!50, dashed](0.3,1.9) rectangle (11.0,-1.3);
	
	\node[right]at (-1.5,0.7) { \tiny{\textbf{Local Model}}};
	
	\node[right]at (-1.5,0.4) { \tiny{\textbf{Training and}}};
	
	\node[right]at (-1.5,0.1) { \tiny{\textbf{Uploading}}};

	\node[left]() at (11,1.6) {\tiny{Update the local model to $ \hat{w}_{i}(t) $;}};
	
	\node[left]() at (11,1.2) {\tiny{Encrypt the local contributions to generate $ \varepsilon_{1}(\vert D_{i} \vert\!\cdot\!\hat{w}_{i}(t)) $ and $ \varepsilon_{1}(\vert D_{i} \vert) $;}};
	
	\node[left]() at (11,0.8) {\tiny{Generate a string $ {\Omega}_{i}(t) $;}};
	
	\node[left]() at (11,0.4) {\tiny{Sign on the string to generate $ {\Omega}_{i}(t)\vert\vert(\sigma_{i}(t),\tilde{\sigma}_{i}(t)) $;}};
	
	\node[left]() at (11,0) {\tiny{Encrypt the signed string to generate $ \varepsilon_{2}({\Omega}_{i}(t)\vert\vert(\sigma_{i}(t),\tilde{\sigma}_{i}(t))) $.}};
	
	\draw[thick,<-](2,-1)-- (10,-1);
	
	\node[above]() at (6,-1.1) {\tiny{Send $ \varepsilon_{2}({\Omega}_{i}(t)\vert\vert(\sigma_{i}(t),\tilde{\sigma}_{i}(t))) $.}}; 
	
	
	\path[fill=yellow!0,rounded corners, draw=black!50, dashed](0.3,-1.5) rectangle (11.0,-6.2);
	
	\node[right]at (-1.5,-3.6) { \tiny{\textbf{Aggregation}}};
	
	\node[right]at (-1.5,-3.9) { \tiny{\textbf{and Update}}};    
	
	\node[right] at (0.3,-1.8){\tiny{Decrypt $ \varepsilon_{2}({\Omega}_{i}(t)\vert\vert(\sigma_{i}(t),\tilde{\sigma}_{i}(t))) $ to $ {\Omega}_{i}(t)\vert\vert(\sigma_{i}(t),\tilde{\sigma}_{i}(t)) $;}};
	
	\node[right] at (0.3,-2.2){\tiny{Verify the signature;}}; 
	
	\node[right] at (0.3,-2.7){\tiny{Aggregate local contributions as \!\!\!\!\!\!\!\! $ \prod\limits_{c_{i}\in \mathcal{C}(t)}\!\!\!\! \varepsilon_{1}(\vert D_{i} \vert\!\cdot\!\hat{w}_{i}(t)) $ and \!\!\!\!\!\!\!\! $ \prod\limits_{c_{i}\in \mathcal{C}(t)} \!\!\!\! \varepsilon_{1}(\vert D_{i} \vert) $;}};
	
	\node[right] at (0.3,-2.9){\tiny{Generate a string $ {\Omega}_{CS}(t) $;}};  
	
	\node[right] at (0.3,-3.3){\tiny{Sign on the string to generate $ {\Omega}_{CS}(t)\vert\vert(\varsigma(t),\tilde{\varsigma}(t)) $.}};
	
	\draw[thick,->](2,-4.3)-- (10,-4.3);
	
	\node[above]() at (6,-4.4) {\tiny{Send $ {\Omega}_{CS}(t)\vert\vert(\varsigma(t),\tilde{\varsigma}(t)) $.}};
	
	\node[left]() at (11,-4.9) {\tiny{Verify the signature;}};
	
	\node[left]() at (11,-5.4) {\tiny{Decrypt the aggregated results to \!\!\!\!\!\!\!\! $ \sum\limits_{c_{i}\in \mathcal{C}(t)} \!\!\!\! \vert D_{i} \vert\!\cdot\!\hat{w}_{i}(t) $ and \!\!\!\!\!\!\!\! $ \sum\limits_{c_{i}\in \mathcal{C}(t)} \!\!\!\! \vert D_{i} \vert $;}};
	
	\node[left]() at (11,-5.9) {\tiny{Update the global model to $ W(t+1) $.}};

	\draw[dashed](2,-6.2)--(2,-3.4);\draw[dashed](2,-1.6)--(2,5.4);\draw[dashed](2,5.8)--(2,6.2);\draw[dashed](2,6.6)--(CS);
	
	\draw[dashed](10,-4.8)--(10,-0.25);\draw[dashed](10,1.8)--(10,2.25);\draw[dashed](10,3.4)--(10,3.85);\draw[dashed](10,4.2)--(10,5.7);\draw[dashed](10,6.3)--(User);

	\end{tikzpicture}
	\caption{High-view of Dap-FL system.}
	\label{fig-3}
\end{figure}

\subsection{Local DDPG Model Update}
For the FL system under our consideration, clients need to select their local learning rates and local training epochs to train the local ML models. Therefore, the proposed DDPG-assisted adaptive hyper-parameter selection scheme is introduced to assist clients selecting proper $ \eta_{i}(t) $ and $ \alpha_{i}(t) $. Without loss of generality, we take the $ t $-th training round as an example to demonstrate how does the proposed scheme works.

In the $ t $-th training round, every client $ c_{i} $ initializes the local model $ w_{i}(t) $ as the the received global model $ W(t) $. By testing $ w_{i}(t) $ on the local dataset $ D_{i} $, the state $ {s}_{i}(t)=\left[\mathit{loss}_{i}(t),\mathit{acc}_{i}(t),\mathit{fs}_{i}(t)\right] $ can be observed. Then, $ c_{i} $ calculates a reward $ r_{i}(t-1) $ of the $ (t\!-\!1) $-th training round according to Equation (\ref{eq-22}). Next, $ c_{i} $ stores the experience tuple $ (s_{i}(t-1),a_{i}(t-1),r_{i}(t-1),s_{i}(t)) $ in the reply buffer $ \mathcal{B}_{i} $. Thus, $ c_{i} $ could update the local DDPG model on $ \mathcal{B}_{i} $, which can be simply expressed as:
\begin{equation}\label{eq-32}
	m_{i}(t)\leftarrow \textit{DDPG.train}\left( m_{i}(t-1),{{\mathcal{B}}_{i}} \right),
\end{equation}
where $ m_{i}(t) $ and $ m_{i}(t-1) $ are the DDPG models before and after updating in the $ t $-th training round. Consequently, the DDPG model $ m_{i}(t) $ outputs an action $ {a}_{i}(t)=\left[{\eta}_{i}(t), {\alpha}_{i}(t)\right] $. It is worth mentioning that the DDPG model cannot update in the $ 0 $-th training round because of the lack of prior experience. Hence, the client $ c_{i} $ will select $ \eta_{i}(0) $ and $ \alpha_{i}(0) $ randomly in the initial training round.

\subsection{Local ML Model Training and Uploading}

Based on the obtained action $ {a}_{i}(t) $, every client $ c_{i} $ updates $ {w}_{i}(t) $ to $ \hat{w}_{i}(t) $ according to Equation (\ref{eq-2}). Thus, $ c_{i} $ could upload $ \hat{w}_{i}(t) $ as a local contribution to the global model.

However, under the honest-but-curious assumption of the $ \mathit{CS} $ and clients, the precise individual local model should keep private to anyone but the client itself. Besides, owing to the existence of the active adversary $ \mathscr{A} $, the security of the local ML model must be under consideration during transmission. Therefore, we introduce the Paillier cryptosystem to mask the local model and verify the identity when transmitting the local ML models. In detail, $ c_{i} $ first encrypts the product of $ \vert D_{i} \vert $ and $ \hat{w}_{i}(t) $ using $ {\mathit{PK}}_{\mathit{cs}}^{\mathit{Enc}} $, which is expressed as:
\begin{equation}\label{eq-33}
	\varepsilon_{1}\left(\vert D_{i} \vert\cdot\hat{w}_{i}(t)\right)\leftarrow\textit{Enc} \left(\vert D_{i} \vert\cdot\hat{w}_{i}(t);{\mathit{PK}}_{\mathit{cs}}^{\mathit{Enc}}\right).
\end{equation}
Similarly, the size of the dataset $ \vert D_{i} \vert $ is encrypted by:
\begin{equation}
\varepsilon_{1}(\vert D_{i} \vert)\leftarrow\textit{Enc} \left(\vert D_{i} \vert;{\mathit{PK}}_{\mathit{cs}}^{\mathit{Enc}}\right).
\end{equation}
Subsequently, $ c_{i} $ combines $ \varepsilon_{1}(\vert D_{i} \vert\cdot\hat{w}_{i}(t)) $ and $ \varepsilon_{1}(\vert D_{i} \vert) $ as a string $ {\Omega}_{i}(t)=\varepsilon_{1}(\vert D_{i} \vert\cdot\hat{w}_{i}(t)) \vert\vert \varepsilon_{1}(\vert D_{i} \vert) $ and signs on it, which is expressed as:
\begin{equation}\label{eq-35}
{\Omega}_{i}(t)\vert\vert\!\left(\sigma_{i}(t),\!\tilde{\sigma}_{i}(t)\right)\!\leftarrow\!\textit{Sign}({\Omega}_{i}(t),\!H,\! {{\mathit{SK}}}_{i}^{\mathit{Sign}}).
\end{equation}
Further, $ c_{i} $ masks $ {\Omega}_{i}(t)\vert\vert\left(\sigma_{i}(t),\tilde{\sigma}_{i}(t)\right) $ by $ {{\mathit{PK}}}_{i}^{\mathit{Enc}} $, shown as:
\begin{equation}\label{eq-36}
\varepsilon_{2}\left({\Omega}_{i}(t)\vert\vert\!\left(\sigma_{i}(t),\tilde{\sigma}_{i}(t)\right)\right)\!\leftarrow\!\textit{Enc}({\Omega}_{i}(t)\vert\vert\!\left(\sigma_{i}(t),\!\tilde{\sigma}_{i}(t)\right)\!,\! {{\mathit{PK}}}_{i}^{\mathit{Enc}}).
\end{equation}

As a result, $ c_{i} $ could upload the double masked and signed local contribution $ \varepsilon_{2}({\Omega}_{i}(t)\vert\vert\left(\sigma_{i}(t),\tilde{\sigma}_{i}(t)\right)) $ to the $ \mathit{CS} $.

\subsection{Local model aggregation and global model update}
When receiving the local contributions, the $ \mathit{CS} $ first unmasks them by their private keys, shown as:
\begin{equation}\label{eq-37}
	{\Omega}_{i}(t)\vert\vert\!\left(\sigma_{i}(t),\tilde{\sigma}_{i}(t)\right)\!\leftarrow\!\textit{Dec}\left(\varepsilon_{2}({\Omega}_{i}(t)\vert\vert\!\left(\sigma_{i}(t),\tilde{\sigma}_{i}(t)\right))\!,\!{{\mathit{SK}}}_{i}^{\mathit{Enc}}\right).
\end{equation}
Next, the $ \mathit{CS} $ verifies the signature, expressed as:
\begin{equation}\label{eq-38}
	\{0,1\}\leftarrow\textit{Verify}\left({\Omega}_{i}(t)\vert\vert\left(\sigma_{i}(t),\tilde{\sigma}_{i}(t)\right),H,{\mathit{PK}}_{i}^{\mathit{Sign}} \right).
\end{equation}
The clients passing the authentication are flagged as legitimate clients and are incorporated into a set $ \mathcal{C}(t) $. Then, every local contribution $ {\Omega}_{i}(t) $ belonging to a legitimate client is divided back into $ \varepsilon_{1}(\vert D_{i} \vert\cdot\hat{w}_{i}(t)) $ and $ \varepsilon_{1}(\vert D_{i} \vert) $.
Afterwards, the $ \mathit{CS} $ calculates the products of all legitimate clients' weighted local models and local dataset sizes respectively, which are expressed as:
\begin{equation*}
	\prod_{c_{i}\in \mathcal{C}(t)}  \varepsilon_{1}(\vert D_{i} \vert\cdot\hat{w}_{i}(t))\quad \text{and} \quad \prod_{c_{i}\in \mathcal{C}(t)}  \varepsilon_{1}(\vert D_{i} \vert).
\end{equation*}
Then, the $ \mathit{CS} $ combines the two products as a string:
\begin{equation}\label{eq-39}
   \Omega_{CS}(t)={\prod\limits_{c_{i}\in \mathcal{C}(t)}  \varepsilon_{1}(\vert D_{i} \vert\cdot\hat{w}_{i}(t))}\vert\vert {\prod\limits_{c_{i}\in \mathcal{C}^{t}}  \varepsilon_{1}(\vert D_{i} \vert)}. 
\end{equation}
Subsequently, the $ \mathit{CS} $ signs on $ \Omega_{CS}(t) $, shown as:
\begin{equation}\label{eq-40}
	{\Omega}_{CS}(t)\vert\vert\left(\varsigma(t),\tilde{\varsigma}(t)\right)\leftarrow\textit{Sign}({\Omega}_{CS}(t), H, {{\mathit{SK}}}_{CS}^{\mathit{Sign}}).
\end{equation}
Finally, the $ \mathit{CS} $ distributes the signed result to all clients.

After receiving $ {\Omega}_{CS}(t)\vert\vert\left(\varsigma(t),\tilde{\varsigma}(t)\right) $, each client could update the global model locally. Specifically, all clients first verify the digital signature, shown as:
\begin{equation}\label{eq-41}
	\{0,1\}\leftarrow\textit{Verify}\left({\Omega}_{CS}(t)\vert\vert\left(\varsigma(t),\tilde{\varsigma}(t)\right),H,{\mathit{PK}}_{ CS}^{\mathit{Sign}} \right).
\end{equation}
Then, every client divides $ {\Omega}_{CS}(t) $ back to the two products and decrypts them to obtain the aggregated results, shown as:
\begin{equation}\label{eq-42}
	\sum_{c_{i}\in \mathcal{C}(t)}\!\!\!\vert D_{i} \vert\!\cdot\!\hat{w}_{i}(t)\!\leftarrow\!\textit{Dec}\!\left(\!\prod_{c_{i}\in \mathcal{C}(t)}\!\!\!\varepsilon_{1}(\vert D_{i} \vert\!\cdot\!\hat{w}_{i}(t)), {{\mathit{SK}}}_{CS}^{\mathit{Enc}}\!\right)\!,
\end{equation}
and
\begin{equation}\label{eq-43}
\sum_{c_{i}\in \mathcal{C}(t)}  \vert D_{i} \vert\leftarrow\textit{Dec}\left(\!\prod_{c_{i}\in \mathcal{C}(t)}  \varepsilon_{1}(\vert D_{i} \vert), {{\mathit{SK}}}_{CS}^{\mathit{Enc}}\!\right).
\end{equation}
As a result, $ c_{i} $ could calculate the updated global model, which is expressed as:
\begin{equation}\label{eq-44}
	W(t+1)=\frac{\sum_{c_{i}\in \mathcal{C}(t)}  \vert D_{i} \vert\cdot\hat{w}_{i}(t)}{\sum_{c_{i}\in \mathcal{C}(t)}  \vert D_{i} \vert}.
\end{equation}
Note that the $ \mathit{CS} $ would request the global model to any clients if necessary.  

A complete and detailed description of the proposed Dap-FL system workflow is provided in Fig. \ref{fig-4}. We stress that, in the figure, the first training round is different from subsequent training rounds, as the local DDPG model cannot be updated without prior experience. By executing the proposed scheme, all clients could collaboratively train an optimal global model $ W^{*} $ in an efficient and privacy-preserving manner.

\begin{figure*}[htp]
	\centering
	\begin{tikzpicture}[scale=0.7]
		\path[fill=yellow!0, draw=black!50](0,0) rectangle (25,-34.5);
			
		\node[right] at (3.0, -0.5){\textbf{DDPG-assisted Adaptive and Privacy-preserving Federated Learning (Dap-FL) System}};
			
		\draw[thick](0.5,-0.9) -- (24.5,-0.9);
			
		\node[right] at (0.4,-1.4){$\bullet$ \textbf{Phase 0: Syetem Initialization}};
			
		\node[right] at (0.8, -2.0){\small - $ \mathit{KDC} $ generates $ ({\mathit{SK}}_{\mathit{ CS}}^{\mathit{Enc}},{\mathit{PK}}_{\mathit{CS}}^{\mathit{Enc}}) $ and $\{ ({\mathit{SK}}_{\mathit{ i}}^{\mathit{Enc}},{\mathit{PK}}_{\mathit{i}}^{\mathit{Enc}}),i=1,\cdots,N \}$ for encryption and decryption.};
		
		\node[right] at (0.8, -2.6){\small - $ \mathit{KDC} $ generates $ ({\mathit{SK}}_{\mathit{ CS}}^{\mathit{Sign}},{\mathit{PK}}_{\mathit{CS}}^{\mathit{Sign}}) $ and $\{ ({\mathit{SK}}_{\mathit{ i}}^{\mathit{Sign}},{\mathit{PK}}_{\mathit{i}}^{\mathit{Sign}}),i=1,\cdots,N \}$ for digital signature.};
			
		\node[right] at (0.8, -3.2) {\small - $ \mathit{KDC} $ sends $ {\mathit{SK}}_{\mathit{i}}^{\mathit{Enc}}, i=1,\cdots,N $ to the $ \mathit{CS} $.};
		
		\node[right] at (0.8, -3.8) {\small - $ \mathit{KDC} $ sends $ {\mathit{PK}}_{\mathit{CS}}^{\mathit{Enc}} $, $ {\mathit{SK}}_{\mathit{CS}}^{\mathit{Enc}} $, $ {\mathit{PK}}_{\mathit{i}}^{\mathit{Enc}} $ to each client $ c_{i} $.};
	
		\node[right] at (0.8, -4.4){\small - $ \mathit{KDC} $ sends $ {\mathit{SK}}_{\mathit{CS}}^{\mathit{Sign}} $ to the $ \mathit{CS} $ and $ {\mathit{SK}}_{\mathit{i}}^{\mathit{Sign}} $ to each $ c_{i} $, and publishes $ {\mathit{PK}}_{\mathit{CS}}^{\mathit{Sign}} $ and $ {\mathit{PK}}_{\mathit{i}}^{\mathit{Sign}}, i=1,\cdots,N $ to all participants.};
			
		\node[right] at (0.8, -5.0){\small  - $ \mathit{CS} $ signs on the original global model $ W(0) $, shown as $ W(0)\vert\vert(\varsigma(0),\tilde{\varsigma}(0))\leftarrow\textit{Sign}(W(0), H, {\mathit{SK}}_{\mathit{CS}}^{\mathit{Sign}}) $.};
			
		\node[right] at (0.8, -5.6){\small - $ \mathit{CS} $ distributes the signed original global model $ W(0)\vert\vert(\varsigma(0),\tilde{\varsigma}(0)) $ to all clients.};
		
		\node[right] at (0.8, -6.2){\small - $ c_{i} $ initializes the local model $ w_{i}(0) $ as $ W(0) $ after verifying the signature: $ \{0,1\}\leftarrow\textit{Verify}(W(0)\vert\vert (\varsigma(0),\tilde{\varsigma}(0) ),H,{\mathit{PK}}_{\mathit{ CS}}^{\mathit{Sign}} ) $.};
			
		\node[right] at (0.5, -6.8){$\bullet$ \textbf{Phase 1: Local DDPG Model Update}};
			
		\node[right] at (0.8, -7.4){\small - $ c_{i} $ initializes the local model $ w_{i}(t) $ as $ W(t) $, and observes the state $ {s}_{i}(t)=[\mathit{loss}_{i}(t),\mathit{acc}_{i}(t),\mathit{fs}_{i}(t)] $ by testing $ w_{i}(t) $ on $ D_{i} $.};		
			
		\node[right] at (0.8, -8.0){\small - $ c_{i} $ calculates the reward $ r_{i}(t-1) $ for the $ (t\!-\!1) $-th training round, shown as $ r_{i}(t-1)=R({s}_{i}(t-1),{a}_{i}(t-1),{s}_{i}(t)) $.};
		
        \node[right] at (0.8, -8.6){\small - $ c_{i} $ generates an  experience tuple $ (s_{i}(t-1),a_{i}(t-1),r_{i}(t-1),s_{i}(t)) $, and stores it in the reply buffer $ \mathcal{B}_{i} $.};

        \node[right] at (0.8, -9.2){\small - $ c_{i} $ randomly samples a batch of $ B $ of experience tuples from $ \mathcal{B}_{i} $.};      
        
        \node[right] at (0.8, -9.8){\small - $ c_{i} $ calculates $ y_{i}(k) $ for each experience sample by the TargetNet according to Equation (\ref{eq-24}).};  
        
        \node[right] at (0.8, -10.4){\small - $ c_{i} $ calculates $ Loss_{i}(t-1) $ across all sampled experience tuples according to Equation (\ref{eq-25}).};
        
        \node[right] at (0.8, -11.0){\small - $ c_{i} $ updates the {\it Critic} of the MainNet, shown as $ {\theta}_{i}^{Q}(t)={\theta}_{i}^{Q}(t-1)-{l}_{i}^{C}\cdot {{\nabla}_{{\theta}_{i}^{Q}}} \mathit{Loss}_{i}(t-1) $.}; 
        
        \node[right] at (0.8, -11.6){\small - $ c_{i} $ updates the \textit{Actor} of the MainNet, shown as $ {\theta}_{i}^{\mu}(t)={\theta}_{i}^{\mu}(t-1)+{l}_{i}^{A}\cdot {{\nabla}_{{\theta}_{i}^{\mu}}} J_{i}(t-1) $.};  
        
        \node[right] at (0.8, -12.2){\small - $ c_{i} $ updates the TargetNet, shown as:};
        
        \node[right] at (1.2, -12.8){\small $ {\theta }_{i}^{\mu'}\!(t)\!=\! \beta_{i} \cdot {{\theta }_{i}^{\mu }(t\!-\!1)}\!+\!\left( 1\!-\!\beta_{i} \right) \cdot {\theta }_{i}^{\mu'}(t\!-\!1) $ and $ {\theta }_{i}^{Q'}(t)\!=\! \beta_{i} \cdot {{\theta }_{i}^{Q}(t\!-\!1)}\!+\!\left( 1\!-\!\beta_{i} \right) \cdot {\theta }_{i}^{Q'}(t\!-\!1) $.}; 
        
        \node[right] at (0.8, -13.4){\small - $ c_{i} $ updates the Lagrangian multiplier, shown as $ \lambda'_{i}(t)=\lambda_{i}(t)-l_{i}^{L}\cdot{{\nabla}_{{\lambda}_{i}}} J_{i}(t-1) $.};
        
         \node[right] at (0.8, -14.0){\small - $ m_{i}(t) $ outputs an action $ {a}_{i}(t)=[{\eta}_{i}(t), {\alpha}_{i}(t)] $ as the hyper-parameters for the local training of $ w_{i}(t) $.};  
         
         \node[rotate = 90] at (18.9, -11.3) {$\underbrace{\hspace{3.2cm}}$}; 
         
         \node[right] at (18.9, -10.7){\emph{\small Simply expressed as updating}}; 
         
         \node[right] at (18.9, -11.3){\emph{\small local DDPG model on $ \mathcal{B}_{i} $:}};          
         
         \node[right] at (18.9, -11.9){\scriptsize $ m_{i}(t)\!\leftarrow\! \textit{DDPG.train}\left( m_{i}(t\!-\!1),{{\mathcal{B}}_{i}} \right) $};         
         
		\node[right] at (0.5,-14.6){$\bullet$ \textbf{Phase 2: Local ML Model Training and Uploading}};  	   	
       	
        \node[right] at (0.8, -15.2){\small - ${c_{i}}$ update the local model to $ \hat{w}_{i}(t) $, shown as $ \hat{w}_{i}(t)={w}_{i}(t)-{{\eta}_{i}(t)}\cdot \nabla {f}_{i}\left({{w}_{i}(t)};{D}_{i};{\alpha}_{i}(t)\right) $. };       
        
		\node[right] at (0.8, -15.8){\small - $ c_{i} $ encrypts $\vert{{D}_{i}}\vert\!\cdot\! \hat{w}_{i}(t)$ and $\vert{{D}_{i}}\vert$, shown as $ \varepsilon_{1}(\vert D_{i} \vert\!\cdot\!\hat{w}_{i}(t))\leftarrow\textit{Enc} (\vert D_{i} \vert\!\cdot\!\hat{w}_{i}(t),{\mathit{PK}}_{\mathit{CS}}^{\mathit{Enc}}) $ and $ \varepsilon_{1}(\vert D_{i} \vert)\leftarrow\textit{Enc} (\vert D_{i} \vert,{\mathit{PK}}_{\mathit{CS}}^{\mathit{Enc}}) $. };
		
		\node[right] at (0.8, -16.4){\small - $ c_{i} $ combines $ \varepsilon_{1}(\vert D_{i} \vert\cdot\hat{w}_{i}(t)) $ and $ \varepsilon_{1}(\vert D_{i} \vert) $ as a string $ {\Omega}_{i}(t)=\varepsilon_{1}(\vert D_{i} \vert\cdot\hat{w}_{i}(t)) \vert\vert \varepsilon_{1}(\vert D_{i} \vert) $. };
		
		\node[right] at (0.8, -17.0){\small - $ c_{i} $ signs on $ {\Omega}_{i}(t) $, shown as $ {\Omega}_{i}(t)\vert\vert\left(\sigma_{i}(t),\tilde{\sigma}_{i}(t)\right)\leftarrow\textit{Sign}({\Omega}_{i}(t), H, {{\mathit{SK}}}_{i}^{\mathit{Sign}}) $. };
                   
        \node[right] at (0.8, -17.6){\small - ${c_{i}}$ encrypts the signed string by $ {\mathit{PK}}_{i}^{\mathit{Enc}} $, shown as $ \varepsilon_{2}({\Omega}_{i}(t)\vert\vert\left(\sigma_{i}(t),\tilde{\sigma}_{i}(t)\right))\leftarrow\textit{Enc}({\Omega}_{i}(t)\vert\vert\left(\sigma_{i}(t),\tilde{\sigma}_{i}(t)\right), {{\mathit{PK}}}_{i}^{\mathit{Enc}}) $. };     
        
		\node[right] at (0.8, -18.2){\small - $ c_{i} $ uploads $ \varepsilon_{2}({\Omega}_{i}(t)\vert\vert(\sigma_{i}(t),\tilde{\sigma}_{i}(t))) $ to the $ \mathit{CS} $.};
	
		\node[right] at (0.5, -18.8){$\bullet$ \textbf{Phase 3: Local ML Model Aggregation and Global Model Update}};
		
		\node[right] at (0.8, -19.4){\small - $ \mathit{CS} $ decrypts $ \varepsilon_{2}({\Omega}_{i}(t)\vert\vert(\sigma_{i}(t),\tilde{\sigma}_{i}(t))) $ by $ {{\mathit{SK}}}_{i}^{\mathit{Enc}} $, shown as $ {\Omega}_{i}(t)\vert\vert(\sigma_{i}(t),\tilde{\sigma}_{i}(t))\leftarrow\textit{Dec}(\varepsilon_{2}({\Omega}_{i}(t)\vert\vert(\sigma_{i}(t),\tilde{\sigma}_{i}(t))), {{\mathit{SK}}}_{i}^{\mathit{Enc}}) $.};		
		
		\node[right] at (0.8, -20.0){\small - $ \mathit{CS} $ verifies the signature for every $ {\Omega}_{i}(t)\vert\vert(\sigma_{i}(t),\tilde{\sigma}_{i}(t)) $, shown as $ \{0,1\}\leftarrow\textit{Verify}({\Omega}_{i}(t)\vert\vert\left(\sigma_{i}(t),\tilde{\sigma}_{i}(t)\right),H,{\mathit{PK}}_{ i}^{\mathit{Sign}}) $.};		
			
        \node[right] at (1.1, -20.6){\small ( The clients passing the signature verification are flagged as legitimate clients, and are incorporated into a set $ \mathcal{C}(t) $. )};
        
        \node[right] at (0.8, -21.2){\small - $ \mathit{CS} $ divides legitimate client $ c_{i} $'s $ {\Omega}_{i}(t) $ back into $ \varepsilon_{1}(\vert D_{i} \vert\cdot\hat{w}_{i}(t)) $ and $ \varepsilon_{1}(\vert D_{i} \vert) $.};     
        
        \node[right] at (0.8, -21.8){\small - $ \mathit{CS} $ calculates the products of $ \varepsilon_{1}(\vert D_{i} \vert\!\cdot\!\hat{w}_{i}(t)) $s of all legitimate clients to obtain $ \prod_{c_{i}\in \mathcal{C}(t)}  \varepsilon_{1}(\vert D_{i} \vert\!\cdot\!\hat{w}_{i}(t)) $.}; 
        
        \node[right] at (0.8, -22.4){\small - $ \mathit{CS} $ calculates the products of $ \varepsilon_{1}(\vert D_{i} \vert) $s of all legitimate clients to obtain $ \prod_{c_{i}\in \mathcal{C}(t)}  \varepsilon_{1}(\vert D_{i} \vert) $.};    
        
        \node[right] at (0.8, -23.0){\small - $ \mathit{CS} $ combines the above two products as a string $ \Omega_{CS}(t)={\prod_{c_{i}\in \mathcal{C}(t)}  \varepsilon_{1}(\vert D_{i} \vert\cdot\hat{w}_{i}(t))}\vert\vert {\prod_{c_{i}\in \mathcal{C}(t)}  \varepsilon_{1}(\vert D_{i} \vert)} $.};        
        
        \node[right] at (0.8, -23.6){\small - $ \mathit{CS} $ signs on $ \Omega_{CS}(t) $, shown as $ {\Omega}_{CS}(t)\vert\vert\left(\varsigma(t),\tilde{\varsigma}(t)\right)\leftarrow\textit{Sign}({\Omega}_{CS}(t), H, {{\mathit{SK}}}_{CS}^{\mathit{Sign}}) $.};         
        
        \node[right] at (0.8, -24.2){\small - $ \mathit{CS} $ distributes $ {\Omega}_{CS}(t)\vert\vert(\varsigma(t),\tilde{\varsigma}(t)) $ to all clinets.}; 
                
        \node[right] at (0.8, -24.8){\small - $ c_{i} $ verifies the signature, shown as $ \{0,1\}\leftarrow\textit{Verify}({\Omega}_{CS}(t)\vert\vert(\varsigma(t),\tilde{\varsigma}(t)),H,{\mathit{PK}}_{CS}^{\mathit{Sign}} ) $.}; 
        
        \node[right] at (0.8, -25.4){\small - $ c_{i} $ divides $ {\Omega}_{CS}(t)\vert\vert(\varsigma(t),\tilde{\varsigma}(t)) $ back into $ \prod_{c_{i}\in \mathcal{C}(t)}  \varepsilon_{1}(\vert D_{i} \vert\!\cdot\!\hat{w}_{i}(t)) $ and $ \prod_{c_{i}\in \mathcal{C}(t)}  \varepsilon_{1}(\vert D_{i} \vert) $.};
                        	
		\node[right] at (0.8, -26.0){\small - $ c_{i} $ decrypts $ \prod_{c_{i}\in \mathcal{C}(t)}  \varepsilon_{1}(\vert D_{i} \vert\!\cdot\!\hat{w}_{i}(t)) $, shown as $ \sum_{c_{i}\in \mathcal{C}(t)}  \vert D_{i} \vert\cdot\hat{w}_{i}(t)\leftarrow\textit{Dec}(\prod_{c_{i}\in \mathcal{C}(t)}  \varepsilon_{1}(\vert D_{i} \vert\cdot\hat{w}_{i}(t)), {{\mathit{SK}}}_{i}^{\mathit{Enc}}) $.};	
		
		\node[right] at (0.8, -26.6){\small - $ c_{i} $ decrypts $ \prod_{c_{i}\in \mathcal{C}(t)}  \varepsilon_{1}(\vert D_{i} \vert) $, shown as $ \sum_{c_{i}\in \mathcal{C}(t)}  \vert D_{i} \vert\leftarrow\textit{Dec}(\prod_{c_{i}\in \mathcal{C}(t)}  \varepsilon_{1}(\vert D_{i} \vert), {{\mathit{SK}}}_{CS}^{\mathit{Enc}}) $.};	
		
		\node[right] at (0.8, -27.2){\small - $ c_{i} $ calculates the updated global model $ W(t+1) $, shown as $ W(t+1)=({\sum_{c_{i}\in \mathcal{C}(t)}  \vert D_{i} \vert\cdot\hat{w}_{i}(t)})/{\sum_{c_{i}\in \mathcal{C}(t)}  \vert D_{i} \vert} $.};		
			
		\node[right] at (0.5, -28.2){\textbf{Periodic Training:}};
			
		\node[right] at (0.8, -28.8){\small All clients and the $ \mathit{CS} $ execute \textbf{phase 1} to \textbf{phase 3} periodically until the global model converges to $ W^{*} $.};
		
		\node[right] at (0.5, -29.8){\textbf{NOTE:}};		
	
		\node[right] at (0.5, -30.4){$\bullet$ \textbf{Little Difference in the Initial Training Round}};
		
		\node[right] at (0.8, -31.0){\small - Due to the lack of prior experience, DDPG model cannot be updated and output actions.};		

		\node[right] at (0.8, -31.6){\small - Consequently, $ c_{i} $ randomly selects $ \eta_{i}(0) $ and $ \alpha_{i}(0) $ to update the local ML model.};
		
		\node[right] at (0.5, -32.2){$\bullet$ \textbf{Global Model in the $ \mathit{CS} $}};
		
		\node[right] at (0.8, -32.8){\small - $ \mathit{CS} $ would request the gobal model to any clients.};
		
		\node[right] at (0.8, -33.4){\small - $ c_{i} $ would upload the global model to $ \mathit{CS} $ after verifying the legality of the request.};		
			
	\end{tikzpicture}
	\caption {Detailed description of the proposed Dap-FL system.}
	\label{fig-4}
\end{figure*}

\section{Security Analysis}

In this section, we analyze the security properties of the proposed Dap-FL system. Particularly, following the security requirements discussed in Section II, our analysis mainly focuses on two aspects. On the one hand, the privacy of clients' individual local contributions is guaranteed under the assumption that clients and the $ \mathit{CS} $ are honest-but-curious. On the other hand, the source authentication and data integrity of local contributions and aggregation results is achieved under the assumption that an external active adversary $ \mathscr{A} $ exists.

\subsection{Privacy of local models is guaranteed}
Following the security requirements, any individual local model is considered private. Therefore, we analyze the confidentiality of $ c_{i} $'s updated local model $ \hat{w}_{i}(t) $ against the $ \mathit{CS} $ and other clients. 

$\bullet$ {\it Local model is confidential to the $ \mathit{CS} $ in the proposed Dap-FL.} When a client $ c_{i} $ uploads the updated local model to the $ \mathit{CS} $, it masks the plaintext twice based on the Paillier cryptosystem. Particularly, {\it Theorem \ref{thm-2}} proves how the first mask guarantees the semantic security of the ciphertext $ \varepsilon_{1}(\vert D_{i} \vert\!\cdot\!\hat{w}_{i}(t)) $ against chosen-plaintext attacks (CPA) \cite{Goldreich2004} from the $ \mathit{CS} $.

\begin{theorem}\label{thm-2}
	$ \varepsilon_{1}(\vert D_{i} \vert\!\cdot\!\hat{w}_{i}(t)) $ is semantic secure against CPA from the $ \mathit{CS} $ under the DCR assumption.
\end{theorem}
\begin{proof}
	For the sake of discussion, we regard $ \vert D_{i} \vert\!\cdot\!\hat{w}_{i}(t) $ as a constant. According to the Paillier Encryption scheme, the public key $ {\mathit{PK}}_{CS}^{\mathit{Enc}} $ is composed of $ n $ and $ g $. Thus, the ciphertext $ \varepsilon_{1}(\vert D_{i} \vert\!\cdot\!\hat{w}_{i}(t)) $ can be calculated by:
	\begin{equation}\label{eq-45}
		\varepsilon_{1}(\vert D_{i} \vert\cdot\hat{w}_{i}(t))=g^{\vert D_{i} \vert\cdot\hat{w}_{i}(t)}\cdot \zeta^{n} \bmod{{n}^{2}},
	\end{equation}
	where $ \zeta\in \mathbb{Z}_{{n}}^{*} $ is a random selected element.
	
	If the $ \mathit{CS} $ aims to recover the plaintext from $ \varepsilon_{1}(\vert D_{i} \vert\!\cdot\!\hat{w}_{i}(t)) $, it has to calculate the result of $ n $-the residues modulo $ n^{2} $. However, the $ \mathit{CS} $ cannot find a polynomial time distinguisher for $ n $-th residues modulo $ n^{2} $, according to the DCR assumption, i,e., {\it Assumption \ref{asp-1}}. As a result, $ \varepsilon_{1}(\vert D_{i} \vert\!\cdot\!\hat{w}_{i}(t)) $ is semantic secure against CPA from the $ \mathit{CS} $ under the DCR assumption.	
\end{proof}

$\bullet$ {\it Local model is confidential to other clients in the proposed Dap-FL.} Although the confidentiality against the $ \mathit{CS} $ is achieved, $  \varepsilon_{1}(\vert D_{i} \vert\!\cdot\!\hat{w}_{i}(t)) $ could be decrypted by other clients. Thus, $ c_{i} $ further encrypts $  {\Omega}_{i}(t)\vert\vert\left(\sigma_{i}(t),\tilde{\sigma}_{i}(t)\right) $ to $ \varepsilon_{2}({\Omega}_{i}(t)\vert\vert\left(\sigma_{i}(t),\tilde{\sigma}_{i}(t)\right)) $ using $ {{\mathit{PK}}}_{i}^{\mathit{Enc}} $ based on the Paillier Encryption scheme, and the ciphertext can be only decrypted by the $ \mathit{CS} $. Similarly, we give the {\it Theorem \ref{thm-3}} to demonstrate the semantic security of $ \varepsilon_{2}({\Omega}_{i}(t)\vert\vert\left(\sigma_{i}(t),\tilde{\sigma}_{i}(t)\right)) $ against CPA from other clients. Note that we omit the proof of {\it Theorem \ref{thm-3}}, as it is similar to {\it Theorem \ref{thm-2}}.
\begin{theorem}\label{thm-3}
	$ \varepsilon_{2}({\Omega}_{i}(t)\vert\vert\left(\sigma_{i}(t),\tilde{\sigma}_{i}(t)\right)) $ is semantic secure against CPA from other clients under the DCR assumption.
\end{theorem}

\subsection{Source authentication and data integrity is achieved}

As mentioned in the security requirements, the active adversary $\mathscr{A}$ may threaten the source authentication and data integrity by forging local contributions or aggregated results, thereby breaking down the whole FL system. Therefore, we analyze the security of the proposed Dap-FL against $\mathscr{A}$. 

$\bullet$ {\it The source authentication and data integrity of the individual client’s local contributions and the aggregated results are guaranteed in the proposed Dap-FL.} In Dap-FL, each client’s local contributions and the aggregated results are signed by Paillier signature. Since the Paillier signature is provably secure against chosen-message attacks (CMA) \cite{Goldreich2004} under the DCR assumption in the random oracle model \cite{Bellare1993}, the source authentication and data integrity can be guaranteed, which is shown in {\it Theorem \ref{thm-4}}.
\begin{theorem}\label{thm-4}
	Local contributions and aggregated results are secure against CMA from the active adversary $ \mathscr{A} $ under the DCR assumption.
\end{theorem}
\begin{proof}
	Without loss of generality, we express the local contribution or the aggregated result as a message $ M $, and the Paillier signature scheme is expressed as $ \Xi=\left(\textit{KeyGen}, \textit{Sign}, \textit{Verify} \right) $.
	
	Given a message $ M $, a legitimate signature $ \left(\sigma, \tilde{\sigma} \right) $ is computed as ${\sigma}=\frac{L\left( {{h}^{\varrho }}\bmod {{n}^{2}} \right)}{L\left( {{g}^{\varrho }}\bmod {{n}^{2}} \right)}\bmod n $ and $ \tilde{\sigma}={{\left(h\cdot {{g}^{-{\sigma}}} \right)}^{\frac{1}{n}\bmod \varrho }}\bmod n $. According to the DCR assumption, the adversary $ \mathscr{A} $ cannot distinguish for $ n $-th residues modulo $ n^{2} $ in polynomial time. In other words, in the random oracle model, there exists no negligible function $ negl $ such that:
	\begin{equation}
		\rm{Pr} \left[{\textit{Sig-forge}}_{\mathscr{A}, \Xi}^{CMA} \left(\iota \right)=1 \right] \leq \textit{negl}\left(\iota \right),
	\end{equation}
	where $ \iota $ is the length of the public key, and $ {\textit{Sig-forge}}_{\mathscr{A}, \Xi}^{\mathit{CMA}} \left(\iota \right) $ represents the CMA executed by the polynomial-time adversary $\mathscr{A}$. Therefore, the source authentication and data integrity of local contributions and aggregated results are guaranteed.
\end{proof}

\begin{table*}[htbp] 
	\centering  
	\caption{Local training hyper-parameters for different ML tasks and experimental settings}
	\label{tab-1} 
	\begin{threeparttable}  
		\begin{tabular}{m{1.7cm}<{\centering}|m{1.84cm}<{\centering}|m{1.85cm}<{\centering}|m{1.84cm}<{\centering}|m{1.84cm}<{\centering}|m{1.84cm}<{\centering}|m{1.84cm}<{\centering}|m{1.84cm}<{\centering}} 
			
			\hline  
			\hline  
			\makecell[c]{Task} & 
			\makecell[c]{{\it Dap-FL}} & 
			\makecell[c]{{\it Large}} & 
			\makecell[c]{{\it Small}} & 
			\makecell[c]{{\it DDPG-$ \eta $}} & 
			\makecell[c]{{\it DDPG-$ \alpha $}} & 
			\makecell[c]{{\it DDPG-client}\cite{Zhang2021}} & 
			\makecell[c]{{\it DQN}\cite{Sun2021}} \\
			\hline
			\makecell[c]{Logistic \\ MNIST} & 
			\makecell[c]{DDPG / DDPG} & 
			\makecell[c]{$ 10^{-2} $ / $ 20 $} & 
			\makecell[c]{$ 10^{-4} $ / $ 1 $} & 
			\makecell[c]{DDPG / $ 16 $} & 
			\makecell[c]{$ 10^{-3} $ / DDPG} & 
			\makecell[c]{$\backslash$} & 
			\makecell[c]{$\backslash$} \\
			\hline
			\makecell[c]{CNN \\ MNIST} & 
			\makecell[c]{DDPG / DDPG} & 
			\makecell[c]{$ 10^{-2} $ / $ 30 $ } & 
			\makecell[c]{$ 10^{-4} $ / $ 1 $} & 
			\makecell[c]{DDPG / $ 15 $} & 
			\makecell[c]{$ 10^{-3} $ / DDPG} & 
			\makecell[c]{$\backslash$} & 
			\makecell[c]{$\backslash$} \\
			\hline
			\makecell[c]{ResNet-18 \\ MNIST} & 
			\makecell[c]{DDPG / DDPG} & 
			\makecell[c]{$ 10^{-2} $ / $ 20 $} & 
			\makecell[c]{$ 10^{-4} $ / $ 1 $} & 
			\makecell[c]{DDPG / $ 11 $} & 
			\makecell[c]{$ 10^{-3} $ / DDPG} & 
			\makecell[c]{$\backslash$} & 
			\makecell[c]{$\backslash$} \\
			\hline
			\makecell[c]{CNN \\ \!\!Fashion-MNIST\!\!} & 
			\makecell[c]{DDPG / DDPG} & 
			\makecell[c]{\!\!$ 0.5\times10^{-3} $ / $ 25 $\!\!} & 
			\makecell[c]{$ 10^{-4} $ / $ 1 $} & 
			\makecell[c]{DDPG / $ 18 $} & 
			\makecell[c]{$ 10^{-3} $ / DDPG} & 
			\makecell[c]{$ 10^{-3} $ / $ 18 $} & 
			\makecell[c]{{\tiny$ 10\!^{\!-4} $or$ 10\!^{-\!3} $}/DQN} \\			 
			\hline  
			\hline 
			
		\end{tabular}
		
		\begin{tablenotes}
			\footnotesize
			\item[1] Each cell shows the local training hyper-parameters. The left of the Forward Slash is the local learning rate, and the right is the local training epoch. The Back Slash means we omit the experiments.   
		\end{tablenotes}
	\end{threeparttable}  
\end{table*}

\section{Experiments and Evaluation}

In this section, we conduct simulation experiments to validate that the proposed Dap-FL outperforms the basic FL and the state-of-the-art RL-based adaptive FL schemes in terms of global model prediction performance and communication efficiency. All simulations are implemented on the same computing environment (Windows 10, Intel (R) Core (TM) i5-8400 CPU @ 2.80GHz, NVIDIA GeForce RTX 3070, 16GB of RAM and 4T of memory) with Tensorflow, Keras, and PyCryptodome.

\subsection{Experimental Settings}

\noindent{\bfseries Dataset } 

The datasets used in our experiments include two widely used datasets, i.e., MNIST and Fashion-MNIST. Both of them have a training set of 60000 samples and a testing set of 10000 samples. Every sample has been size-normalized and centered in a fixed-size ($28\times28$) gray-level image depicting an item from one of ten different labels.

\noindent{\bfseries Machine Learning Tasks } 

The ML models used in our experiments include three types, i.e., multi-class logistic regression model, Convolutional Neural Network (CNN), and ResNet-18. In detail, the logistic regression model has $ 784 $ input nodes and outputs a $ 1 \times 10 $ vector indicating the prediction label of a given sample. The leveraged CNN contains 2 convolution layers, 2 pooling layers, and 2 fully connected layers. The ResNet-18 has 17 convolution layers, 1 fully connected layer, and corresponding pooling and batch normalization layers.

To explore the generality of the proposed Dap-FL system, we conduct simulation experiments on 4 different ML tasks: (1) Logistic regression on MNIST, (2) CNN on MNIST, (3) ResNet-18 on MNIST, and (4) CNN on Fashion-MNIST.

\noindent{\bfseries Adaptive FL Settings } 

With regard to the effectiveness of the proposed Dap-FL system, we conduct comparison experiments which are implemented as 7 different FL settings. All the considered FL settings share the same basic setting, where 20 clients with local training data and heterogeneous local resources aim to finish the corresponding ML task. Specifically, following a normal distribution with a mean of 600 and a standard deviation of 200, we divided the two leveraged datasets into 20 parts respectively as clients' local data. Besides, every 4 clients are assigned to one type of computational capability of 100\% GPU, 80\% GPU, 60\% GPU, 40\% GPU, and 20\% GPU to simulate the heterogeneous computational resources. The specific local training hyper-parameters combining corresponding ML tasks and the considered FL settings are illustrated in TABLE \ref{tab-1}, and the detailed description of each setting is as follows:  

\noindent$\bullet$ {\it Dap-FL.} Clients collaboratively finish ML tasks through the proposed Dap-FL system, which means each client's local learning rates and local training epochs are adjusted adaptively by a locally maintained DDPG model.

\noindent$\bullet$ {\it Large.} Clients collaboratively finish ML tasks with fixed and large local learning rates and local training epochs.

\noindent$\bullet$ {\it Small.} Clients collaboratively finish ML tasks with fixed and small local learning rates and local training epochs.

\noindent$\bullet$ {\it DDPG-$ \eta $.} Clients collaboratively finish ML tasks with fixed local training epochs and adaptive local learning rates adjusted by locally maintained DDPG models.

\noindent$\bullet$ {\it DDPG-$ \alpha $.} Clients collaboratively finish ML tasks with fixed local learning rates and adaptive local training epochs adjusted by locally maintained DDPG models.

\noindent$\bullet$ {\it DDPG-client.} The aggregator selects clients' local contributions according to Zhang {\em et al.}'s DDPG-based scheme \cite{Zhang2021}, and the clients' local training epochs and local learning rates are fixed to appropriate values.

\noindent$\bullet$ {\it DQN.} Clients' local training strategies follow the state-of-the-art DQN-based adaptive FL system proposed by Sun {\em et al.} \cite{Sun2021}. That is, clients' local learning rates are fixed to appropriate values, and the local training epochs are adjusted by a DQN model.
	
\subsection{Experimental results and evaluation}

In the {\em Dap-FL} setting, all client adaptively adjust their local training hyper-parameters by training their DDPG models. Fig. \ref{fig-5}. illustrates the loss value curves of the Actor and Critic of the DDPG model deployed on a client with 100\% GPU, from which we can see that the client's loss values of the Actor and Critic converge. It means that DDPG models could output the best policies to guide the local training of the client. Moreover, we illustrate the actions of the client given by the best policy in each training round in Fig. \ref{fig-6}. Taking the task CNN on MNIST in Fig. \ref{fig-6}.(b) as an example, the local learning rate and local training epoch decline slightly and flatten off ultimately along with the training round, which means the client can adopt larger local training hyper-parameters to find an approximately optimal global model rapidly and fine-tune the global model to a more precise optimal solution by adaptively adjusting local training hyper-parameters. Such a downswing is a little bit similar to an exponential decay of learning rate in the ML community but obviously better, since the local learning rate and local training epoch converge to appropriate values, i.e., $ 3.9\times10^{-4} $ and $ 2 $, rather than decline without limits, thereby preventing the global model from convergence.

\begin{figure}[htbp]
 	\centering
 	\subfigure[Logistic regression on MNIST.]{
 		\begin{minipage}[t]{0.5\linewidth}
 			\centering
 			\vspace{-0.5mm}
 			\includegraphics[scale=0.19,trim= 15 0 8 0,clip]{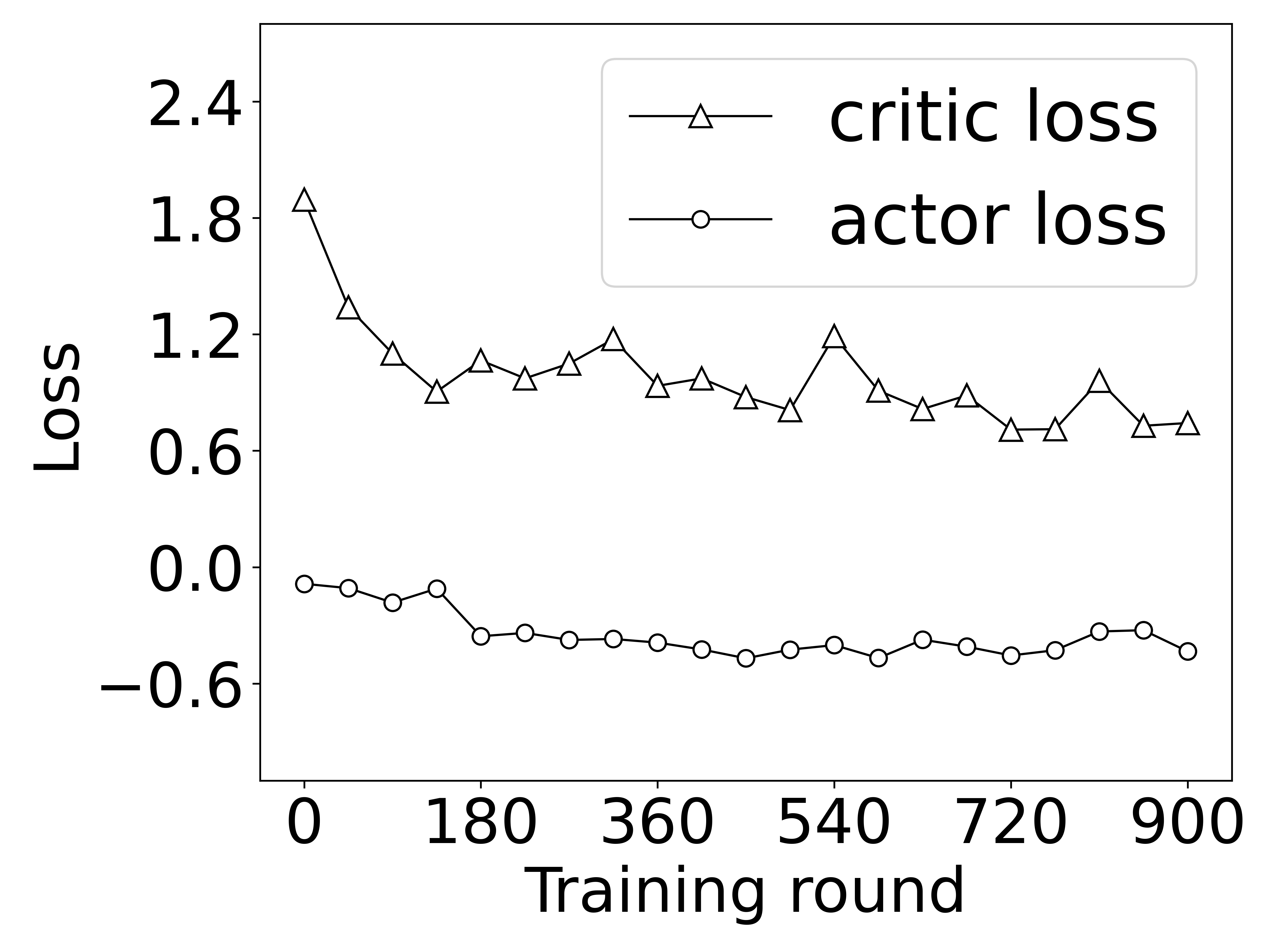}
 		\end{minipage}%
 	}%
    \subfigure[CNN on MNIST.]{
 	   \begin{minipage}[t]{0.5\linewidth}
 		   \centering
 	       \vspace{-0.5mm}
 	   	   \includegraphics[scale=0.19,trim= 15 0 8 0,clip]{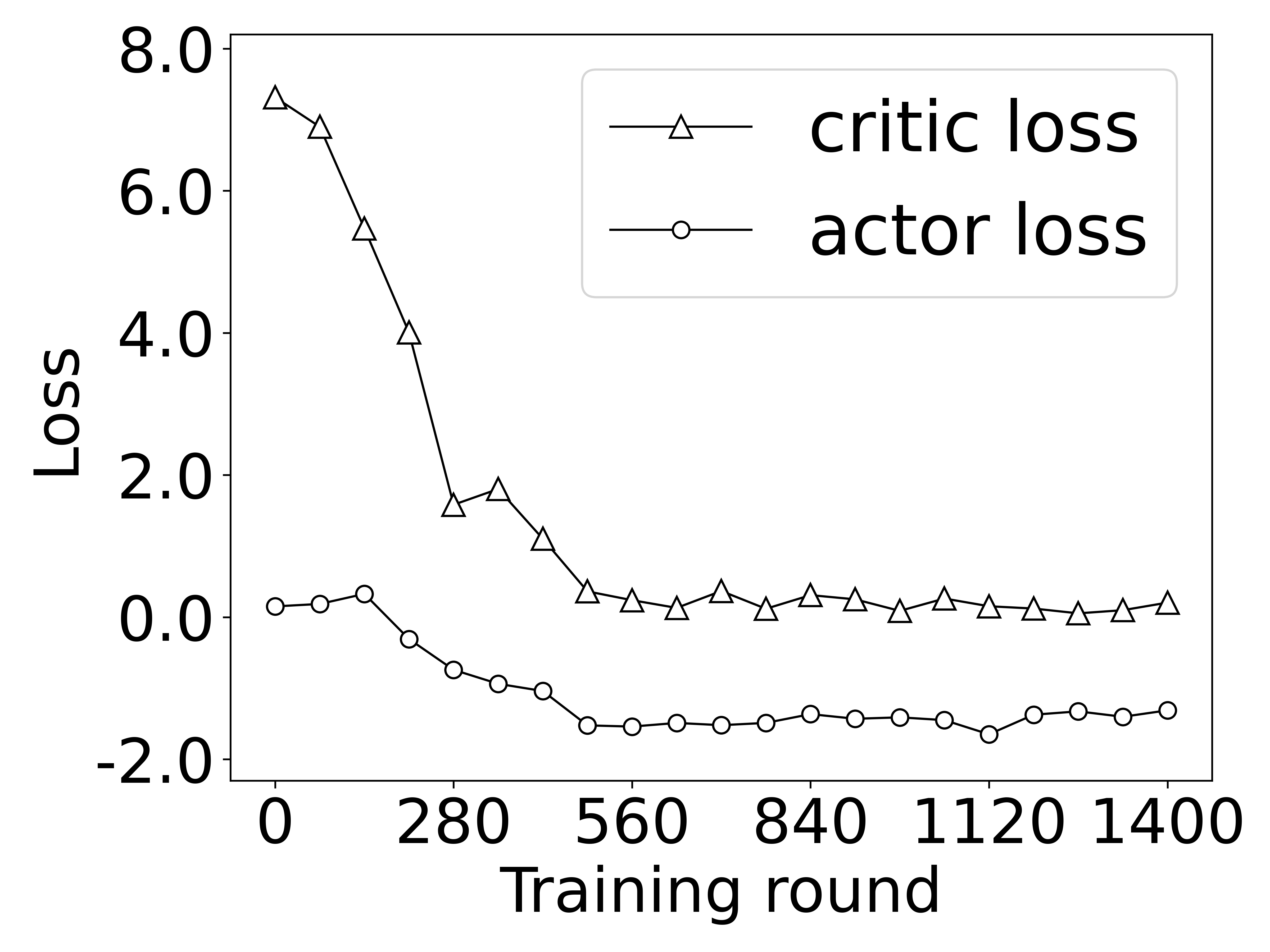}
 	   \end{minipage}%
    }%
    \vfill
    \subfigure[ResNet-18 on MNIST.]{
    	\begin{minipage}[t]{0.5\linewidth}
    		\centering
    		\vspace{-0.5mm}
    		\includegraphics[scale=0.19,trim= 15 0 8 0,clip]{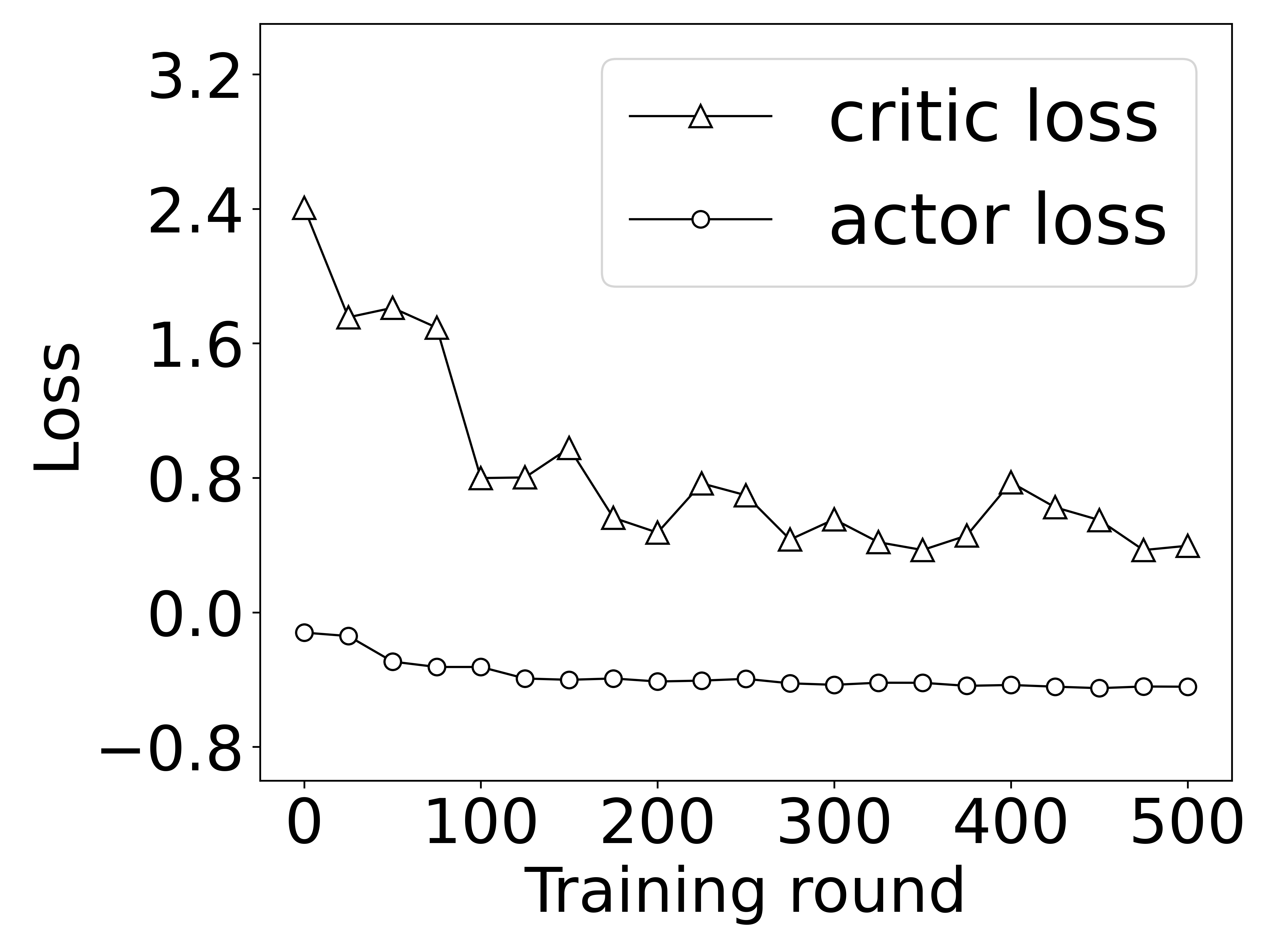}
    	\end{minipage}%
    }%
 	\subfigure[CNN on Fashion-MNIST.]{
	    \begin{minipage}[t]{0.5\linewidth}
		    \centering
		    \vspace{-0.5mm}
		    \includegraphics[scale=0.19,trim= 15 0 8 0,clip]{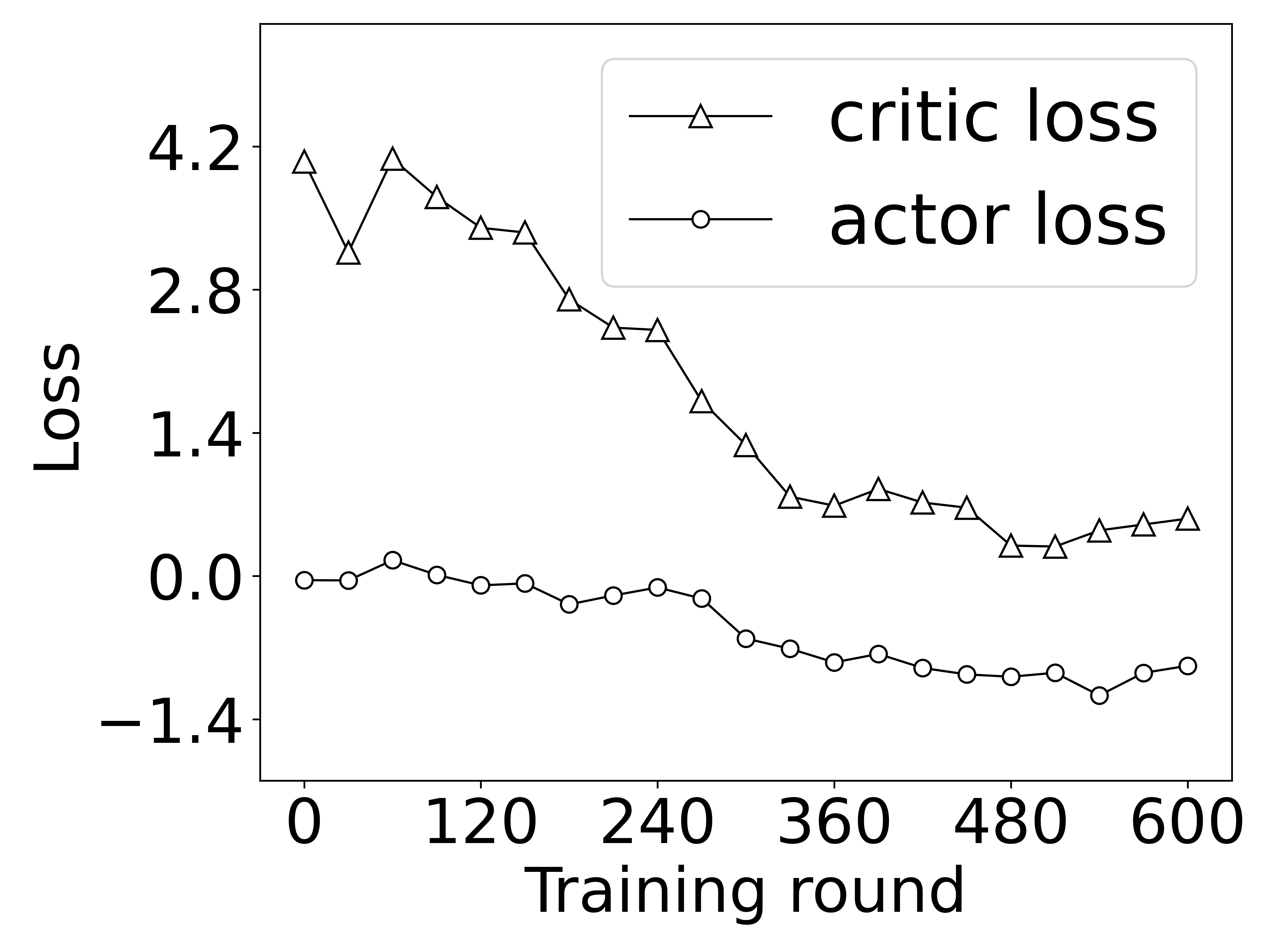}
    	\end{minipage}%
    }%
 	\centering
 	\caption{Loss value of the DDPG model maintained by a 100\% GPU client.} 
 	\label{fig-5}
 	
\end{figure}

\begin{figure}[htbp]
	\centering
	\subfigure[Logistic regression on MNIST.]{
	\begin{minipage}[t]{0.5\linewidth}
		\centering
		\vspace{-0.5mm}
		\includegraphics[scale=0.2,trim= 10 0 10 0,clip]{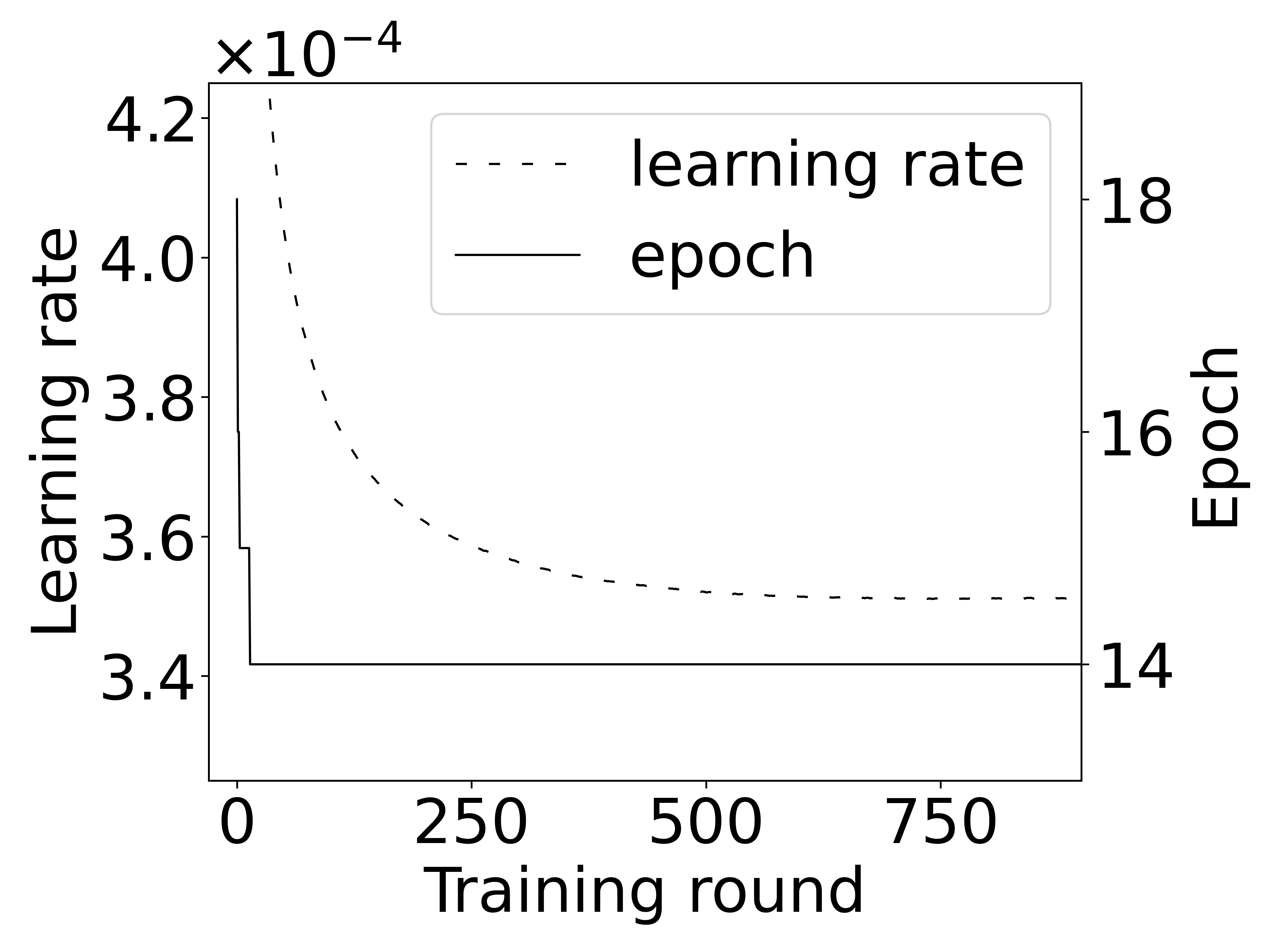}
	\end{minipage}%
    }%
	\subfigure[CNN on MNIST.]{
		\begin{minipage}[t]{0.5\linewidth}
			\centering
			\vspace{-0.5mm}
			\includegraphics[scale=0.2,trim= 10 0 10 0,clip]{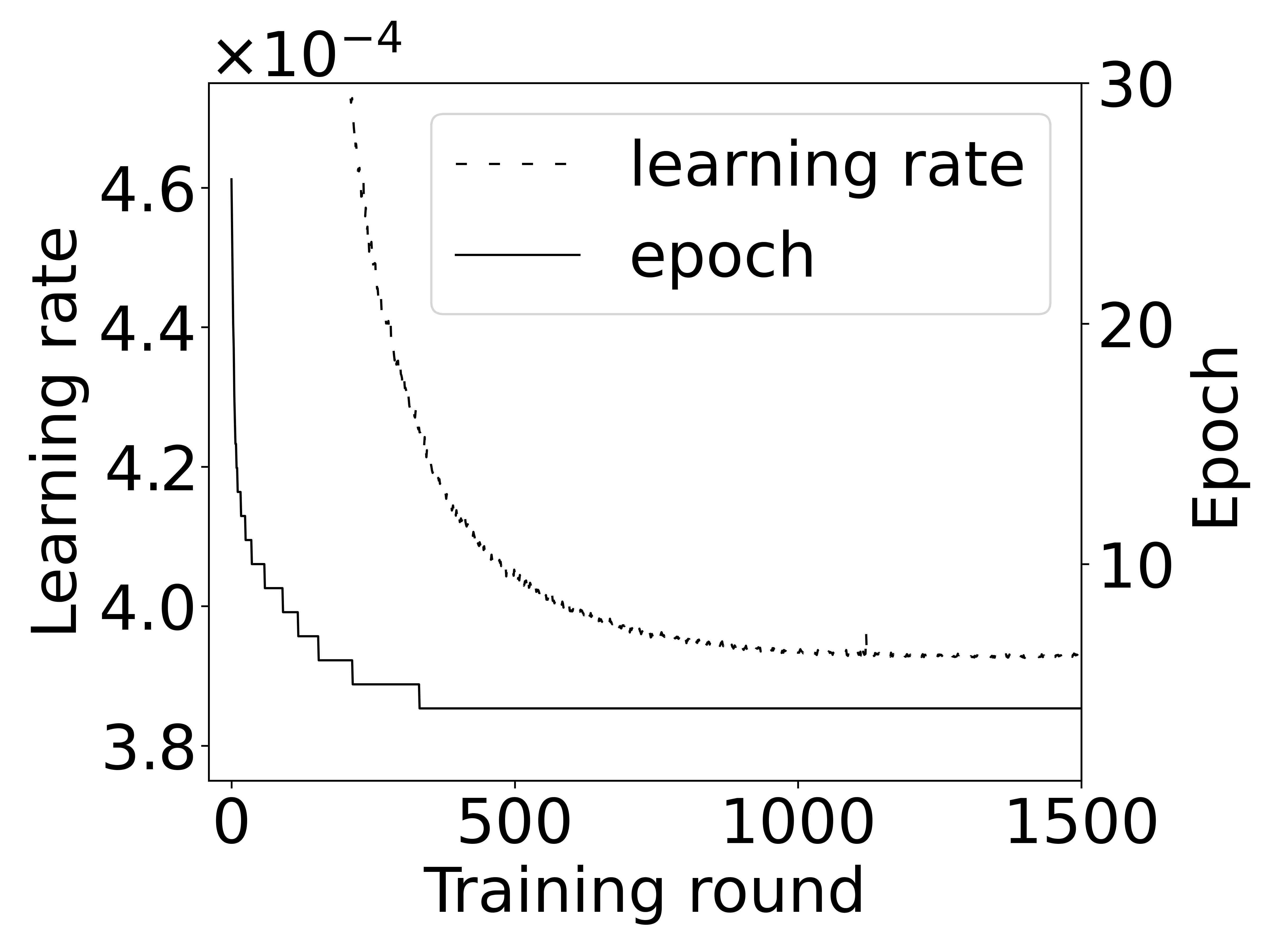}
		\end{minipage}%
	}%
	\vfill
	\subfigure[ResNet-18 on MNIST.]{
		\begin{minipage}[t]{0.5\linewidth}
			\centering
			\vspace{-0.5mm}
			\includegraphics[scale=0.2,trim= 10 0 10 0,clip]{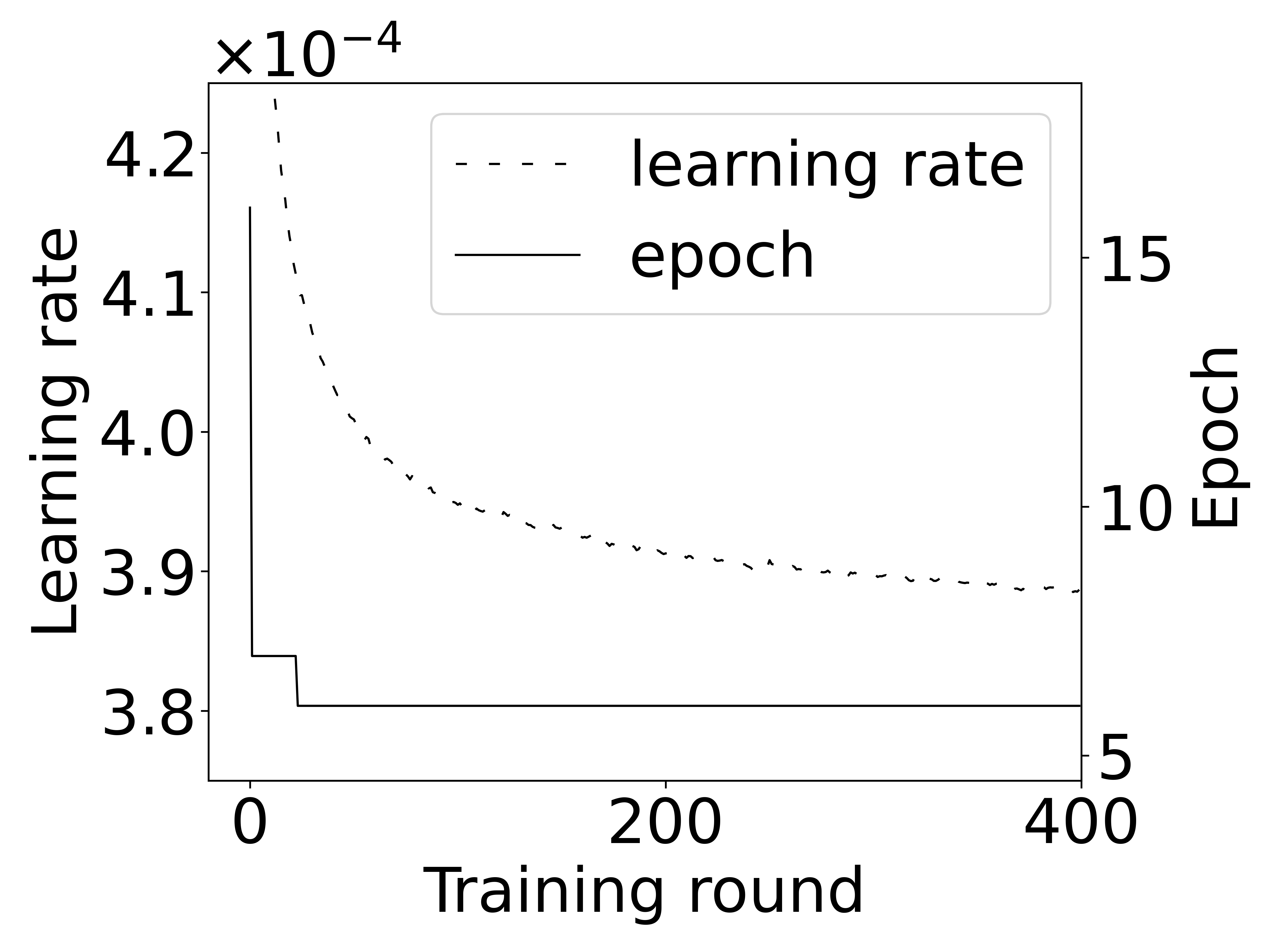}
		\end{minipage}%
	}%
	\subfigure[CNN on Fashion-MNIST.]{
	\begin{minipage}[t]{0.5\linewidth}
		\centering
		\vspace{-0.5mm}
		\includegraphics[scale=0.2,trim= 10 0 10 0,clip]{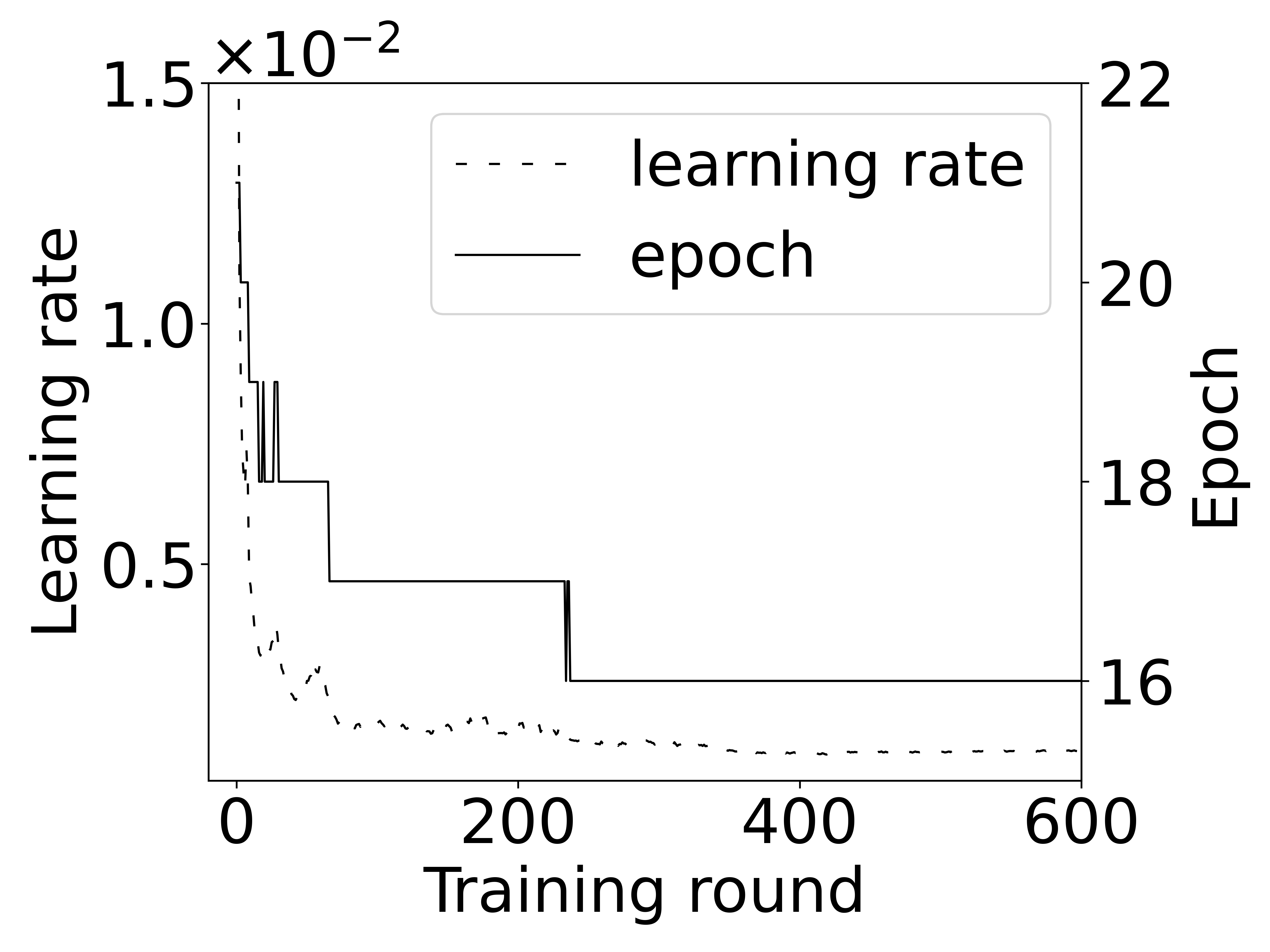}
	\end{minipage}%
    }%
	\centering
	\caption{Hyper-parameters adjusted by a 100\% GPU client.} 
	\label{fig-6}
	
\end{figure}

To validate how effective the proposed Dap-FL system achieves, we further illustrate the accuracy and loss values of the global models on the corresponding testing sets in the {\em Dap-FL}, {\em Large}, and {\em Small} settings in Fig. \ref{fig-7}., where the black, yellow and green lines represent the {\em Dap-FL}, {\em Large} and {\em Small} settings, respectively. In a nutshell, the global models trained through the proposed Dap-FL system converge faster than those in conventional hyper-parameter-fixed FL, and have better prediction accuracy. Specifically, taking the models of the CNN on MNIST task shown in Fig. \ref{fig-7}.(b) as an example, the converged loss values of the global models in the {\em Large} and {\em Small} are $ 0.1169 $ at $ 110 $-th training round and $ 0.0704 $ at $ 6704 $-th training round, while the global model trained through Dap-FL system has a converged loss value of $ 0.0898 $ at about $ 1340 $-th training round. Meanwhile, as shown in Fig. \ref{fig-7}.(f), the accuracy of the final global models in the {\em Large}, {\em Small} and {\em Dap-FL} are $ 97.55\% $, $ 98.02\% $, and $ 97.82\% $, respectively. Obviously, the proposed Dap-FL could accelerate the convergence rate of the global model compared to the global model in the {\em Small} setting, while achieving an accuracy close to the latter one. On the contrary, the final global model in the {\em Dap-FL} achieves a higher prediction accuracy than the global model in the {\em Large} setting, although the convergence rate is slightly slower. Note that the final global models even achieve the best prediction accuracy without losing convergence rate in the ML tasks of Logistic on MNIST and CNN on Fashion-MNIST, as shown in Fig. \ref{fig-7}.(e) and (h). To summarize, compared to conventional hyper-parameter-fixed FL, the proposed Dap-FL not only achieves a higher global model prediction accuracy, but also a more rapid global model convergence rate, which significantly reduces the communication overhead of all clients.

\begin{figure*}[htbp]
	\centering
	\subfigure[Loss of Logtistic on MNIST.]{
		\begin{minipage}[t]{0.25\linewidth}
			\centering
			\vspace{-0.5mm}
			\includegraphics[scale=0.22,trim= 10 0 5 0,clip]{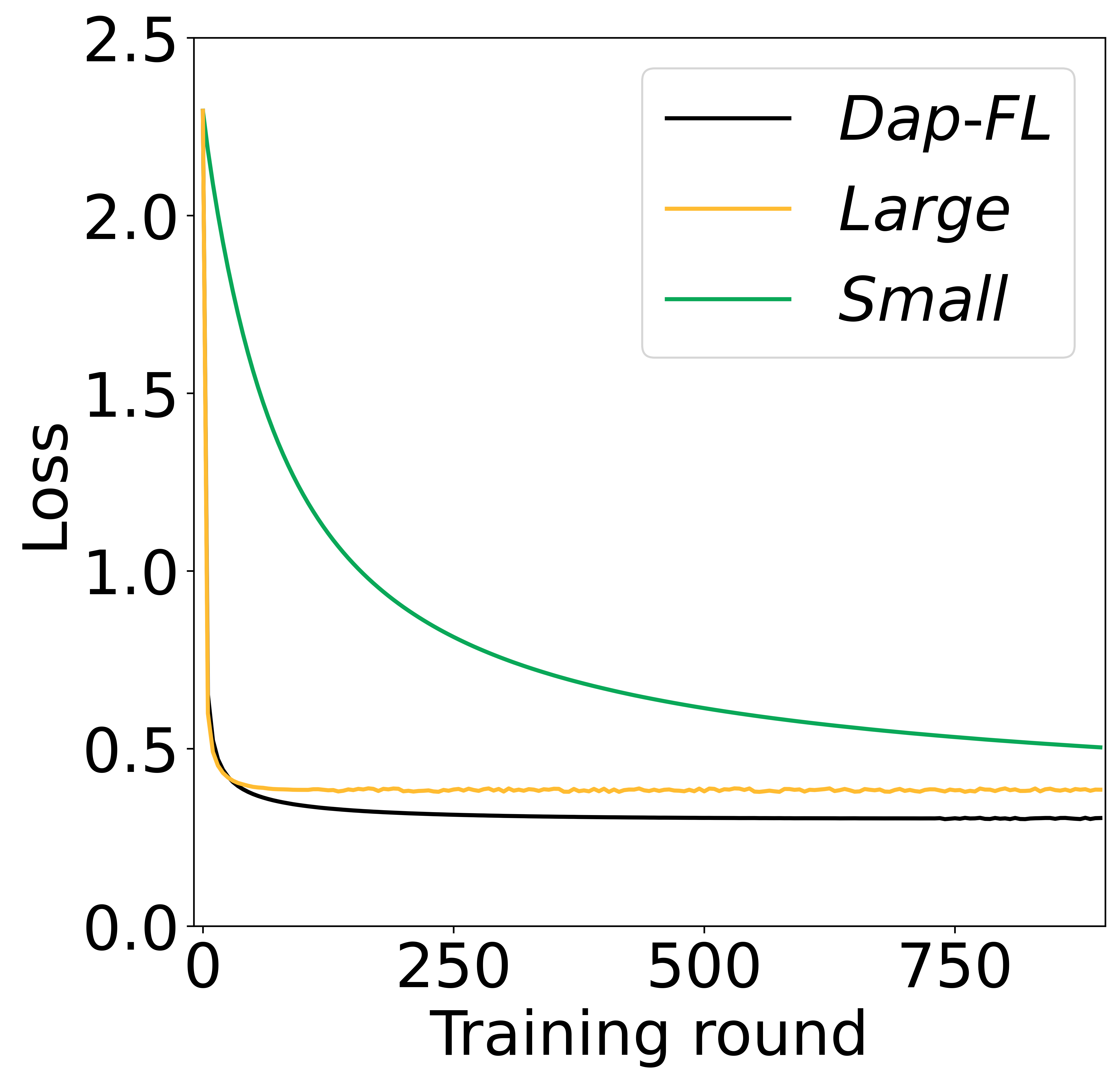}
		\end{minipage}%
		}%
	\subfigure[Loss of CNN on MNIST.]{
		\begin{minipage}[t]{0.25\linewidth}
			\centering
			\vspace{-0.5mm}
			\includegraphics[scale=0.22,trim= 10 0 5 0,clip]{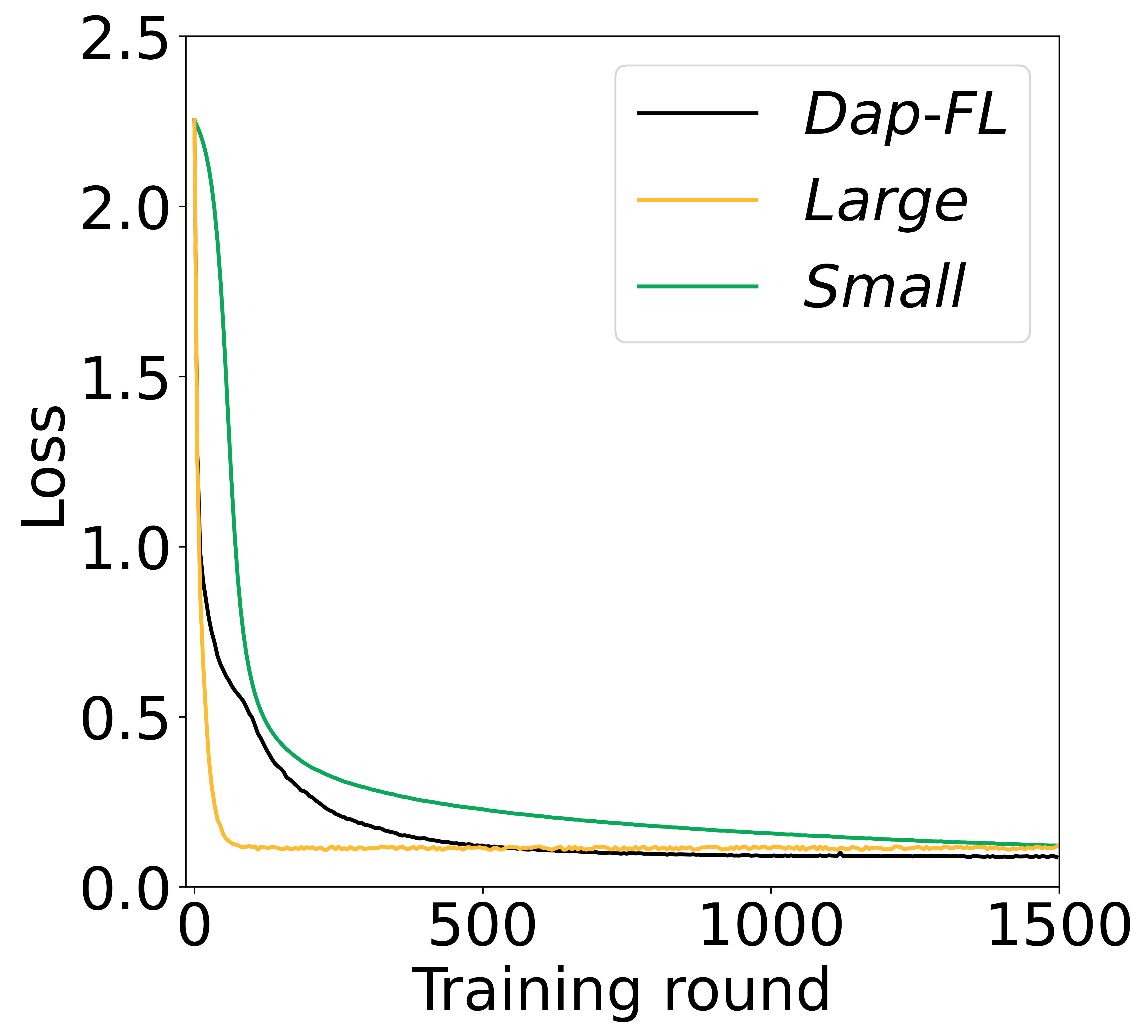}
		\end{minipage}%
	}%
	\subfigure[Loss of ResNet-18 on MNIST.]{
		\begin{minipage}[t]{0.25\linewidth}
			\centering
			\vspace{-0.5mm}
			\includegraphics[scale=0.22,trim= 10 0 5 0,clip]{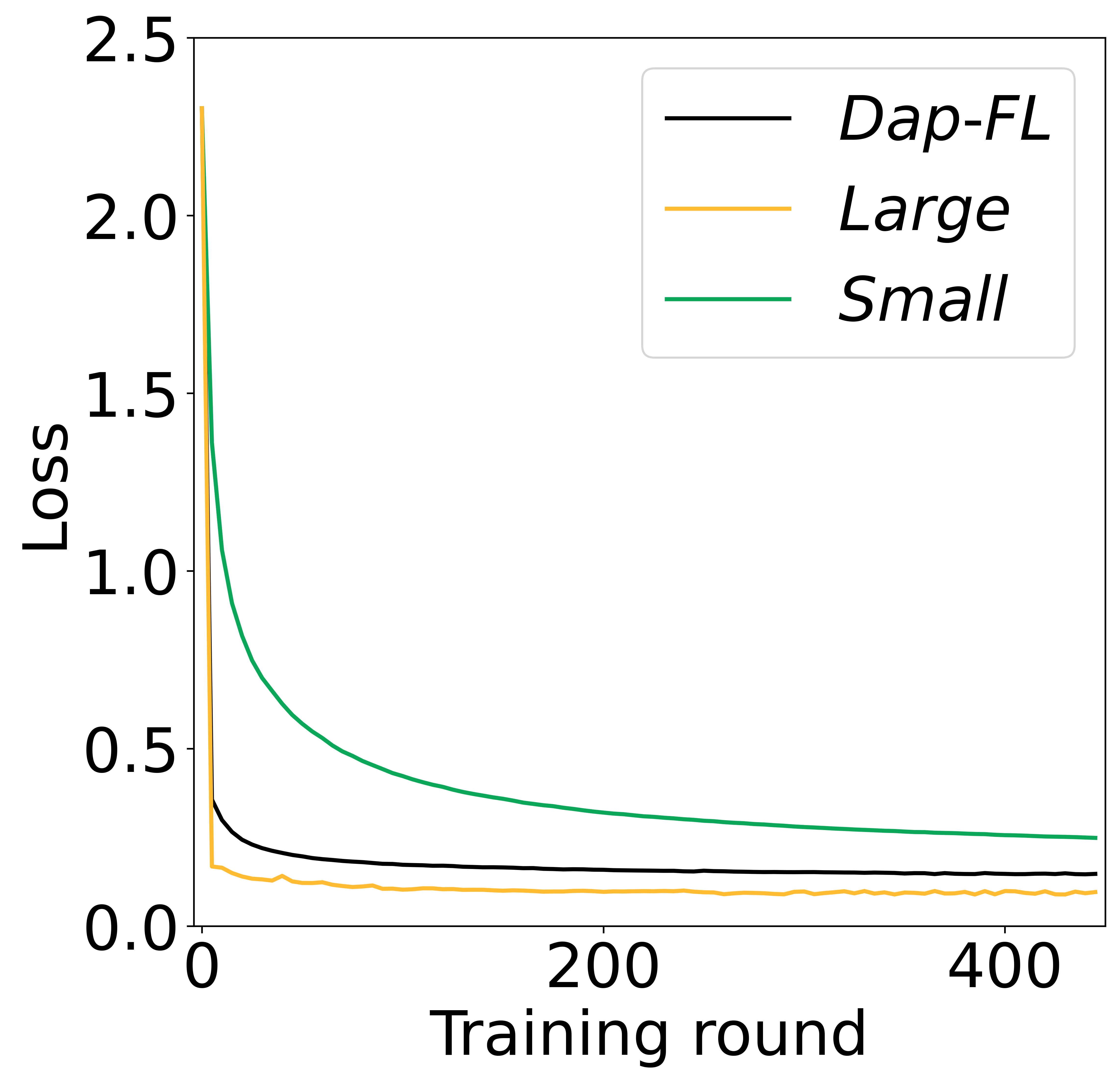}
		\end{minipage}%
	}%
	\subfigure[Loss of CNN on Fashion-MNIST.]{
	    \begin{minipage}[t]{0.25\linewidth}
		    \centering
		    \vspace{-0.5mm}
		    \includegraphics[scale=0.22,trim= 10 0 5 0,clip]{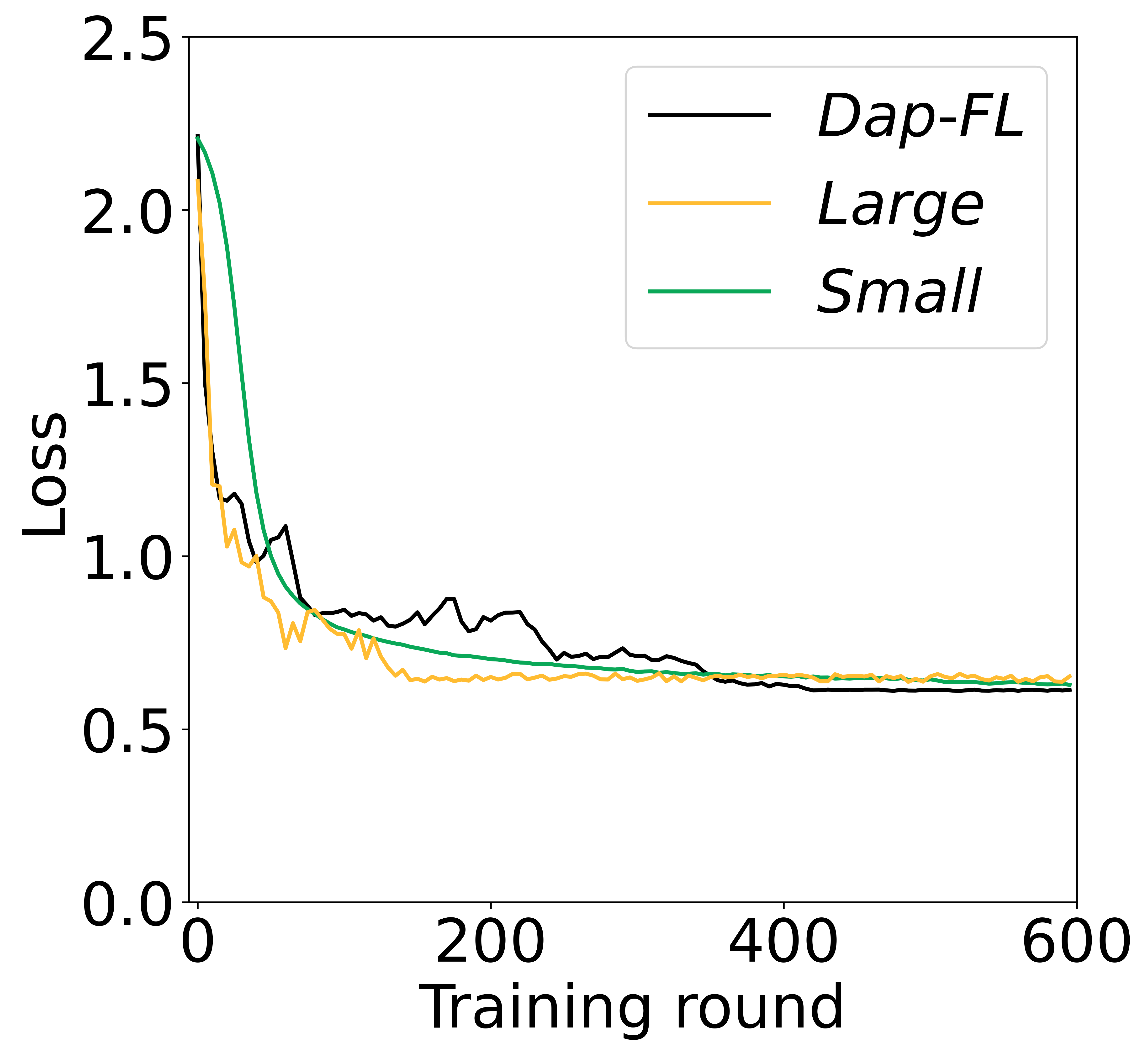}
	    \end{minipage}%
    }%
	\vfill
	\subfigure[Accuracy of Logtistic on MNIST.]{
	\begin{minipage}[t]{0.25\linewidth}
		\centering
		\vspace{-0.5mm}
		\includegraphics[scale=0.22,trim= 10 0 5 0,clip]{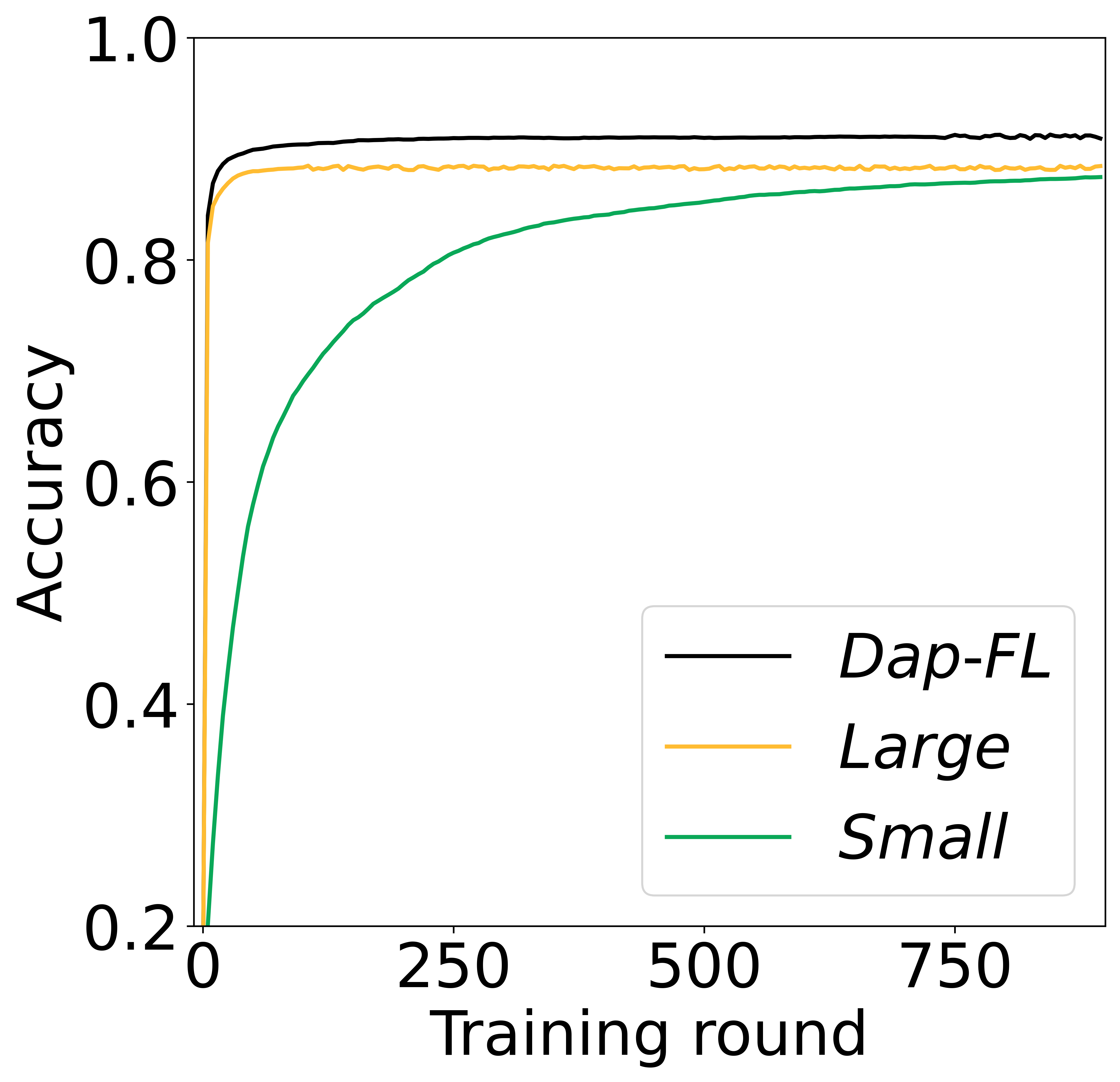}
	\end{minipage}%
	}%
	\subfigure[Accuracy of CNN on MNIST.]{
		\begin{minipage}[t]{0.25\linewidth}
			\centering
			\vspace{-0.5mm}
			\includegraphics[scale=0.22,trim= 10 0 5 0,clip]{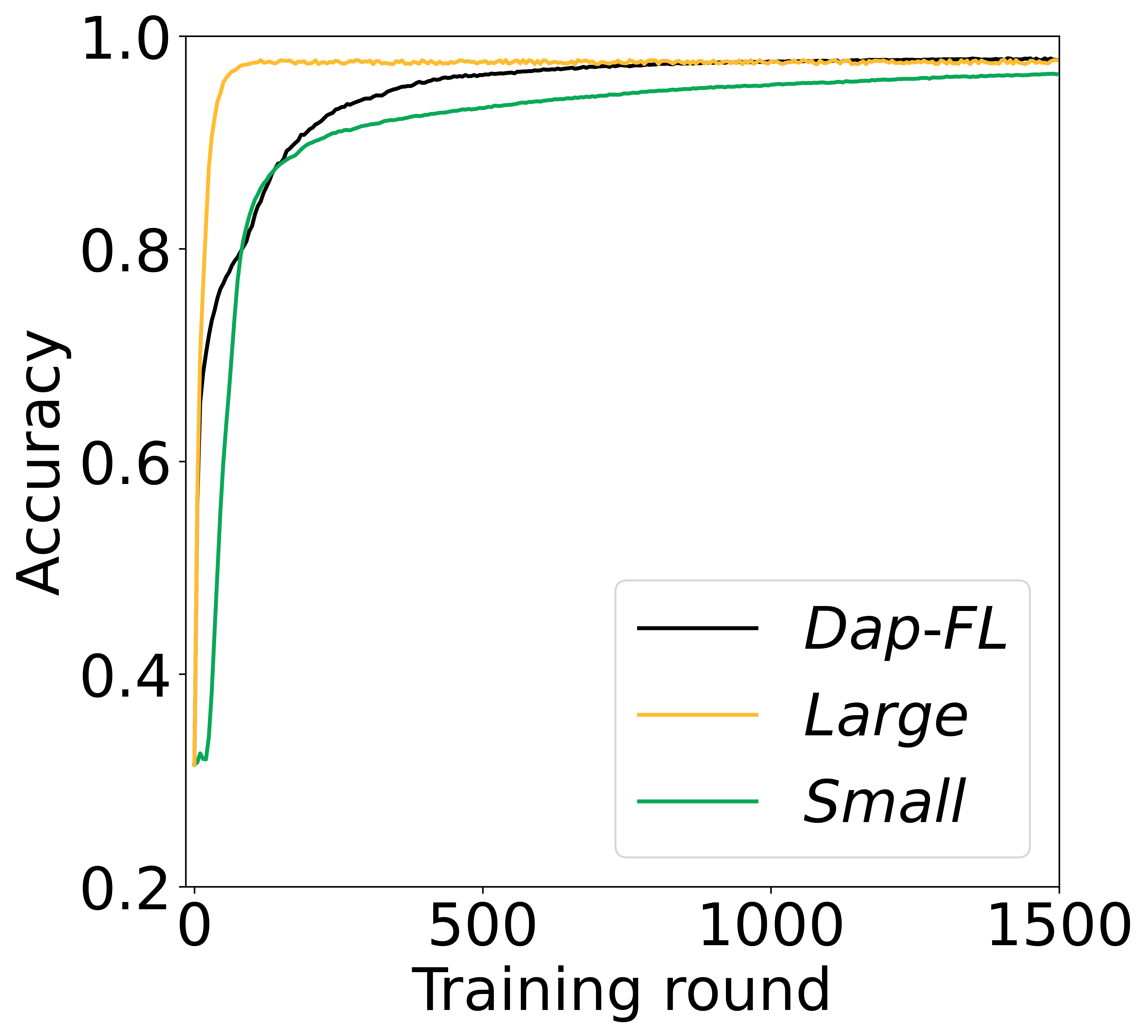}
		\end{minipage}%
	}%
	\subfigure[Accuracy of ResNet-18 on MNIST.]{
	    \begin{minipage}[t]{0.25\linewidth}
		    \centering
		    \vspace{-0.5mm}
		    \includegraphics[scale=0.22,trim= 10 0 5 0,clip]{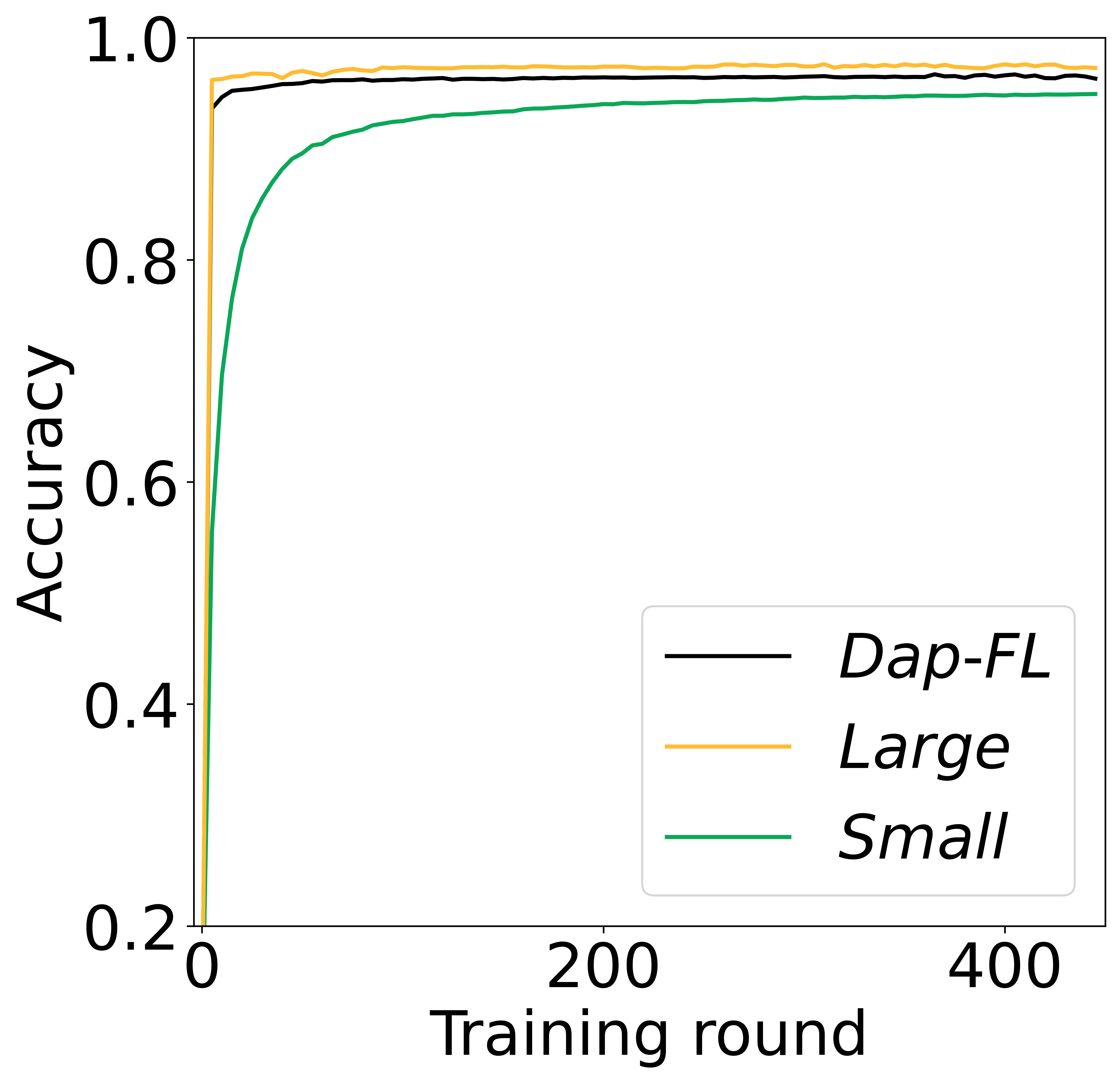}
	    \end{minipage}%
    }%
	\subfigure[\!\!\!\!Accuracy\;of\;CNN\;on\;Fashion-MNIST.]{
		\begin{minipage}[t]{0.25\linewidth}
			\centering
			\vspace{-0.5mm}
			\includegraphics[scale=0.22,trim= 10 0 5 0,clip]{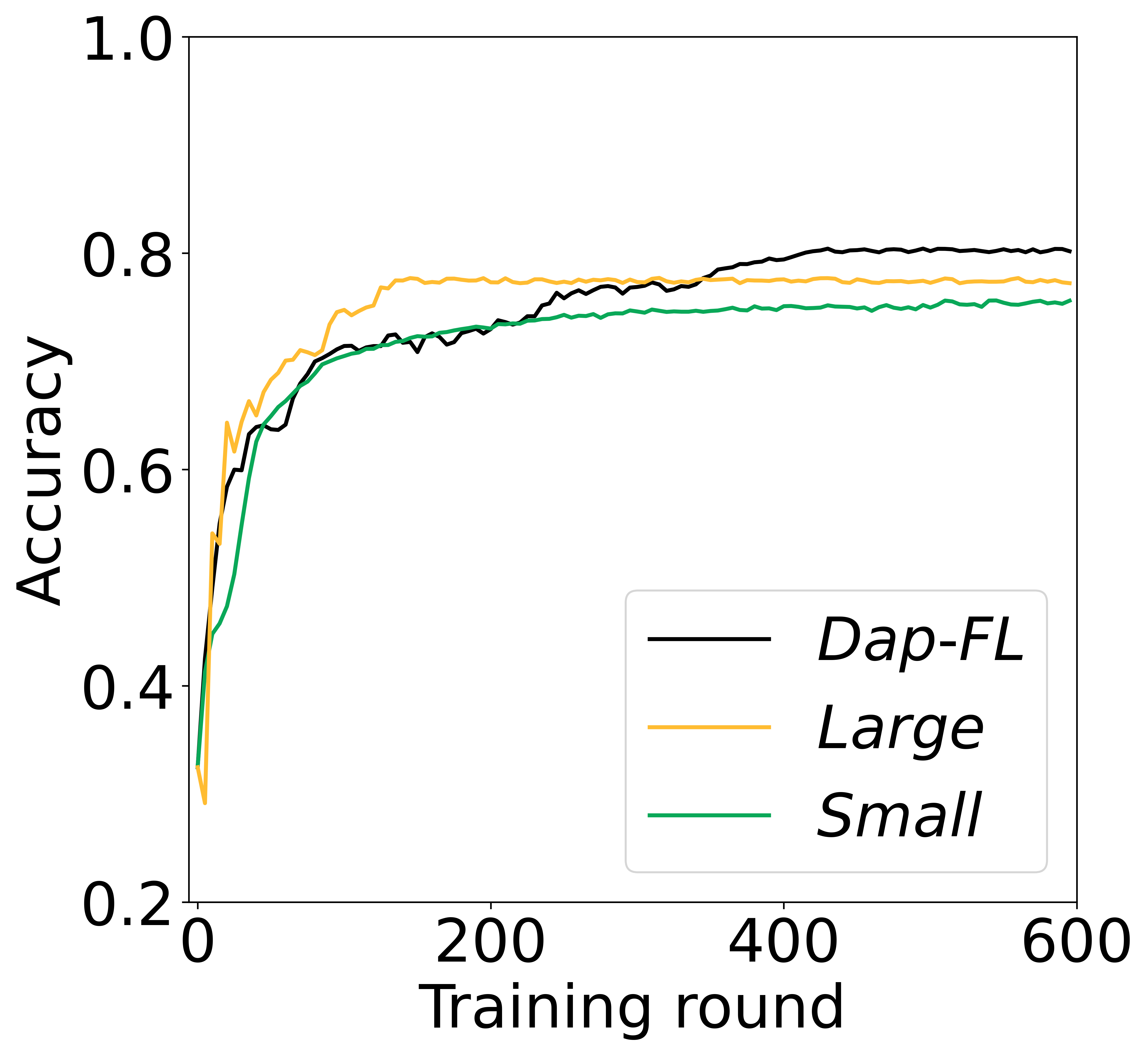}
		\end{minipage}%
	}%
	\centering
	\caption{Performance of global models in the {\em Large}, {\em Small}, and {\em Dap-FL} settings.} 
	\label{fig-7}	
\end{figure*}	

\begin{figure*}[htbp]
	\centering
	\subfigure[Loss of Logistic on MNIST.]{
	\begin{minipage}[t]{0.25\linewidth}
		\centering
		\vspace{-0.5mm}
		\includegraphics[scale=0.35, trim= 5 0 0 0,clip]{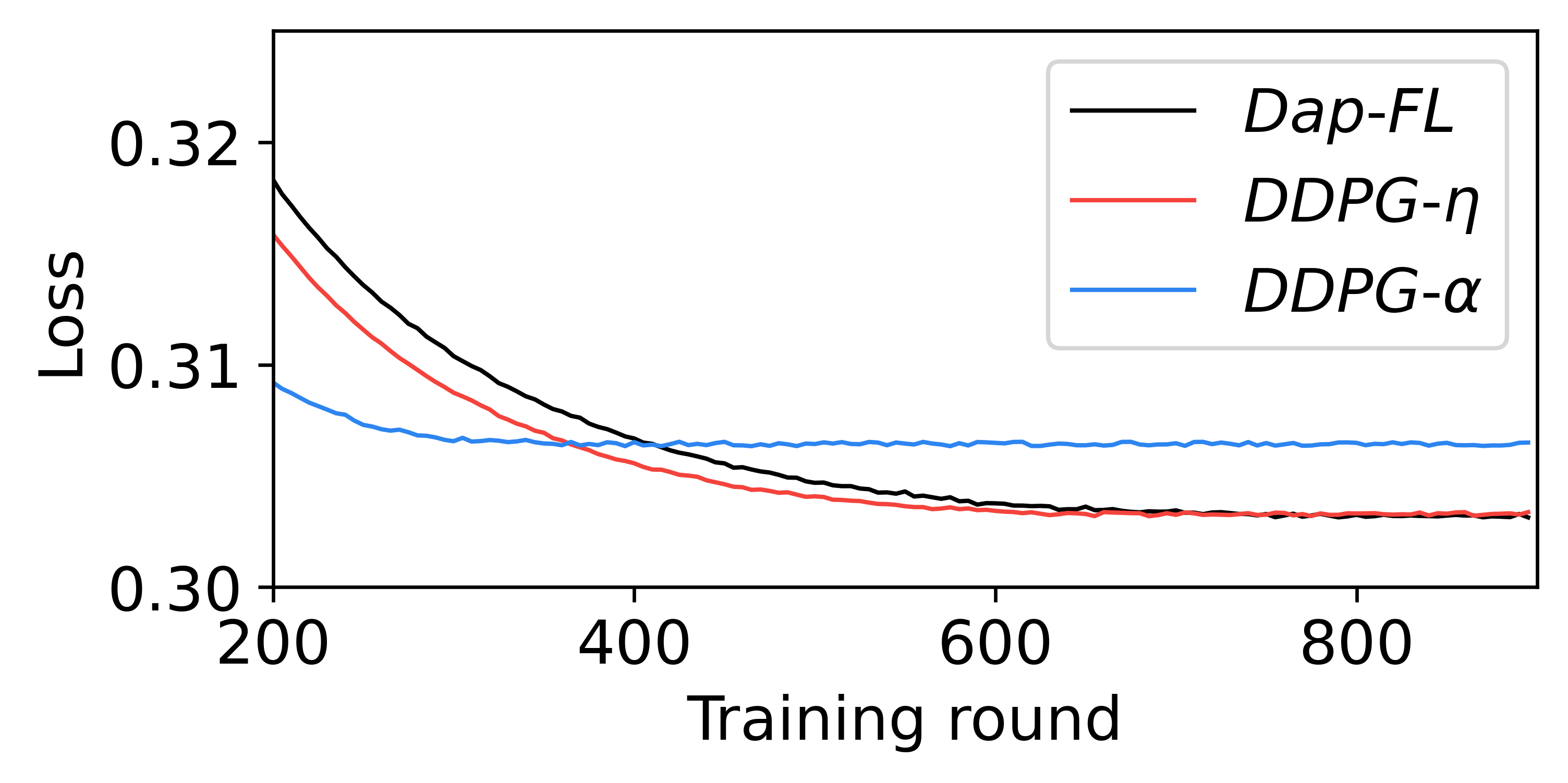}
	\end{minipage}%
    }%
	\subfigure[Loss of CNN on MNIST.]{
		\begin{minipage}[t]{0.25\linewidth}
			\centering
			\vspace{-0.5mm}
			\includegraphics[scale=0.35, trim= 5 0 0 0,clip]{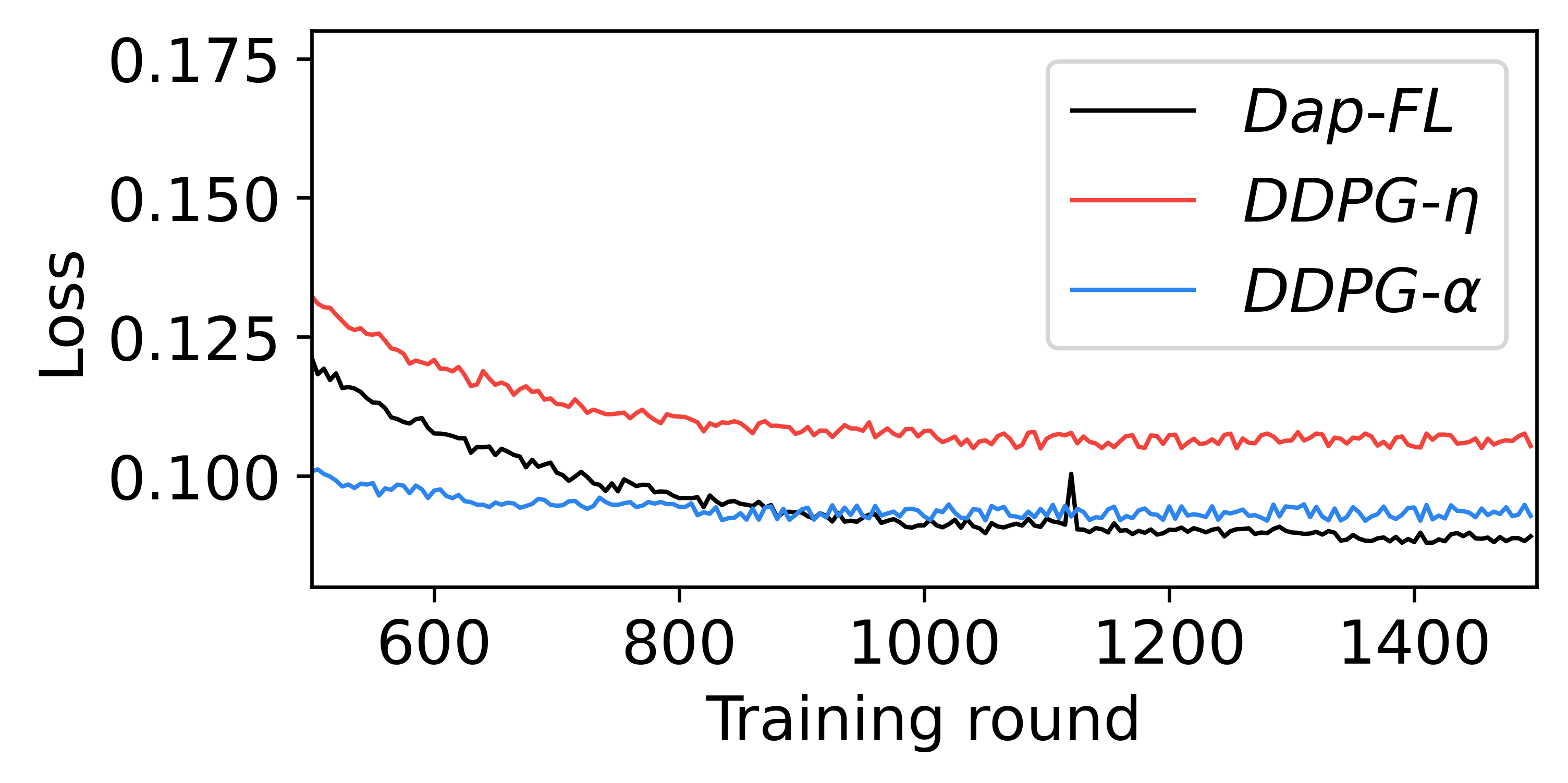}
		\end{minipage}%
	}%
	\subfigure[Loss of ResNet-18 on MNIST.]{
	\begin{minipage}[t]{0.25\linewidth}
		\centering
		\vspace{-0.5mm}
		\includegraphics[scale=0.35, trim= 5 0 0 0,clip]{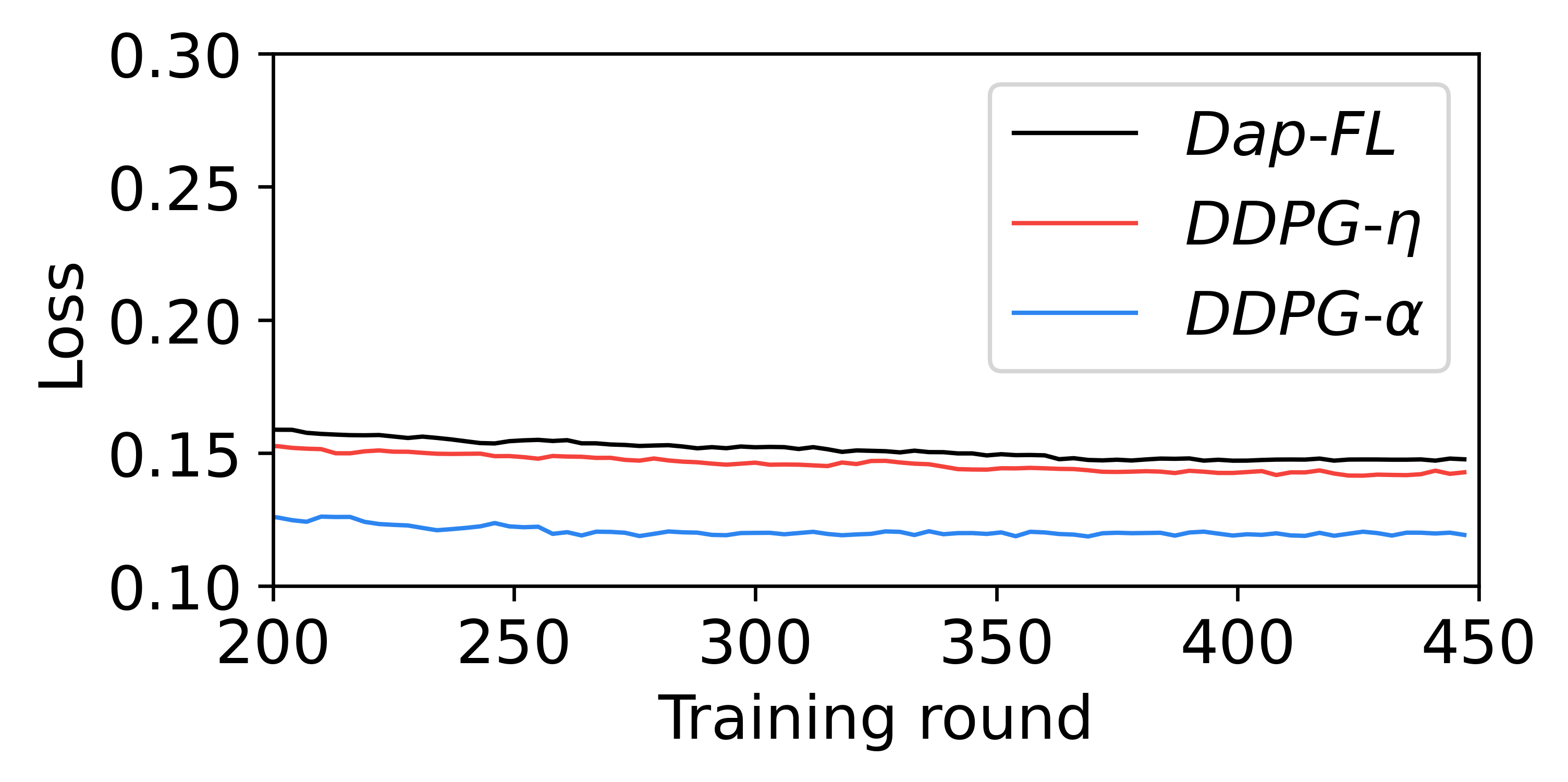}
	\end{minipage}%
    }%
	\subfigure[Loss of CNN on Fashion-MNIST.]{
		\begin{minipage}[t]{0.25\linewidth}
			\centering
			\vspace{0.5mm}
			\includegraphics[scale=0.35, trim= 5 0 0 0,clip]{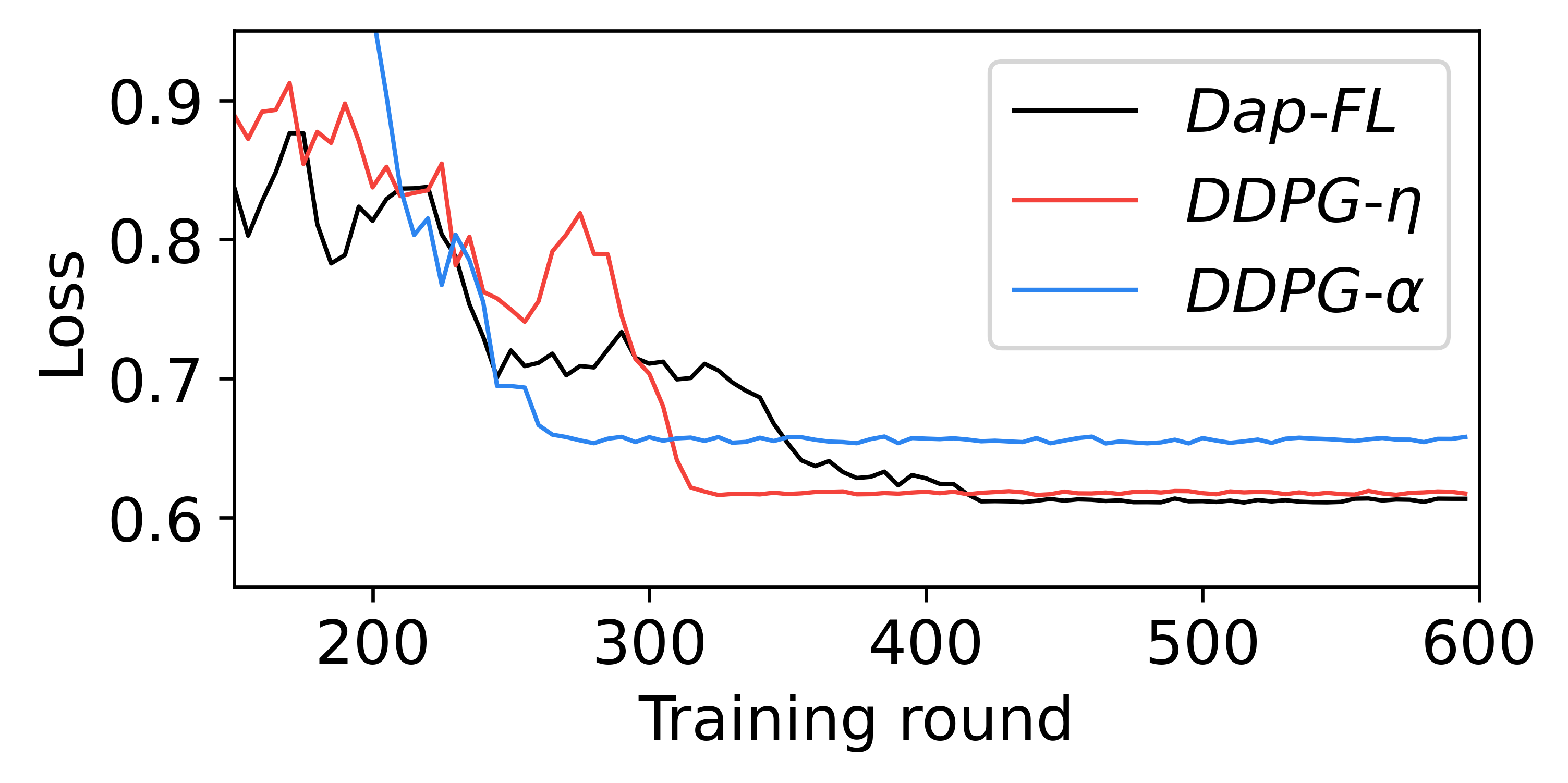}
		\end{minipage}%
	}%
	\vfill
	\subfigure[Accuracy of Logistic on MNIST.]{
	\begin{minipage}[t]{0.25\linewidth}
		\centering
		\vspace{-0.5mm}
		\includegraphics[scale=0.35, trim= 5 0 0 0,clip]{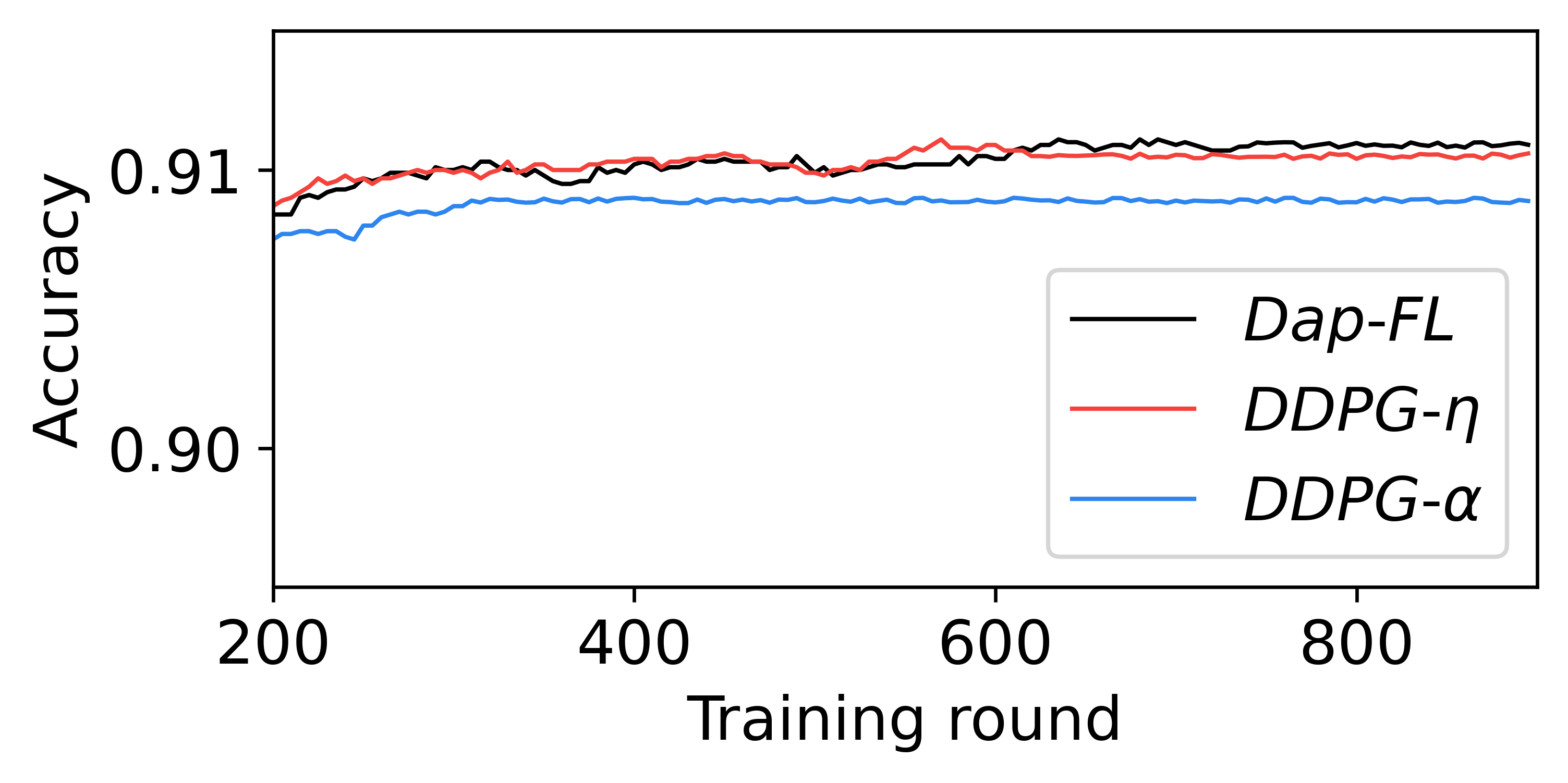}
	\end{minipage}%
    }%
	\subfigure[Accuracy of CNN on MNIST.]{
		\begin{minipage}[t]{0.25\linewidth}
			\centering
			\vspace{-0.5mm}
			\includegraphics[scale=0.35, trim= 5 0 0 0,clip]{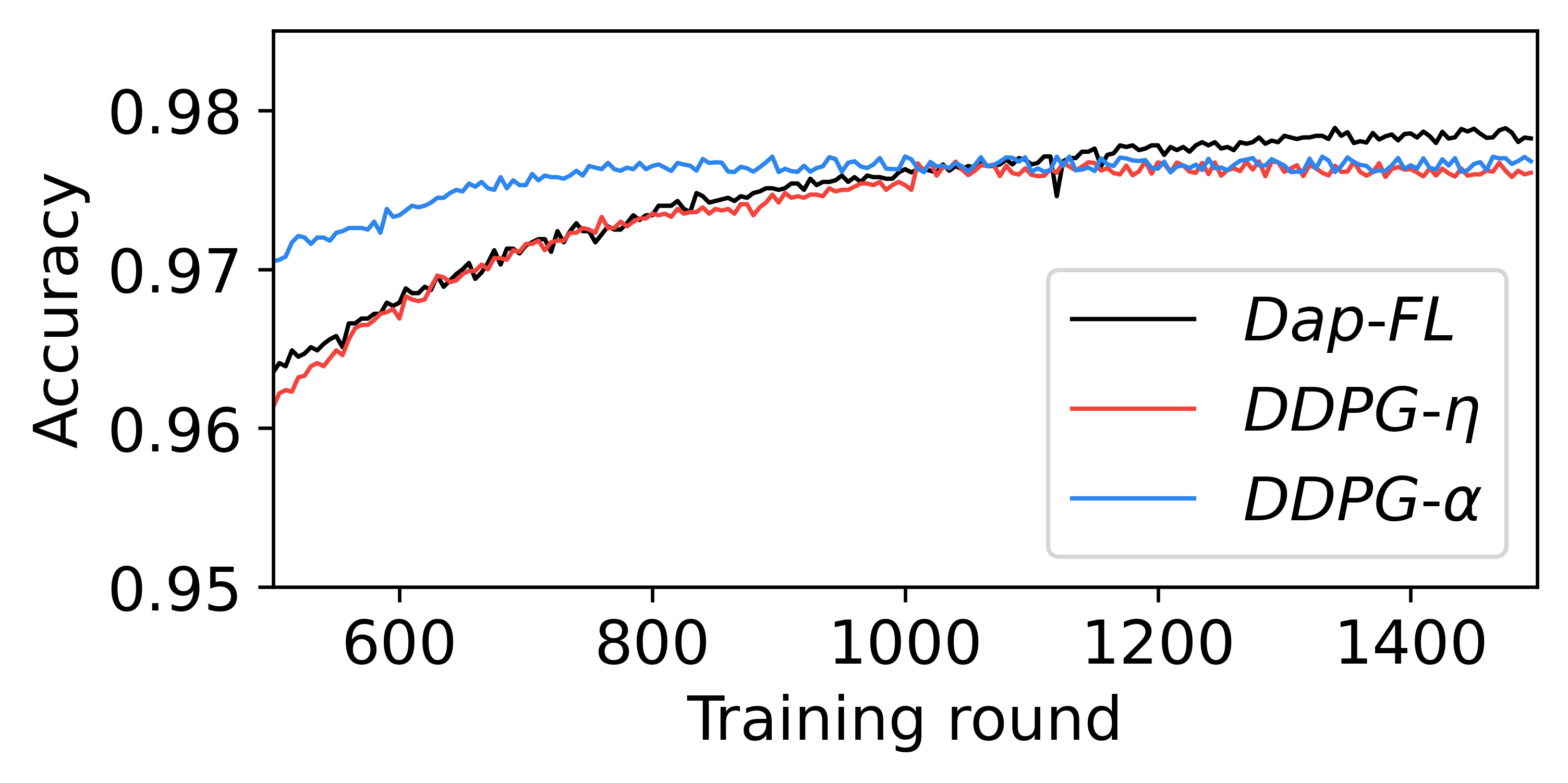}
		\end{minipage}%
	}%
	\subfigure[Accuracy of ResNet-18 on MNIST.]{
	\begin{minipage}[t]{0.25\linewidth}
		\centering
		\vspace{-0.5mm}
		\includegraphics[scale=0.35, trim= 5 0 0 0,clip]{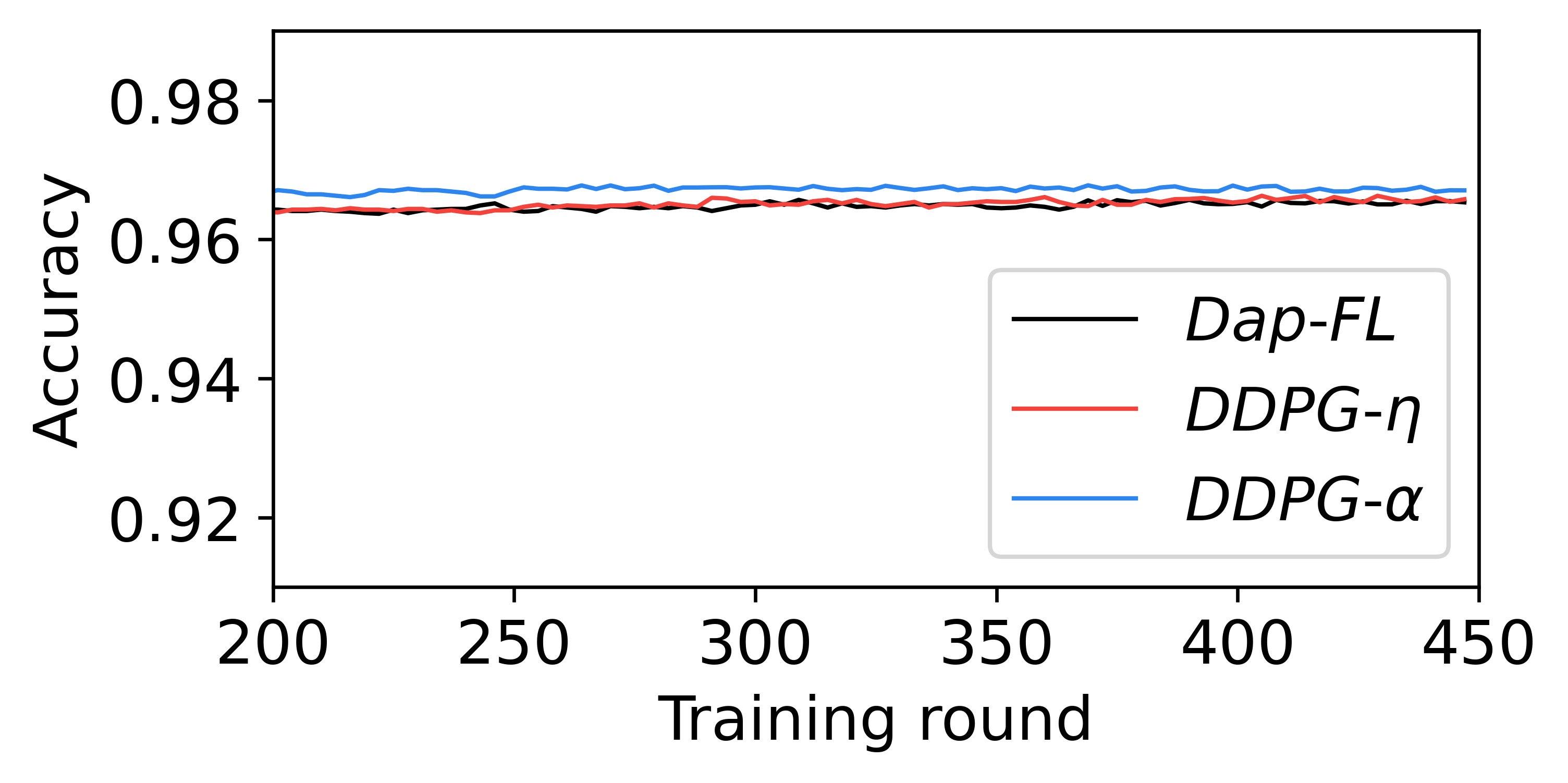}
	\end{minipage}%
    }%
	\subfigure[\!\!\!\!Accuracy\;of\;CNN\;on\;Fashion-MNIST.]{
		\begin{minipage}[t]{0.25\linewidth}
			\centering
			\vspace{-0.5mm}
			\includegraphics[scale=0.35, trim= 5 0 0 0,clip]{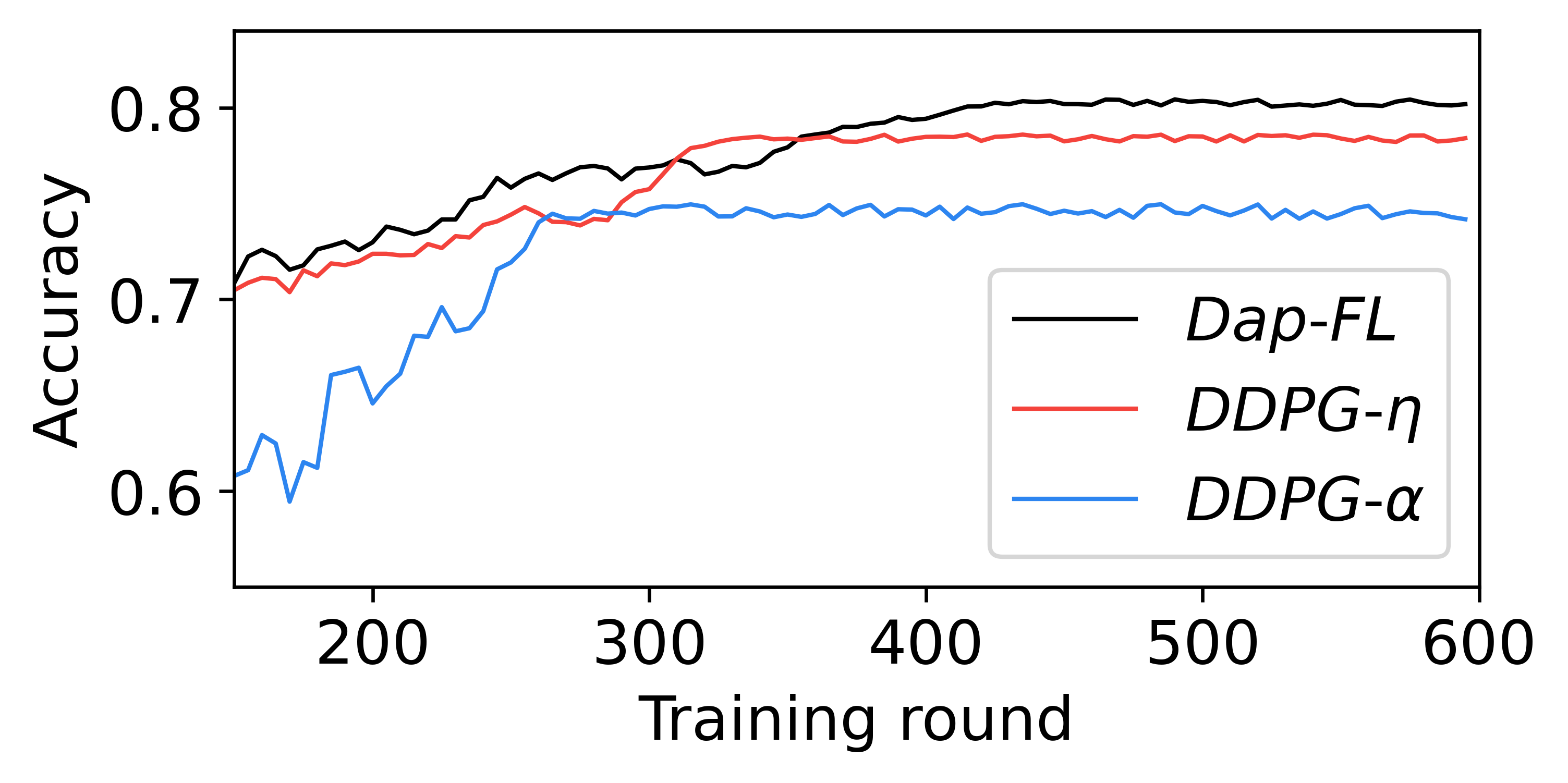}
		\end{minipage}%
	}%
	\centering
	\caption{Performance of global models in the {\em DDPG-$ \eta $}, {\em DDPG-$ \alpha $}, and {\em Dap-FL} settings.} 
	\label{fig-8}	
\end{figure*}

Furthermore, we evaluate the comprehensiveness of the selected hyper-parameters in the proposed Dap-FL system by fixing one type of hyper-parameter and adaptively adjusting the other one, i.e., the {\em DDPG-$ \eta $} and {\em DDPG-$ \alpha $} settings. As shown in Fig. \ref{fig-8}., the prediction accuracy of the global models trained by adaptively adjusting two training hyper-parameters are higher than that trained by adjusting one type of training hyper-parameters for all considered ML tasks except for ResNet-18 on MNIST, and the convergence rates of the global models in the {\em Dap-FL} setting are higher than that in the {\em DDPG-$ \eta $} and {\em DDPG-$ \alpha $} settings for all considered ML tasks except for ResNet-18 on MNIST. Such experimental results illustrate that adaptively adjusting two selected training hyper-parameters in the proposed Dap-FL system is comprehensive in most instances, since merely adjusting one type of local training hyper-parameters cannot obtain global models with high performance.
A possible reason behind the abnormal performance in the ResNet-18 on MNIST task is that the powerful training ability of ResNet-18 benefiting from the residual structure obscures the tuning effect of the proposed Dap-FL system to the global model on such a simple dataset. If changing the dataset to a complex one, such as CIFAR-100, a more significant effect might be present, which remains in future works.

More importantly, to better understand how the proposed adaptive FL system outperforms other state-of-the-art RL-based adaptive FL schemes, we plot the loss and accuracy curves of the global models of the CNN on Fashion-MNIST task in the {\em Dap-FL}, and {\em DDPG-client}, and {\em DQN} settings in Fig. \ref{fig-9}. First of all, on account that Sun {\em et al}.\cite{Sun2021} only adaptively adjust local training epoch and do not give the learning rate in their experiments, we divide the {\em DQN} setting into two sub-settings by fixing the local learning rate as a large value and a small value, i.e., $ 10^{-3} $ and $ 10^{-4} $. As can be observed in Fig. 9.(a), the convergence rate of the global model in the {\em Dap-FL} setting is faster than that with a small learning rate and slower than that with a large learning rate in the {\em DQN} setting, while the final converged loss value of the global model in the {\em Dap-FL} setting is smaller than that in the {\em DQN} setting no matter what the learning rate is. Meanwhile, as shown in Fig. \ref{fig-9}.(b), the accuracy of the final global model in the {\em Dap-FL} setting is $ 80.25\% $, which is $ 7.85\% $ higher than that ($ 72.40\% $) in the {\em DQN} setting with a large learning rate and $ 2.65\% $ higher than that ($ 77.60\% $) in the {\em DQN} setting with a small learning rate. In other words, the proposed Dap-FL is better than Sun {\em et al}.' DQN-based adaptive FL method\cite{Sun2021} in terms of global model convergence rate and accuracy. 

In addition, we can observe in Fig. \ref{fig-9}.(a) that the loss value curve of the global model in the {\em Dap-FL} setting converges faster to a lower value than that in the {\em DDPG-client} setting, and the oscillations of the former one are smaller than the latter one. Besides, the accuracy of final global model in the {\em Dap-FL} setting is $ 6.03\% $ higher than that ($ 74.22\% $) in the {\em DDPG-client} setting. The reason behind such performance gaps is that the global model cannot learn scattered features across local data faster and better through only a fraction of clients' local contributions in each aggregation. In summary, compared to the DDPG client selection scheme proposed by Zhang {\em et al}. \cite{Zhang2021}, i.e., the {\em DDPG-FL} setting in our experiments, the proposed Dap-FL system has better performance in terms of smaller oscillations of loss value curve, faster convergence rate, and higher final global model prediction accuracy.

\begin{figure}[htbp]
	\centering
	\subfigure[CNN on Fashion-MNIST/Loss.]{
		\begin{minipage}[t]{0.5\linewidth}
			\centering
			\vspace{-4.5mm}
			\includegraphics[scale=0.14, trim= 5 0 5 0,clip]{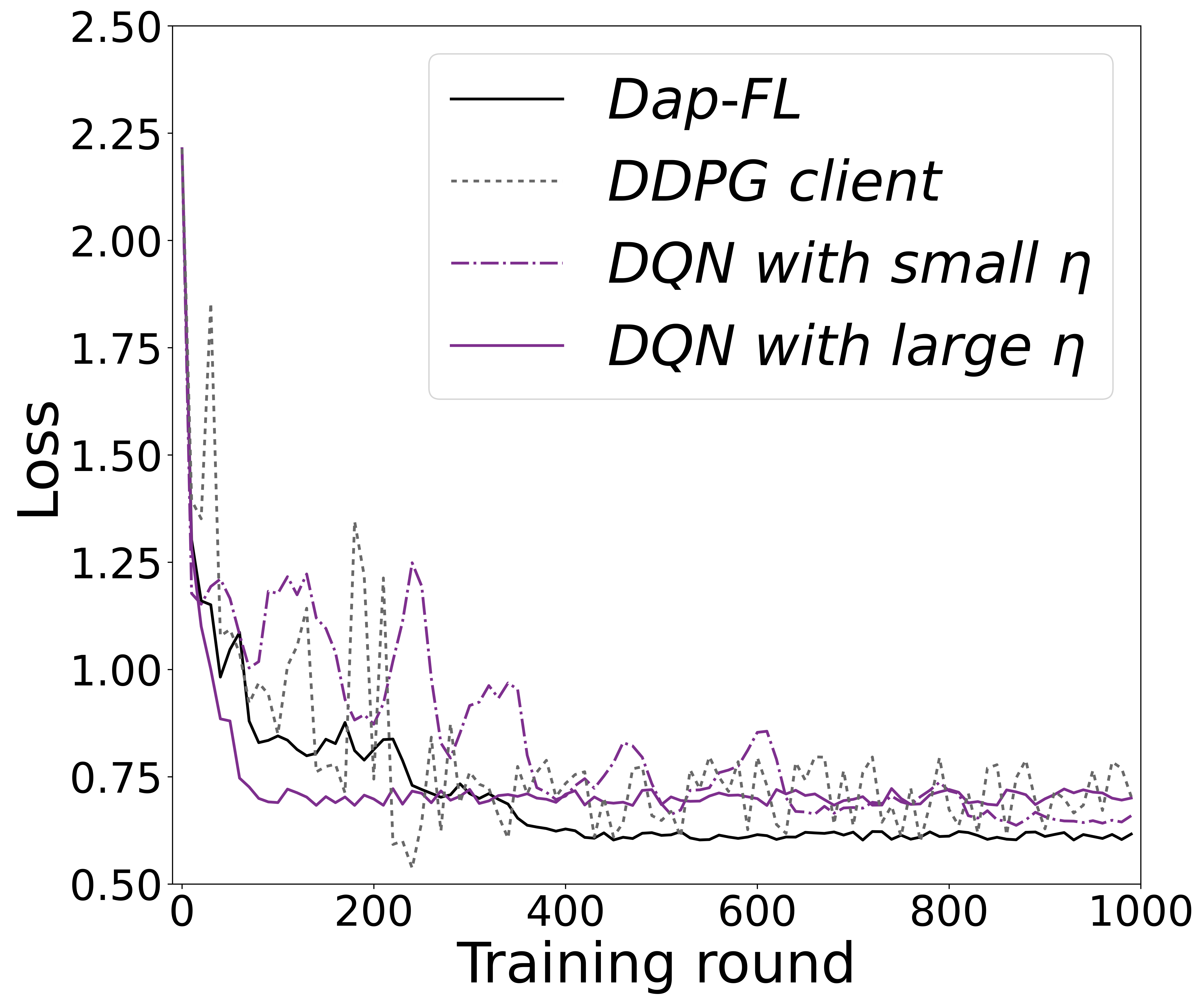}
		\end{minipage}
	}%
	\subfigure[CNN on Fashion-MNIST/Accuracy.]{
		\begin{minipage}[t]{0.5\linewidth}
			\centering
			\vspace{-4.5mm}
			\includegraphics[scale=0.14, trim= 5 0 5 0,clip]{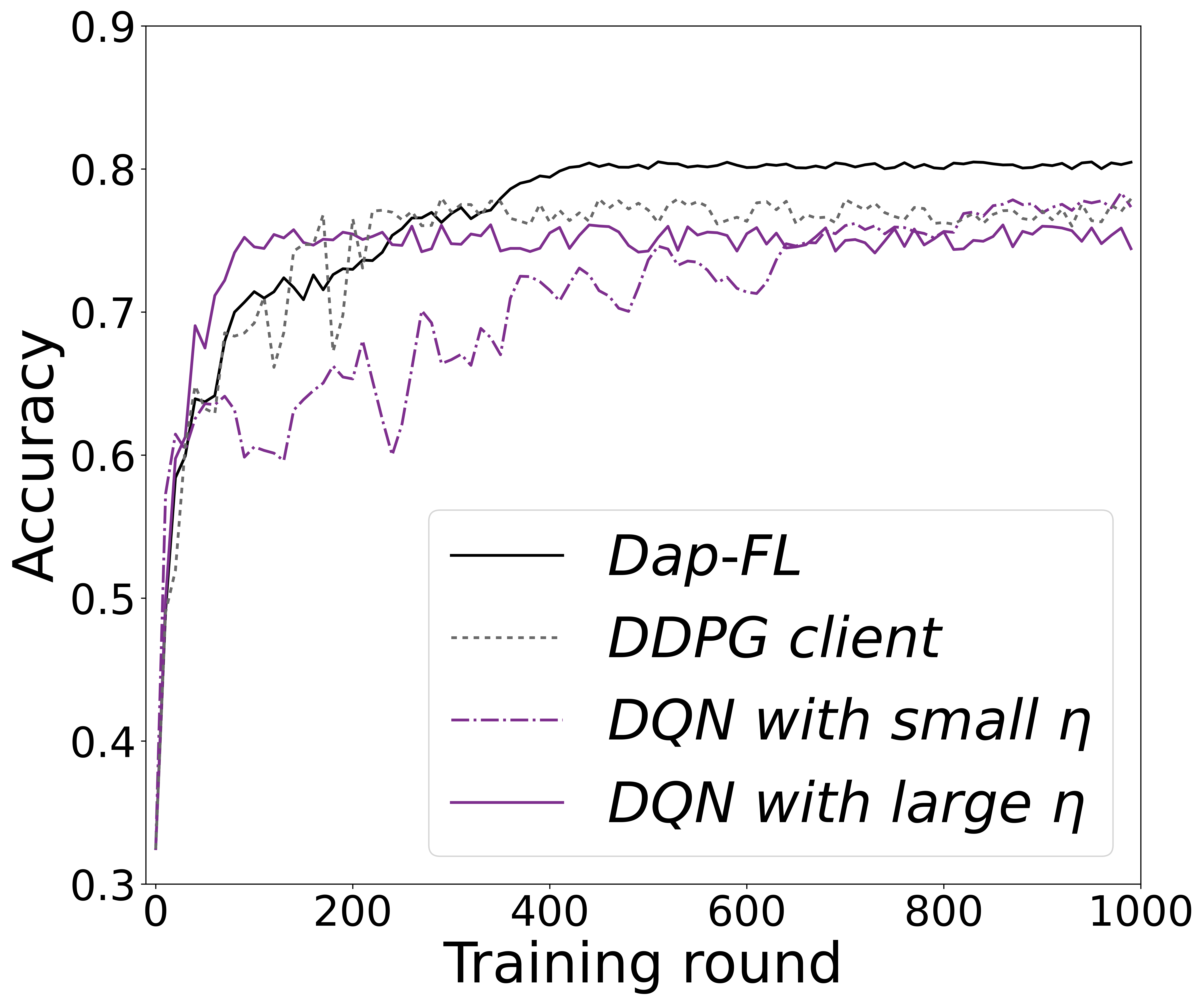}
		\end{minipage}%
	}%
	\centering
	\vspace{-1.5mm}
	\caption{Model performance in the {\em Dap-FL}, {\em DQN}, and {\em DDPG-client} settings.}
	\label{fig-9}
	
\end{figure}

%
%
%
%
%

\section{Related Work}

\noindent{\bfseries Conventional FL }

FL is proposed by Mcmahan {\em et al.} \cite{McMahan2017} aiming to train a global model from decentralized data distributed in different clients.  Then, Google's researchers further develop the FL system to improve communication efficiency \cite{Konecny2016}, system scalability \cite{Bonawitz2019}, and privacy \cite{Bonawitz2017}. Recently, other works build on top of FL by researching different paradigms \cite{Chen2022} and applications \cite{Brisimi2018,Sun2021,Zhang2021}.

\noindent{\bfseries Adaptive FL }

Since the foundation of FL is ML (or DL), adaptively selecting training hyper-parameters, i.e., adaptive FL, is of crucial importance to the flourish of FL. Existing designs for adaptive FL mainly rely on two major categories: theoretic methods and RL-based methods.
 
\textit{\textbf{Adaptive FL based on the theoretic method}.}
Theoretic methods always achieve adaptive FL by formulating the FL process as an optimization problem. 
For example, Luo {\it et al.} \cite{Luo2021} analyze the relationship between the convergence of the global model and the total cost, based on which a biconvex optimization problem with respect to the numbers of local training iterations and client selection is established. The proposed adaptive FL minimize the total cost of learning time and energy consumption while ensuring convergence.
Wu {\it et al.} \cite{Wu2021} propose an adaptive weighting algorithm, FedAdp, to accelerate the global model convergence by quantifying participants' aggregation weights adaptively and jointly.
Wang {\it et al.} \cite{Wang2019} solve the problem that efficiently utilizing the limited computation and communication resources by dynamically adapting the frequency of aggregation.
Shi {\it et al.} \cite{Shi2020} design an efficient binary search algorithm to obtain the scheduling policy in terms of the highest achievable global model accuracy under a given training time budget. 
Tran {\it et al.} \cite{Tran2019} formulate the FL problem as a non-convex optimization problem FEDL that captures the trade-off among computation and communication latencies, clients' energy consumptions, and FL time. By transforming it into three convex sub-problems, they obtain qualitative insights about optimal model accuracy, energy consumption, and learning time.

However, theoretically optimizing the global model accuracy and convergence rate within time-vary constraints, e.g., variable local computational consumption caused by battery power, is less efficiency. Although they make some prior hypotheses to constraints, it seems not that reasonable, as the practical time-varying constraints usually have no statistical regularity to follow.

\textit{\textbf{Adaptive FL based on the RL method}.} An alternative method to achieve adaptive FL points to RL, as the FL process could be formulated as an MDP.
Wang {\it et al.} \cite{Wang2020} propose an experience-driven FL framework FAVOR, which speeds up convergence by intelligence selecting clients. Particularly, Deep Q-Network (DQN) is introduced to maximize a reward that encourages the increase of validation accuracy and penalizes the use of more communication rounds. 
Sun {\it et al.} \cite{Sun2021} adaptively adjust the aggregation frequency of FL based on DQN in IIoT, which is characterized by improved learning accuracy, convergence, and energy saving. 
Su {\it et al.} \cite{Su2022} deploy DQN on clients to derive their optimal training strategies of local models.
Nguyen {\it et al.} \cite{Nguyen2020} propose a double DQN-based scheme to select the optimal communication channel to reduce energy cost during model transmission.
It seems that DQN is an effective method to adaptively adjust local training strategies. 
However, DQN is more proper for solving decision-making problems with discrete states and actions, while the state and action space of FL under our consideration are continuous, e.g., the local learning rate.

Two very recent works introduce DDPG, another DRL method suitable for continuous problems, to achieve adaptive FL. 
On the one hand, Zhang {\it et al.} \cite{Zhang2021} adopt DDPG to select clients with low training costs and high model accuracy to improve the rate of model aggregation and reduce the communication cost for IIoT. 
On the other hand, Lu {\it et al.} \cite{Lu2020} leverage DDPG to improve the efficiency of FL, which can select clients with a larger amount of resources of computing and communication capacity.
Nevertheless, selecting a sub-set of participant clients for faster convergence time runs counter to the fact that more clients increase the convergence rate.
In addition, rejecting straggler clients all the time may raise a fairness issue.

Different from all existing work on adaptive FL, our scheme deploys the DDPG-assisted hyper-parameter selection scheme on every client to adaptively adjust training hyper-parameters locally for the purpose of involving more clients rather than abandoning straggler clients. In this way, all participants could adaptively execute local training according to their time-varying training state by themselves, and the convergence rate of the global model is accelerated. Furthermore, the employment of the Paillier cryptosystem makes the adaptive FL system more privacy-preserving and secure.

\section{Conclusion and Future Work}

Dap-FL is a DDPG-assisted adaptive and privacy-preserving FL system, which guarantees clients with poor resources could participate in FL by adaptively adjusting local training hyper-parameters, and preserves model privacy through a secure aggregation method based on the Paillier cryptosystem. 

An interesting experimental result is that the global model prediction accuracy of ResNet-18 on MNIST is not improved with the deployment of the proposed Dap-FL. We guess the possible reason is that the powerful model training capability of the ResNet-18 overshadows the hyper-parameter tuning effect of the proposed Dap-FL, which should be further confirmed in the future. Besides, we plan to explore more types of hyper-parameters in FL, and generalize them into Dap-FL, which also remains in future works. 

\appendices

\ifCLASSOPTIONcaptionsoff
\newpage
\fi
	

	

	
\end{document}